\PassOptionsToPackage{names,dvipsnames}{xcolor} %
\documentclass[conference]{IEEEtran}

\usepackage[english]{babel}
\usepackage[inference]{semantic}
\usepackage{amsmath}\allowdisplaybreaks
\usepackage{mathpartir}
\usepackage{amsthm}
\usepackage{amssymb}
\usepackage{stmaryrd} %
\usepackage{xspace}
\usepackage{latexsym}
\usepackage{hyperref}
\usepackage{ifthen}
\usepackage{mathtools}
\usepackage{color}
\usepackage{listings}
\usepackage{verbatim}
\usepackage[colorinlistoftodos]{todonotes}
\usepackage{tikz}
\usepackage{floatrow}
\usetikzlibrary{positioning,shadows,arrows,calc,backgrounds,fit,shapes,snakes,shapes.multipart,decorations.pathreplacing,shapes.misc,patterns}
\usepackage[T1]{fontenc}
\usepackage[scaled=.83]{beramono}
\usepackage{epigraph}
\usepackage{cleveref}
\usepackage{booktabs}
\usepackage{float}
\usepackage{etoolbox}
\usepackage{wrapfig}
\usepackage{pgfplots}
\usepackage{nameref}
\usepackage{scalerel}
\usepackage{bm}
\usepackage{centernot}
\usepackage{mdframed}
\usepackage{bussproofs}
\usepackage[sort,square,numbers]{natbib}
\usepackage{dashrule}
\usepackage{paralist}

\hypersetup{ pdfpagemode=UseOutlines, colorlinks=true, linkcolor=ForestGreen, citecolor=ForestGreen }

\newcommand{\mi}[1]{\ensuremath{\mathit{#1}}}

\newcommand{\mtt}[1]{\ensuremath{\mathtt{#1}}}
\newcommand{\mf}[1]{\ensuremath{\mathbf{#1}}}

\newcommand{\mc}[1]{\ensuremath{\mathcal{#1}}}
\newcommand{\ms}[1]{\ensuremath{\mathsf{#1}}}
\newcommand{\mb}[1]{\ensuremath{\mathbb{#1}}}

\newcommand{\isdef}[0]{\ensuremath{\mathrel{\overset{\makebox[0pt]{\mbox{\normalfont\tiny\sffamily def}}}{=}}}}

\newcommand{\relmiddle}[1]{\mathrel{}\middle#1\mathrel{}}

\newcommand\bnfdef{\ensuremath{\mathrel{::=}}}

\newcommand{\OB}[1]{\ensuremath{\overline{#1}}}

\newcommand{\myset}[2]{\ensuremath{\left\{#1 ~\relmiddle|~ #2\right\}}}

\newcommand*{\QEDA}{\hfill\ensuremath{\blacksquare}}%

\AtEndEnvironment{problem}{\null\hfill\QEDA}
\AtEndEnvironment{example}{}%

\Crefname{lstlisting}{Listing}{Listings}
\Crefname{problem}{Problem}{Problems}

\Crefname{equation}{Rule}{Rules}

\newcommand{\funname}[1]{\mtt{#1}}
\newcommand{\fun}[2]{\ensuremath{{\bl{\funname{#1}\left(#2\right)}}}\xspace}

\newcommand{\domG}[1]{\fun{domG}{#1}}
\newcommand{\domM}[1]{\fun{domM}{#1}}
\newcommand{\invcond}[1]{\fun{cond}{#1}}

\newcommand{\contextletter}[0]{A}
\newcommand{\ctx}[1]{\ensuremath{\contextletter}} %

\newcommand{\come}[0]{\com{\emptyset}\xspace}

\newcommand{\SInit}[1]{\ensuremath{{\Omega_0}\left({#1}\right)}\xspace}

\newcommand{\CInit}[1]{\ensuremath{{\sigma_0}\left({#1}\right)}\xspace}

\newcommand{\neutcol}[0]{black}
\newcommand{\stlccol}[0]{RoyalBlue}
\newcommand{\ulccol}[0]{RedOrange}

\newcommand{\commoncol}[0]{black}    %

\newcommand{\col}[2]{\ensuremath{{\color{#1}{#2}}}}

\newcommand{\bl}[1]{\col{\neutcol }{#1}}
\newcommand{\com}[1]{\mi{\col{\commoncol }{#1}}}

\newcommand{\myagree}[0]{\mathop{\raisebox{1mm}{$\frown$}}}
\newcommand{\agree}[2]{\ensuremath{#1\myagree#2}}

\newcounter{typerule}
\crefname{typerule}{rule}{rules}

\newcommand{\typeruleInt}[5]{%
	\def\thetyperule{#1}%
	\refstepcounter{typerule}%
	\label{tr:#4}%
  \ensuremath{\begin{array}{c}#5 \inference{#2}{#3}\end{array}}
}
\newcommand{\typerule}[4]{%
  \typeruleInt{#1}{#2}{#3}{#4}{\textsf{\scriptsize ({#1})} \\      }
}

\pgfdeclarelayer{background}
\pgfdeclarelayer{veryback}
\pgfdeclarelayer{veryback2}
\pgfdeclarelayer{veryback3}
\pgfdeclarelayer{back2}
\pgfdeclarelayer{foreground}
\pgfsetlayers{veryback3,veryback2,veryback,background,back2,main,foreground}

\newcommand{\myfig}[3]{\begin{figure} [!ht]
#1
\caption{\label{fig:#2}#3}
\end{figure}}

\newcommand{\BREAK}[0]{
\botrule
\begin{center}$\spadesuit$\end{center}
\botrule}

\def\botrule{\vspace{0mm}\hrule\vspace{2mm}}

\newcounter{line}

\newcommand{\asm}[1]{\mtt{#1}}

\newcommand{\xto}[1]{\ensuremath{~\mathrel{\xrightarrow{~#1~}}~}}
\newcommand{\Xto}[1]{\ensuremath{~\mathrel{\xRightarrow{~#1~}}~}}

\newcommand{\Xtol}[1]{\ensuremath{\xRightarrow{~#1~}\Low}}

\newcommand{\Low}[0]{\ensuremath{\!\!\!\!\Rightarrow} }

\definecolor{mygreen}{rgb}{0,0.6,0}
\definecolor{mygray}{rgb}{0.5,0.5,0.5}
\definecolor{mymauve}{rgb}{0.58,0,0.82}

\lstdefinelanguage{Java} %
{morekeywords={abstract, all, and, as, assert, but, disj, else, exactly, extends, fact, for, fun, iden, if, iff, implies, in, Int, void, int, let, lone, module, no, none, not, one, open, or, part, pred, run, seq, set, sig, some, sum, then, univ, package, class, public, private, null, return, new, interface, extern, object, implements, System, static, super, try , catch, throw, throws, Unit, var, val, of, principal, trust},
sensitive=true,
keywordstyle=\bfseries\color{\stlccol}, %
commentstyle=\itshape\color{purple!40!black},
morecomment=[l][\small\itshape\color{purple!40!black}]{//},
identifierstyle=\color{\stlccol},
stringstyle=\color{orange},
basicstyle=\small,
basicstyle={\small\ttfamily},
numbers=left,
numberstyle=\tiny\color{mygray},
tabsize=2,
numbersep=5pt,
breaklines=true,
lineskip=-2pt,
stepnumber=1,
captionpos=b,
breaklines=true,
breakatwhitespace=false,
showspaces=false,
showtabs=false,
float=!h,
columns=fullflexible,escapeinside={(*@}{@*)},
moredelim=**[is][\color{red!60}]{@}{@},
literate={->}{{$\to$}}1 {^}{{$\mspace{-3mu}\widehat{\quad}\mspace{-3mu}$}}1
{<}{$<$ }2 {>}{$>$ }2 {>=}{$\geq$ }2 {=<}{$\leq$ }2
{<:}{{$<\mspace{-3mu}:$}}2 {:>}{{$:\mspace{-3mu}>$}}2
{=>}{{$\Rightarrow$ }}2 {+}{$+$ }2 {++}{{$+\mspace{-8mu}+$ }}2
{<=>}{{$\Leftrightarrow$ }}2 {+}{$+$ }2 {++}{{$+\mspace{-8mu}+$ }}2
{\~}{{$\mspace{-3mu}\widetilde{\quad}\mspace{-3mu}$}}1
{!=}{$\neq$ }2 {*}{${}^{\ast}$}1 %
{\#}{$\#$}1
}

\lstdefinelanguage{General} %
{morekeywords={abstract, all, and, as, assert, but, disj, else, exactly, extends, fact, for, fun, iden, if, iff, implies, in, Int, void, int, let, lone, module, no, none, not, one, open, or, part, pred, run, seq, set, sig, some, sum, then, univ, package, class, public, private, null, return, new, interface, extern, object, implements, System, static, super, try , catch, throw, throws, Unit, var, val, of, principal, trust},
sensitive=true,
keywordstyle=\bfseries\color{\neutcol},
commentstyle=\itshape\color{purple!40!black},
morecomment=[l][\small\itshape\color{purple!40!black}]{//},
identifierstyle=\color{\neutcol},
stringstyle=\color{orange},
basicstyle=\small,
basicstyle={\small\ttfamily},
numbers=left,
numberstyle=\tiny\color{mygray},
tabsize=2,
numbersep=5pt,
breaklines=true,
lineskip=-2pt,
stepnumber=1,
captionpos=b,
breaklines=true,
breakatwhitespace=false,
showspaces=false,
showtabs=false,
float=!h,
columns=fullflexible,escapeinside={(*@}{@*)},
moredelim=**[is][\color{red!60}]{@}{@},
literate={->}{{$\to$}}1 {^}{{$\mspace{-3mu}\widehat{\quad}\mspace{-3mu}$}}1
{<}{$<$}2 {>}{$>$}2 {>=}{$\geq$}2 {=<}{$\leq$}2
{<:}{{$<\mspace{-3mu}:$}}2 {:>}{{$:\mspace{-3mu}>$}}2
{=>}{{$\Rightarrow$ }}2 {+}{$+$ }2 {++}{{$+\mspace{-8mu}+$ }}2
{<=>}{{$\Leftrightarrow$ }}2 {+}{$+$ }2 {++}{{$+\mspace{-8mu}+$ }}2
{\~}{{$\mspace{-3mu}\widetilde{\quad}\mspace{-3mu}$}}1
{!=}{$\neq$ }2 {*}{${}^{\ast}$}1 %
{\#}{$\#$}1
}

\lstdefinelanguage{Asm}
{morekeywords={abstract, all, and, as, assert, but, check, disj, else, exactly, extends, fact, for, fun, iden, if, iff, implies, in, Int, void, int, let, lone, module, no, none, not, one, open, part, pred, run, seq, set, sig, some, sum, then, univ, package, class, public, private, null, return, new, interface, extern, object, implements, System, static, super, try , catch, throw, throws, Unit, var, val, principal, trust, label, load, add, addi, into, test, mov, cmov, cmova, movzx, cmp, jbe, sar, cmovbe, or, jmp, shl, ret, jae, lea, lfence, jne},
sensitive=true,
identifierstyle=\color{\ulccol},
keywordstyle=\bfseries\color{\ulccol},
commentstyle=\itshape\color{purple!40!black},
morecomment=[l][\small\itshape\color{purple!40!black}]{//},
stringstyle=\color{orange},
basicstyle=\small,
basicstyle={\small},
numbers=left,
numberstyle=\tiny\color{mygray},
tabsize=2,
numbersep=5pt,
breaklines=true,
lineskip=-2pt,
stepnumber=1,
captionpos=b,
breaklines=true,
breakatwhitespace=false,
showspaces=false,
showtabs=false,
float=!h,
columns=fullflexible,escapeinside={(*@}{@*)},
moredelim=**[is][\color{red!60}]{@}{@},
literate={->}{{$\to$}}1 {^}{{$\mspace{-3mu}\widehat{\quad}\mspace{-3mu}$}}1
{<}{$<$}2 {>}{$>$ }2 {>=}{$\geq$ }2 {=<}{$\leq$ }2
{<:}{{$<\mspace{-3mu}:$}}2 {:>}{{$:\mspace{-3mu}>$}}2
{=>}{{$\Rightarrow$ }}2 {+}{$+$ }2 {++}{{$+\mspace{-8mu}+$ }}2
{<=>}{{$\Leftrightarrow$ }}2 {+}{$+$ }2 {++}{{$+\mspace{-8mu}+$ }}2
{\~}{{$\mspace{-3mu}\widetilde{\quad}\mspace{-3mu}$}}1
{!=}{$\neq$ }2 {*}{${}^{\ast}$}1 %
{\#}{$\#$}1
}
\DeclareMathOperator\ceq{\ensuremath{\mathrel{\simeq_{\mi{ctx}}}}}

\def\teqaux#1{\vcenter{\hbox{\ooalign{\hfil
       \raise6pt \hbox{\scriptsize{T}}\hfil\cr\hfil
       $=$}}}}

\def\relssaux#1{\vcenter{\hbox{\ooalign{\hfil
       \raise6pt \hbox{\tiny{$\boldsymbol{+}$}}\hfil\cr\hfil
       $\approx$}}}}

\def\hrelssaux#1{\vcenter{\hbox{\ooalign{\hfil
       \raise6pt \hbox{\tiny{\com{H}}}\hfil\cr\hfil
       $\approx$}}}}

\def\vrelssaux#1{\vcenter{\hbox{\ooalign{\hfil
       \raise6pt \hbox{\tiny{\com{V}}}\hfil\cr\hfil
       $\approx$}}}}

\def\srelssaux#1{\vcenter{\hbox{\ooalign{\hfil
       \raise6pt \hbox{\tiny{\com{S}}}\hfil\cr\hfil
       $\approx$}}}}

\def\brelssaux#1{\vcenter{\hbox{\ooalign{\hfil
       \raise6pt \hbox{\tiny{\com{B}}}\hfil\cr\hfil
       $\approx$}}}}

\def\crelssaux#1{\vcenter{\hbox{\ooalign{\hfil
       \raise6pt \hbox{\tiny{\com{C}}}\hfil\cr\hfil
       $\approx$}}}}

\def\arelssaux#1{\vcenter{\hbox{\ooalign{\hfil
       \raise6pt \hbox{\tiny{\com{A}}}\hfil\cr\hfil
       $\approx$}}}}

\newcommand{\relslh}[0]{\Bumpeq} %

\def\shrelssaux#1{\vcenter{\hbox{\ooalign{\hfil
       \raise6pt \hbox{\tiny{\com{H}}}\hfil\cr\hfil
       $\relslh$}}}}

\def\svrelssaux#1{\vcenter{\hbox{\ooalign{\hfil
       \raise6pt \hbox{\tiny{\com{V}}}\hfil\cr\hfil
       $\relslh$}}}}

\def\ssrelssaux#1{\vcenter{\hbox{\ooalign{\hfil
       \raise6pt \hbox{\tiny{\com{S}}}\hfil\cr\hfil
       $\relslh$}}}}

\def\sbrelssaux#1{\vcenter{\hbox{\ooalign{\hfil
       \raise6pt \hbox{\tiny{\com{B}}}\hfil\cr\hfil
       $\relslh$}}}}

\def\screlssaux#1{\vcenter{\hbox{\ooalign{\hfil
       \raise6pt \hbox{\tiny{\com{C}}}\hfil\cr\hfil
       $\relslh$}}}}

\def\sarelssaux#1{\vcenter{\hbox{\ooalign{\hfil
       \raise6pt \hbox{\tiny{\com{A}}}\hfil\cr\hfil
       $\relslh$}}}}

\newcommand{\relslhp}[0]{\propto} %

\def\cssrelssaux#1{\vcenter{\hbox{\ooalign{\hfil
       \raise6pt \hbox{\tiny{\com{S}}}\hfil\cr\hfil
       $\relslhp$}}}}

\def\csbrelssaux#1{\vcenter{\hbox{\ooalign{\hfil
       \raise6pt \hbox{\tiny{\com{B}}}\hfil\cr\hfil
       $\relslhp$}}}}

\def\cshrelssaux#1{\vcenter{\hbox{\ooalign{\hfil
       \raise6pt \hbox{\tiny{\com{H}}}\hfil\cr\hfil
       $\relslhp$}}}}

\def\ceqwaux#1{\vcenter{\hbox{\ooalign{\hfil
       \raise6pt \hbox{\scriptsize{w-b}}\hfil\cr\hfil
       $\ceq$}}}}

\def\praux#1{\vcenter{\hbox{\ooalign{\hfil
       \raise4pt \hbox{$\subset$}\hfil\cr\hfil
       $\sim$}}}}

\newcommand{\labelfont}[1]{\ensuremath{\asm{#1}}}
\newcommand{\trcl}[2]{\ensuremath{\labelfont{call}~ #1~ #2{?}}}
\newcommand{\trcb}[2]{\ensuremath{\labelfont{call}~ #1~ #2{!}}}
\newcommand{\trrt}[1]{\ensuremath{\labelfont{ret}~ #1{!}}}
\newcommand{\trrb}[1]{\ensuremath{\labelfont{ret}~ #1{?}}}

\newcommand{\rth}[2]{\ensuremath{\labelfont{ret}{!}}}	 		%
\newcommand{\rbh}[2]{\ensuremath{\labelfont{ret}{?}}}			%

\newcommand{\ac}[0]{\com{\alpha}\xspace}

\newcommand{\card}[1]{\ensuremath{|\!|{#1}|\!|}}

\newcommand{\var}[1]{\lst{var}~#1}

\AtEndEnvironment{example}{\null\hfill$\boxdot$}

\makeatletter
\xdef\@thefnmark{\@empty}

\newcommand{\Thmref}[1]{\Cref{#1}~(\nameref{#1})}
\makeatother

\renewcommand{\emptyset}[0]{\varnothing}

\newcounter{hps}
\crefname{hps}{}{}

\newcommand{\proven}[1]{\ensuremath{\checkmark}}

\Crefname{mydefinition}{Definition}{Definitions}
\crefname{mydefinition}{Definition}{Definitions}
\theoremstyle{definition}

\newcommand{\xleadsto}[1]{%
   \if\relax\detokenize{#1}\relax
   \rightsquigarrow
   \else
   \mathrel{%
     \begin{tikzpicture}[%
       baseline={(current bounding box.south)}
       ]
       \node[%
       ,inner sep=.44ex
       ,align=center
       ] (tmp) {$\scriptstyle #1$};
       \path[%
       ,draw,<-
       ,decorate,decoration={%
         ,zigzag
         ,amplitude=0.7pt
         ,segment length=1.2mm,pre length=3.5pt
       }
       ]
       (tmp.south east) -- (tmp.south west);
     \end{tikzpicture}
   }
   \fi
 }

\DeclareMathOperator\sem{\ensuremath{\rightsquigarrow}}

\newcounter{criteria}
\crefname{criteria}{}{}

\newcounter{proofr}
\crefname{proofr}{}{}

\newcommand{\inv}[0]{\ensuremath{\iota}\xspace}
\newcommand{\modul}[0]{\ensuremath{\Xi}\xspace}
\newcommand{\local}[0]{\ensuremath{\Lambda}\xspace}

\newcommand{\atk}[0]{{A}\xspace}

\newcommand{\typ}[1]{\ensuremath{\widehat{#1}}}

\newcommand{\len}[1]{\fun{len}{#1}}

\newcommand{\rs}[2]{\ensuremath{ \triangleright_{\mi{RS}}\ \com{#1} : \com{#2} }} %

\newcommand{\nil}[0]{[\, ]}
\newcommand{\listsep}[0]{::}

\newcommand{\restr}[2]{#1|_{\mi{#2}}}

\newcommand{\stronginv}[3]{ \com{#1} \vdash \com{#2} \propto \com{#3} : \com{strong} }
\newcommand{\weakinvinv}[3]{ \com{#1} \vdash \com{#2} \propto \com{#3} : \com{local} }
\newcommand{\weakinvatk}[3]{ \com{#1} \vdash \com{#2} \propto \com{#3} : \com{unreachable} }

\DeclareMathOperator\notcap{\ensuremath{\not\!\cap}}

\newcommand{\ea}[0]{\modul_{\mi{imm}}}
\newcommand{\enb}[0]{\modul_{\mi{mut}}}

\newcommand{\totypes}[1]{\fun{absty}{#1}}

\newcommand{\okref}[0]{\textsf{OkRef}}
\newcommand{\inref}[0]{\textsf{InvRef}}
\newcommand{\noref}[0]{\textsf{NonRef}}

\newcommand{\concfun}[1]{\fun{\gamma}{#1}}

\newcommand\xrsquigarrow[1]{%
\mathrel{%
\begin{tikzpicture}[baseline= {( $ (current bounding box.south) + (0,-0.5ex) $ )}]
  \node[inner sep=.5ex] (a) {$\scriptstyle #1$};
  \path[draw,implies-,double distance between line centers=1.5pt,decorate,
    decoration={zigzag,amplitude=0.7pt,segment length=1.2mm,pre=lineto,
    pre   length=4pt}] 
    (a.south east) -- (a.south west);
\end{tikzpicture}}%
}

\newcommand{\encapsulated}[3]{\ensuremath{#2 \vdash_{\mi{enc}} #1 :#3}}
\newcommand{\localed}[3]{\ensuremath{#3 \vdash_{\mi{loc}} #1 :#2}}

\newcommand{\localinv}[2]{\localed{#1}{#2}{\local}} %

\newcommand{\tc}{\text{trusted code}\xspace}

\renewcommand{\var}[0]{x}

\newcommand{\stt}[0]{\ensuremath{\sigma}}

\newcommand{\typingdef}[0]{ \codeenv, \procid, \instr \vdash \ty{L}, \ty{S} \rightarrow \ty{L'}, \ty{S'} }

\newcommand{\tcenv}[0]{\ensuremath{\Omega^{\!\!\dagger}}}

\newcommand{\trc}[0]{\OB{\ac}}

\newcommand{\acteval}[5]{#1 \triangleright #2 \vdash #3 \xto{#4} #5}
\newcommand{\actevalfat}[5]{#1 \triangleright #2 \vdash #3 \Xto{#4} #5}
\newcommand{\treval}[5]{#1 \triangleright #2 \vdash #3 \Xtol{#4} #5}

\newcommand{\wt}[2]{#1\vdash #2 : \com{wt}}

\newcommand{\mandg}[1]{\fun{mg}{#1}}

\newcommand{\actcheck}[2]{ #1 \Vdash #2}
\newcommand{\tracecheck}[2]{ #1 \Vdash #2}

\newcommand{\dotmidtype}[0]{\mtt{.mid}}
\newcommand{\dotretty}[0]{\mtt{.rety}}
\newcommand{\dotinty}[0]{\mtt{.inty}}

\newcommand{\ty}[1]{\tilde{#1}} %

\theoremstyle{definition}

\newtheorem{definition}{Definition}
\newtheorem{theorem}{Theorem}

\newtheorem{lemma}{Lemma}
\newtheorem{property}{Property}

\Crefname{corollary}{Corollary}{Corollaries}
\Crefname{conjecture}{Conjecture}{Conjectures}
\Crefname{informal}{Definition}{Definition}
\Crefname{assumption}{Assumption}{Assumptions}
\crefname{assumption}{Assumption}{Assumptions}
\Crefname{property}{Property}{Properties}
\crefname{property}{Property}{Properties}
\Crefname{proposition}{Proposition}{Propositions}
\crefname{proposition}{Proposition}{Propositions}

\Crefname{paragraph}{Section}{Sections}

\newcommand{\commentout}[1]{}

\newcommand{\tup}[1]{\left<#1\right>\xspace}
\newcommand{\set}[1]{\left\{#1\right\}\xspace}
\newcommand{\brackets}[1]{\left[#1\right]\xspace}
\newcommand{\dom}[1]{\fun{dom}{#1}\xspace}
\newcommand{\img}[1]{img({#1})\xspace}

\newcommand{\true}{\texttt{true}}
\newcommand{\false}{\texttt{false}}

\newcommand{\move}{\ms{Move}\xspace}

\DeclareMathOperator\rightsquigarrowfat{\ensuremath{\xrsquigarrow{\phantom{a.}}}}

\newcommand{\eval}[3]{#2 \vdash #1 \rightarrow_{loc} #3}
\newcommand{\evalml}[3]{
  \begin{multlined}
    #2 \vdash #1 \rightarrow_{loc} 
    \\
    #3
  \end{multlined}
}
\newcommand{\aeval}[3]{#1 \vdash #2 \rightsquigarrow #3}
\newcommand{\aevalmut}[3]{#1 \vdash_{\mi{mut}} #2 \rightsquigarrow #3}
\newcommand{\aevalfat}[3]{#1 \vdash #2 \rightsquigarrowfat #3}
\newcommand{\aevalfatmut}[3]{#1 \vdash_{\mi{mut}} #2 \rightsquigarrowfat #3}
\newcommand{\smalleval}[4]{\codeenv, #1, #2 \vdash #3 \rightarrow^{*} #4}
\newcommand{\geval}[3]{\procid, #1 \vdash #2 \rightarrow_{glob} #3}
\newcommand{\gevalml}[3]{
  \begin{multlined}
    \procid, #1 \vdash #2 \rightarrow_{glob} 
    \\
    #3
  \end{multlined}}

\newcommand{\peval}[2]{\codeenv \vdash #1 \rightarrow #2}
\newcommand{\pevalml}[2]{
  \begin{multlined}
    \codeenv \vdash #1 \rightarrow 
    \\
    #2
  \end{multlined}}
\newcommand{\pevals}[4]{\codeenv, #1, #2 \vdash #3 \rightarrow #4}

\newcommand{\callevalmultline}[3]{%
  \begin{multlined}
    \codeenv, #1, i \vdash #2 \rightarrow
    \\
     #3
  \end{multlined}
  }

\newcommand{\val}{v}

\newcommand{\tg}{t}

\newcommand{\tv}[2]{\tup{{#1},{#2}}\xspace}
\newcommand{\record}[1]{\set{#1}\xspace}

\newcommand{\cons}{{::}}
\newcommand{\stackht}[2]{{#1}\cons{#2}\xspace}

\newcommand{\rft}[2]{\tup{{#1}, {#2}}\xspace}

\newcommand{\st}[1]{\tup{#1}\xspace}

\newcommand{\ao}[2]{\ensuremath{{#1}\tup{#2}}\xspace}

\newcommand{\mdel}[2]{{#1}\setminus{#2}}
\newcommand{\mset}[3]{{#1}\brackets{{#2} \mapsto {#3}}}
\newcommand{\mread}[2]{#1(#2)\xspace}

\newcommand{\lset}[3]{{#1}\brackets{{#2} \mapsto {#3}}}
\newcommand{\lread}[2]{#1(#2)\xspace}

\newcommand{\gdel}[2]{{#1}\setminus{#2}}
\newcommand{\gset}[3]{{#1}\brackets{{#2} \mapsto {#3}}}
\newcommand{\gread}[2]{#1(#2)\xspace}

\newcommand{\cmdname}[1]{{\ensuremath{\bf {#1}}}\xspace}

\newcommand{\movelocCmd}{\cmdname{MvLoc}}

\newcommand{\copylocCmd}{\cmdname{CpLoc}}

\newcommand{\storelocCmd}{\cmdname{StLoc}}

\newcommand{\borrowlocCmd}{\cmdname{BorrowLoc}}

\newcommand{\readrefCmd}{\cmdname{ReadRef}}

\newcommand{\writerefCmd}{\cmdname{WriteRef}}

\newcommand{\callCmd}{\cmdname{Call}}

\newcommand{\returnCmd}{\cmdname{Ret}}

\newcommand{\packCmd}{\cmdname{Pack}}

\newcommand{\unpackCmd}{\cmdname{Unpack}}

\newcommand{\borrowfieldCmd}{\cmdname{BorrowFld}}

\newcommand{\movetoCmd}{\cmdname{MoveTo}}

\newcommand{\movefromCmd}{\cmdname{MoveFrom}}

\newcommand{\borrowglobalCmd}{\cmdname{BorrowGlobal}}

\newcommand{\existsCmd}{\cmdname{Exists}}

\newcommand{\popCmd}{\cmdname{Pop}}

\newcommand{\loadconstCmd}{\cmdname{LoadConst}}

\newcommand{\stackopCmd}{\cmdname{Op}}

\newcommand{\branchcondCmd}{\cmdname{BranchCond}}
\newcommand{\branchCmd}{   \cmdname{Branch}}

\newcommand{\pc}{\ensuremath{{\mathrm pc}}}

\newcommand{\codeenv}{\Omega}
\newcommand{\procid}{P} %
\newcommand{\pth}{p} %

\newcommand{\sep}{\quad}

\definecolor{pblue}{rgb}{0.13,0.13,1}
\definecolor{pgreen}{rgb}{0,0.5,0}
\definecolor{pred}{rgb}{0.9,0,0}
\definecolor{pgrey}{rgb}{0.46,0.45,0.48}

\definecolor{ckeyword}{HTML}{7F0055}
\definecolor{ccomment}{HTML}{3F7F5F}
\definecolor{cnumber}{HTML}{2A0099}

\lstdefinelanguage{Solidity}{
  keywords={contract, function, payable, event, msg},
  ndkeywords={address, uint},
  showspaces=false,
  showtabs=false,
  breaklines=true,
  showstringspaces=false,
  breakatwhitespace=true,
  lineskip=-0.6pt,
  morecomment=[l]{//}, %
  morecomment=[s]{/*}{*/}, %
  basewidth={0.48em, 0.4em},%
  basicstyle=\scriptsize\ttfamily,
  keywordstyle={\color{ckeyword}\ttfamily\bfseries},
  ndkeywordstyle={\color{pblue}\ttfamily\bfseries},
  commentstyle={\color{ccomment}\itshape},
  stringstyle=\color{green},
  moredelim=[il][\textcolor{pgrey}]{$$},
  moredelim=[is][\textcolor{pgrey}]{\%\%}{\%\%}
}

\lstdefinelanguage{Scilla}{
  keywords={accept, contract, emit, end, field, transition, event, send, _sender, _amount, _recipient, match, with},
  ndkeywords={Address, Map, ByStr20, Uint},
  showspaces=false,
  showtabs=false,
  breaklines=true,
  showstringspaces=false,
  breakatwhitespace=true,
  lineskip=-0.6pt,
  morecomment=[l]{//}, %
  morecomment=[s]{/*}{*/}, %
  basewidth={0.48em, 0.4em},%
  basicstyle=\scriptsize\ttfamily,
  keywordstyle={\color{ckeyword}\ttfamily\bfseries},
  ndkeywordstyle={\color{pblue}\ttfamily\bfseries},
  commentstyle={\color{ccomment}\itshape},
  stringstyle=\color{green},
  moredelim=[il][\textcolor{pgrey}]{$$},
  moredelim=[is][\textcolor{pgrey}]{\%\%}{\%\%}
}

\lstdefinelanguage{Move}{
 keywords={abort, acquires, assert, copy, drop, borrow_global, borrow_global_mut, create_account, fun, freeze, module, move_to, move_from, public, resource, break, use, vector, has, key, store, signer, let, struct, while, sum, if, else},
  ndkeywords={address, bool, u64, bytearray, Self, date, spec, invariant, const, global},
  showspaces=false,
  showtabs=false,
  breaklines=true,
  showstringspaces=false,
  breakatwhitespace=true,
  lineskip=-0.6pt,
  morecomment=[l]{//}, %
  morecomment=[s]{/*}{*/}, %
  basewidth={0.48em, 0.4em},%
  basicstyle=\scriptsize\ttfamily,
  keywordstyle={\color{ckeyword}\ttfamily\bfseries},
  ndkeywordstyle={\color{pblue}\ttfamily\bfseries},
  commentstyle={\color{ccomment}\itshape},
  stringstyle=\color{green},
  literate={>=}{$\geq$}2 {<=}{$\leq$ }2
}

\lstdefinestyle{number}{%
  numbers=left,%
  numberstyle=\scriptsize\em,%
  xleftmargin=1em%
}

\newcommand{\mcode}[1]{\lstinline[language=Move,basicstyle=\small\ttfamily]{#1}}
\newcommand{\code}[1]{\mcode{#1}}

\newcommand{\Reference}{\mathsf{Reference}}

\newcommand{\AccountAddress}{\mathsf{Addr}}

\newcommand{\StoreableValue}{\mathsf{StorableValue}}

\newcommand{\addr}{a}

\newcommand{\fldname}{f} %

\newcommand{\stype}{\rho} %
\newcommand{\instr}{i}
\newcommand{\sname}{s}
\newcommand{\loc}{\ell} %

\newcommand{\seq}[1]{[#1]}

\lstdefinelanguage{Move}{
 keywords={
abort, acquires, assert, copy, drop, borrow_global, borrow_global_mut, create_account, fun, freeze, module, move_to, move_from, public, resource, break, use, vector
    continue,
    else,
    false,
    if,
    let, loop,
    move, mut,
    return,
    struct,
    unsafe,
    true,
    while,
    spec, invariant, const, global, forall, has_key
  },
  ndkeywords={address, bool, u64, bytearray, Self, date},
  showspaces=false,
  showtabs=false,
  breaklines=true,
  showstringspaces=false,
  breakatwhitespace=true,
  lineskip=-0.6pt,
  morecomment=[l]{//}, %
  morecomment=[s]{/*}{*/}, %
  basewidth={0.48em, 0.4em},%
  basicstyle=\scriptsize\ttfamily,
  keywordstyle={\color{ckeyword}\ttfamily\bfseries},
  ndkeywordstyle={\color{pblue}\ttfamily\bfseries},
  commentstyle={\color{ccomment}\itshape},
  stringstyle=\color{green},
  moredelim=[il][\textcolor{pgrey}]{$$},
  moredelim=[is][\textcolor{pgrey}]{\%\%}{\%\%}
}

\lstdefinestyle{number}{%
  numbers=left,%
  numberstyle=\scriptsize\em,%
  xleftmargin=1em%
}

 \renewcommand{\com}[1]{{\col{\commoncol }{#1}}}
\renewcommand{\tcenv}[0]{\ensuremath{\Omega^{\dagger}}}

\begin{document}
\date{}

\title{
  \Large \bf
  Robust Safety for \move
}

\author{
\IEEEauthorblockN{Marco Patrignani}
\IEEEauthorblockA{University of Trento}
\IEEEauthorblockA{\mtt{marco.patrignani@unitn.it}}
\and
\IEEEauthorblockN{Sam Blackshear}
\IEEEauthorblockA{Mysten Labs}

} %

\maketitle

\begin{abstract}
A program that maintains key safety properties even when interacting with arbitrary untrusted code is said to enjoy \emph{robust safety}.
Proving that a program written in a mainstream language is robustly safe is typically challenging because it requires static verification tools that work precisely even in the presence of language features like dynamic dispatch and shared mutability.
The emerging \move programming language was designed to support strong encapsulation and static verification in the service of secure smart contract programming.
However, the language design has not been analysed using a theoretical framework like robust safety.

In this paper, we define robust safety for the \move language and introduce a generic framework for static tools that wish to enforce it.
Our framework consists of two abstract components: a program verifier that can prove an invariant holds in a closed-world setting (e.g., the \move Prover~\cite{DBLP:conf/cav/ZhongCQGBPZBD20,prover-new}), and a novel \emph{encapsulator} that checks if the verifier's result generalizes to an open-world setting.
We formalise an escape analysis as an instantiation of the encapsulator and prove that it attains the required security properties.

Finally, we implement our encapsulator as an extension to the \move Prover and use the combination to analyse a large representative benchmark set of real-world \move programs.
This toolchain certifies $>$99\% of the \move modules we analyse, validating that automatic enforcement of strong security properties like robust safety is practical for \move.
Additionally, our results tell that security-centric language design can be effective in attaining strong security properties such as robust safety.
\end{abstract}

\section{Introduction}\label{sec:intro}
Writing correct code is difficult.
Writing code that maintains key safety properties even when interacting with untrusted code is harder still.
Programs that have this property are said to enjoy \emph{robust safety}~\cite{refty-sec-impl,autysec,davidcaps}, which is important in a number of real-world settings such as:
operating system kernels running correctly in the presence of buggy or malicious user-space apps;
browsers isolating JavaScript programs running on different websites;
smart contracts exchanging funds with authorized users while preventing theft from attackers.

There are many techniques for enforcing robust safety: sandboxing~\cite{dg-rs}, process isolation, programming patterns such as object capabilities~\cite{Miller03capabilitymyths}, and specialized hardware~\cite{DBLP:conf/sp/WatsonWNMACDDGL15}.
Alternatively, one can enforce robust safety at the language level~\cite{davidcaps,ot4jc,sec-typ-prot,autysec,refty-sec-impl,tydisa,cca,catalin-rs}.
The vision of this approach is that first, the programmer writes code and specifies key safety invariants (or, assertions).
Then, the language semantics, in concert with static tools (e.g., type systems or program analyses), ensure that these invariants will hold even when the code links against and interacts with untrusted code.
This approach is clearly appealing due to its lack of runtime overhead and its treatment of security as a first class citizen.

Unfortunately, no real-world programming language attains robust safety using this this approach, and we ascribe this to two reasons.
First, real-world languages typically have features that frustrate writing robustly safe code.
For example, dynamic dispatch, shared mutability, and reflection are all common language features that provide a broad attack surface for violating safety invariants.
Second, most practical languages cannot be easily extended with safety-relevant static tools (e.g., efficient program verifiers) or with expressive, integrated specification languages.
Both of these extensions are critical because robust safety is only meaningful with respect to a set of programmer-defined safety invariants that can be verified by a practical tool.

Unlike many existing programming languages, the emerging \move language~\cite{move_white} was designed to support both writing programs that interact safely with untrusted code and static verification.
This design is due to \move being used to program secure smart contracts on the Diem blockchain~\cite{libra_blockchain_white}.
For example, \move has strong encapsulation primitives and omits unsafe features such as dynamic dispatch that have led to costly \emph{re-entrancy} vulnerabilities in other smart contract languages (e.g., the infamous DAO attack~\cite{re_dao}).
The language is co-developed with the \move Prover~\cite{DBLP:conf/cav/ZhongCQGBPZBD20,prover-new}, a verification tool that checks whether \move code complies with invariants written in \move's integrated specification language \emph{locally}.

Unfortunately, the Prover and the design of \move are not sufficient to ensure robust safety on their own (as we exemplify in \Cref{fig:coin}).
One important gap is the absence of a principled characterisation of what it means for \move programs to be robustly safe.
In addition, it is not clear what security properties must be satisfied by the tools used to enforce robust safety for \move programs.

Thus, in this paper, we formalise robust safety for the \move language, define the security properties required of tools that wish to enforce it, and implement some concrete instantiations of these tools, and evaluate them on a large representative benchmark set of real-world \move programs.
Our evaluation lets us conclude that writing robustly-safe \move programs is \emph{practical} and \emph{achievable} for ordinary programmers.
This conclusion comes from these contributions:
\begin{enumerate}
\item We formalise a parametric framework for defining robust safety on \move modules (i.e., partial programs) of interest, which we call \emph{\tc}.
Defining robust safety on \move modules relies on two tools: a program verifier (such as the existing \move Prover) and an \emph{encapsulator} (which is novel).
Intuitively, the verifier checks whether safety invariants of the \tc hold in a closed world containing only \tc, while the encapsulator is a static analysis that detects whether the \tc contains safety leaks that untrusted code can exploit to violate the invariants.
We prove that the combination of these two tools is sufficient to enforce robust safety.

By building on top of the formal semantics for \move~\cite{move_white}, we give a precise characterisation of the security properties that verifiers and encapsulators must uphold in order to attain robust safety; we call such verifiers and encapsulators \emph{valid}.
Then,
we prove that any \tc verified with a valid verifier and approved by a valid encapsulator is robustly safe: its safety invariants cannot be violated by \emph{any} \move code the \tc interacts with.

\item We focus on the role of the encapsulator and formalise a simple intraprocedural escape analysis that we prove to be a valid encapsulator.
The analysis overapproximates the set of references pointing to internal state of the \tc and flags references that may leak to untrusted code.

\item We implement the escape analysis and evaluate both its efficiency and precision on a large representative set of \move benchmarks from a variety of sources.
Our results show that $>$99\% \move modules pass the analysis while the remaining $<$1\% are easily identifiable false positives.
From this, we conclude that:
1) automatically enforced robust safety is a practically achievable goal for \move programmers
and 2) security-centric language design (such as \move's) can be effective in attaining strong security properties such as robust safety.
\end{enumerate}

The paper proceeds by describing the main features of the \move language and giving a high-level description of robust safety (\Cref{sec:overview}).
Then it recounts the semantics of \move and formally defines robust safety as well as valid verifiers and valid encapsulators (\Cref{sec:formal-rs}).
The paper then presents the encapsulator implementation and evaluates it on the aforementioned benchmarks (\Cref{sec:eval}).
Finally, the paper discusses related work (\Cref{sec:rw}) and concludes (\Cref{sec:conc}).

For space constraints many formal details (e.g., some semantics rules), auxiliary lemmas and proofs are elided, the interested reader can find them in the appendix %
Our implementation is open source as part of the \move Prover tool (see \Cref{sec:eval}).

\section{Overview}\label{sec:overview}

We begin by introducing \move through a running example, focussing on the language features that empower programmers to enforce safety invariants even in the presence of adversarial code (\Cref{sec:move-bg}, we defer the reader interested in a general tour of the \move language to the work of \citet{move_white}).
We then describe how the example is insecure: it respects an invariant locally, but not in the presence of arbitrary code, i.e., it is not robustly safe (\Cref{sec:example}).
To recover from this insecurity, we show how our encapsulator analysis flags the vulnerability in the example; addressing this issue would make the code robustly safe (\Cref{sec:why-rs}).

\subsection{Background: \move Language}\label{sec:move-bg}
This section describes the \move features that are relevant for this paper by relying on the running example in \Cref{fig:coin}.
The example contains a \move \emph{module} that implements a custom currency \code{NextCoin}.
Note that we defer presenting the security-relevant details of \Cref{fig:coin} until \Cref{sec:example}.
\begin{figure}[!htb]
\begin{lstlisting}[
  escapechar=|,
  label=fig:coin,
  caption = Implementation of a coin asset in \move.
  ]
module 0x1::NextCoin {
  use 0x1::Signer;                                                          |\label{line:signer}|

  struct Coin has key { value: u64 }
  struct Info has key { total_supply: u64 }

  const ADMIN: address = 0xB055;

  // below is the definition of an invariant
  spec { invariant: forall c: Coin, global<Info>(ADMIN).total_supply = sum(c.value)  }              |\label{line:inv1}|

  public fun initialize(account: &signer) {
    assert(Signer::address_of(account) == ADMIN, 0); |\label{line:assert-init}|
    move_to<Info>(account, Info { total_supply: 0 }) |\label{line:moveto-ex}|
  }

  // the next function temporarily violates
  // and then restores the invariant
  public fun mint(account: &signer, value: u64): Coin {
    let addr = Signer::address_of(account);                                 |\label{line:check1}|
    assert(addr == ADMIN, 0);                                               |\label{line:check2}|
    let info = borrow_global_mut<Info>(addr);                               |\label{line:borrowglobalmut-ex}|
    info.total_supply = info.total_supply + value;                          |\label{line:inv-break}|
    // invariant temporarily violated
    Coin { value } // invariant restored                       |\label{line:inv-rest}|
  }

  // this function violates the invariant
  public fun value_mut(coin: &mut Coin): &mut u64 {
    &mut coin.value   // not safe!                                          |\label{line:unsafe}|
  }
}
\end{lstlisting}

\end{figure}

\paragraph{Modules}
Each \move module consists of a list of struct type and procedure definitions.
A module can import type definitions (e.g., \code{use 0x1::Signer} on \Cref{line:signer}) and call procedures declared in other modules.
The fully-qualified name of a module begins with a 16 byte \emph{account address} where the code for the module is stored (here, we write an account address like \code{0x1} as shorthand for a 16 byte hexadecimal address padded out with leading 0s). %
The account address acts as a namespace that distinguishes modules with the same name; e.g., \code{0x1::NextCoin} and \code{0x2::NextCoin} are different modules with their own types and procedures.

\paragraph{Structs}
The module defines two data structures \code{Coin} and \code{Info}.
A \code{Coin} represents the currency allocated to users of the module while \code{Info} records how much of that currency exists in total.
Both of these structs can be stored in the persistent global key/value store since they define keys; this is indicated by the \code{has key} syntax on the declaration.

\paragraph{Procedures}
The code defines an initialization, a safe procedure, and an unsafe procedure, which we now describe.

The \code{initialize} procedure must be called before any \code{Coin} is created, and it initializes the \code{total_supply} of the singleton \code{Info} value to zero.
Here, \code{signer} is a special type that represents a user authenticated by logic outside of \move (similar to e.g., a Unix UID).
Asserting that the \code{signer}'s address is equal to \code{ADMIN} ensures that this procedure can only be called by a designated administrator account (\code{0xB055}, in this case).

Procedure \code{mint} lets the administrator create new coins of a desired amount (\Cref{line:inv-rest}); this is done after the total amount of coins is updated (\Cref{line:inv-break}).
Like \code{initialize}, this procedure has access control to ensure that it can only be called by the administrator account (\Cref{line:check1,line:check2}).

Procedure \code{value_mut} takes a mutable reference to a \code{Coin} as input (thus the type \code{&mut}) and returns a mutable reference that points to the \code{value} field of the coin.

\paragraph{Persistent Global Store}
The global store allows Move programmers to store persistent data (e.g., \code{Coin} balances) that can only be programmatically read/written by the module that owns it, but is also stored by a public ledger that can be viewed by users running code in other modules.

Each key in the global store consists of a fully-qualified type name (e.g., \code{0x1::NextCoin::Coin}) and an account address where a value of that type is stored (account addresses store both module code and struct data).
Although the global store is shared among all modules, each module has exclusive read/write access to keys that contain its declared types.
Thus, only the module that declares a struct type such as \code{Coin} can:
\begin{itemize}
\item Publish a value to global storage via the \code{move_to<Coin>} instruction (e.g., \Cref{line:moveto-ex});
\item Remove a value from global storage via the \code{move_from<Coin>} instruction;
\item Acquire a reference to a value in global storage via the \code{borrow_global_mut<Coin>} instruction (e.g., \Cref{line:borrowglobalmut-ex}).
\end{itemize}
Since a module ``owns'' the global storage entries keyed by its types, it can enforce constraints on this memory.
For example, the code in \Cref{fig:coin} ensures that only the \code{ADMIN} account address can hold a struct of type \code{0x1::NextCoin::Info}.
It does this by only defining one procedure (\code{initialize}) that uses \code{move_to} on an \code{Info} type and enforcing the precondition that that \code{move_to} is called on the \code{ADMIN} address (\Cref{line:assert-init}).%
\footnote{
  Note that \move has transactional semantics---any program that fails an assertion or encounters a runtime error (e.g., integer overflow/underflow, \code{move_to<T>(a)} on an account address \code{a} that already stores a \code{T}) will \emph{abort} and have no effect on the global storage.
}
These constraints are unlike invariants (which we describe next) since they require runtime checking.
In this case, since parameter \code{account} is supplied at runtime, the programmer cannot enforce statically that it will always be \code{ADMIN}, hence the check on \Cref{line:assert-init}.

\paragraph{Invariants}
The module contains an invariant to be checked statically on \Cref{line:inv1}: the amount stored in \code{Info} correctly tracks how many \code{Coin}s have been allocated.
This is described more in depth in \Cref{sec:example}.

\paragraph{\move Bytecode Verifier: Safe Type Reuse and Linearity}
Although other modules cannot access global storage cells keyed by \code{0x1::NextCoin::Coin}, they can use this type in their own procedure and struct declarations.
For example, another module could expose a \code{pay} function that accepts a \code{0x1::NextCoin::Coin} as input or a \code{Bank} struct with a \code{balance} field whose type is \code{0x1::NextCoin::Coin}.

At first glance, allowing sensitive values like \code{Coin}s to flow out of the module that created them might seem dangerous --
what stops a malicious client module from creating counterfeit \code{Coin}s, artificially increasing the \code{value} of a \code{Coin} it possesses, or copying/destroying existing \code{Coin}s?
Fortunately, \move has a bytecode verifier (a type system enforced at the bytecode level, as in the JVM~\cite{jvm} and CLR~\cite{clr}) that allows module authors to prevent these undesired outcomes.
In particular, only the module that declares a struct type \code{Coin} can:
\begin{itemize}
\item Create a value of type \code{Coin} (e.g., \Cref{line:inv-rest});
\item ``Unpack'' a value of type \code{Coin} into its component field(s) (\code{value}, in this case);
\item Acquire a reference to a field of \code{Coin} via a Rust-style~\cite{rust} mutable or immutable borrow (e.g., \code{&mut coin.value} at \Cref{line:unsafe}).
\end{itemize}
This allows the module author to enforce invariants on the creation and field values of the structs declared in the module.

The verifier also enforces structs to be \emph{linear} by default~\cite{linear_logic,linearhs,Wadler90lineartypes}.
Linearity prevents copying and destruction (e.g., via overwriting the variable that stores the struct or allowing it to go out of scope) outside of the module that declared the struct.%
\footnote{
  The programmer can choose to override these defaults by declaring a struct with the \code{copy} (e.g., \code{struct S has copy}) \emph{ability} to allow copying or the \code{drop} ability to allow unconditional destruction.
}
Although the bytecode verifier of Move enforces many useful properties such as type safety, memory safety, and resource safety~\cite{blackshear2020resources, blackshear2022borrow}, it is not powerful enough to enforce robust safety.
We now explain why by describing \move code invariants.

\subsection{Invariants and Vulnerability}\label{sec:example}

A safety invariant of the module is described on \Cref{line:inv1}: the sum of the \code{value} fields of all the \code{Coin} objects in the system must be equal to the \code{total_value} field of the \code{Info} object stored at the \code{ADMIN} address.
We refer to this invariant as the ``\emph{conservation property}''.
We want the conservation property to hold for all possible clients of the module (including malicious ones): any violation undermines the integrity of the currency.
As such, the invariant talks about not just on a single object, but on a collection of them (i.e., all the \code{Coin}s).
We now show how the conservation property is established and maintained using the encapsulation features of \move before explaining how procedure \code{value_mut} allows the property to be violated (despite the module being well-typed according to the verifier).

\paragraph{Establishing the Conservation Property}
Calling \code{initialize} sets up the module with the invariant: no \code{Coin} exists and thus the \code{total_supply} in \code{Info} is set to \code{0}.

After initialization, the \code{mint} procedure can be invoked to create \code{Coin}s.
Note that this procedure temporarily violates the conservation property!
The invariant is not required to hold at every program point (which would be overly strict~\cite{DBLP:conf/popl/CoughlinC14}); only at the beginning (precondition) and at the end (postcondition) of every public procedure of the module.
And indeed the final line of the procedure restores the invariant by creating and returning a \code{Coin} with the corresponding \code{value}.

\paragraph{Violating the Conservation Property}
For procedures \code{initialize} and \code{mint}, the conservation property always holds at the postcondition under the assumption that it holds at the precondition.
However, an attacker (\Cref{fig:attacker}) can violate the property by leveraging procedure \code{value_mut}.
Note that this procedure does not violate the conservation property on its own, but an attacker can use it to break the property:
\begin{figure}[!htb]
\begin{lstlisting}[
  escapechar=|,
  label=fig:attacker,
  caption = An attacker to the code of \protect\Cref{fig:coin}.
  ]
fun attacker(c: &mut Coin) {
  let value_ref = Coin::value_mut(c);
  *value_ref = *value_ref + 1000; // violates conservation!
}
\end{lstlisting}
\end{figure}

Although our \code{Coin} example is somewhat artificial (a realistic coin implementation would have no need for a procedure like \code{value_mut}), it illustrates the difficulty of writing robustly safe code.
It is not enough for the module code to establish and maintain key safety invariants internally; it must also ensure that no possible client can violate the invariant.

The way to ensure no client of \code{Coin} can violate the conservation property is to show that it is robustly safe, so we now describe how to enforce this property in practice.

\subsection{Detecting Robust Safety Violations With Encapsulator Analysis}\label{sec:why-rs}
At an intuitive level, leaking references to fields of declared structs is the only way \move programs can fail to be robustly safe.
Stated differently: if a \move module establishes its key invariants locally \emph{and} avoids leaking references to structs involved in these invariants, then these invariants also hold globally for all possible clients of the module.
Making this statement precise (and true) is the goal of the formalisation of robust safety in \Cref{sec:formal-rs}.

We detect leaks of structs involved in programmer-specified safety invariants with an intraprocedural escape analysis.
When the analysis begins analysing a procedure, it binds all mutable reference parameters to the abstract value $\okref$.
Borrowing an invariant-relevant field from $\okref$ produces the abstract value $\inref$, indicating a pointer into module-internal state.
The analysis flags a leak if an $\inref$ flows into the return value of the function; such a flag means that sensitive writes to module-internal state may occur outside of the module.
Because \move structs cannot store references and the global storage only holds struct values, this is the only way such a leak can occur.

If we apply this analysis to the problematic function \code{value_mut}, the \code{coin} parameter is initialized to $\okref$.
Borrowing the invariant-relevant \code{value} field (\Cref{line:unsafe}) consumes the $\okref$ and produces an $\inref$.
This value is subsequently returned by the function, which is flagged by the analysis.
None of the other functions return references, so the analysis flags only this vulnerable function.
Deleting the function makes the module robustly safe with respect to the conservation property, and the analysis will recognize this.

\paragraph{Is Robust Safety so Simple?}
Our running example may leave the reader with the impression that it is trivial to enforce robust safety: just avoid leaking internal state!
We emphasize that this principle is also sufficient to ensure robust safety in other languages (e.g., Solidity, C++).
However, languages typically provide many ways to ``leak'' (e.g., returning references, references stored in data structures, re-entrancy, memory safety violations, \ldots), and precisely identifying such leaks with an intraprocedural analysis (or even a much more sophisticated analysis) is not practical.
The difference between existing languages and \move is that it possible to state sufficient conditions for robust safety and design an efficient, local analysis that checks whether these conditions hold.
This enables the development of a generic developer tool that checks robust safety, i.e., the escape analysis that we present in \Cref{sec:eval}.
Thus, our work validates the fact that \move's careful design enables efficient, precise verification of robust safety.
We recap the benefits of our approach in more detail in \Cref{sec:compare-rs}, after providing more details.

Despite this intuitive simplicity, formalising what robust safety means precisely for \move (and what security properties that tools such as the escape analysis must uphold) is non-trivial -- that is what the next section discusses.

\section{Formal Results: Robust Safety for \move} %
\label{sec:formal-rs}
This section provides a formalisation of the key security property the \move language attains: robust safety.
For this, it first provides a brief definition the semantics of the \move language (\Cref{sec:move-sem}), as taken from the work of \citet{blackshear2020resources}.
Then, it describes the threat model we consider (\Cref{sec:threatmodel}): this includes the attacker formalisation, the trace model used to capture security-relevant behaviour, and the invariants that define robust safety.
As robust safety is attained by virtue of three tools (the bytecode verifier, the prover, and the encapsulator) this section describes the formal properties such tools must fulfil (\Cref{sec:tools}).
Finally, this section proves that any \move module certified by tools that satisfy these properties is robustly safe (\Cref{sec:rs-ppr}).

Due to space constraints, this section contains a subset of the formal rules, no auxiliary lemmas, and no proofs; the interested reader can find the full formalisation and proofs in the appendix. %

\subsection{\move Language Excerpts} %
\label{sec:move-sem}
\move programs are functions that execute on a stack machine whose peculiarity is the treatment of the storage.
Formally, the global store mentioned in \Cref{sec:move-bg} is split into two parts: the memory and the globals~\cite{blackshear2020resources}.
The memory is a first-order store and as such its cells cannot be used to store pointers (which we call locations) to memory cells.
Globals are instead used to store pointers to memory cells, but they are indexed differently from the memory.
In order to access a global, the code provides an address (a literal) and a type bound to that address (this is akin to the type structs mentioned in \Cref{sec:move-bg}).
This division simplifies formalising the semantics of moving values on the operand stack.
In the \move language, any value can be destructively \emph{moved} (invalidating the storage location that formerly held the value), but only certain values (e.g., integers) can be copied.

\move programs are organised in modules (\com{\codeenv}), which contain lists of functions declarations (\com{\procid}), which in turn contain input and output types as well as their list of instructions (\com{\seq{\instr}}).
\move programs run on a stack machine whose state (\com{\stt}) is a tuple \com{\st{C,M,G,S}} composed of: the call stack, the memory, the globals and the operand stack.
The state of the stack machine also maintains a function table (the module \com{\codeenv} itself) to resolve the instructions comprising the bodies of functions.

The call stack \com{C} contains a stack of triples that record which function is executing.
Each triple \com{\st{\procid,\pc,L}} contains the name of the function and the program counter (\com{\procid} and \com{\pc}, which are used to find the current instruction in the lookup table), and a stack of locals (\com{L}), which are bindings (\com{\var\mapsto\val}) from local variables (\com{\var}) to arbitrary values (\com{\val}).
Values (\com{\val}) can either be locations (\com{\loc}) or storable values (\com{r}); the latter can either be ground values (\com{z}, which include addresses \com{\addr}) or records (whose id we indicate as \com{\sname}).
The memory (\com{M}) is a map from locations to storable values (\com{\loc\mapsto r}) while the globals (\com{G}) map resource identifiers to locations (\com{\st{\addr,\stype}\mapsto\loc}) that only contain records.
Resource identifiers (\com{\st{\addr,\stype}}) are a pair of an address (\com{\addr}) and a type (\com{\stype}), the latter is used to identify the function that defined that global and the type of the record the global points to (and for this, technically, \com{\stype} contains a module id and a struct id \com{s}).
The shared operand stack (\com{S}) contains all values consumed (and produced) by instructions as well as those passed to (and returned by) functions.
Given the presence of records, \move uses paths (\com{\pth}) i.e., lists of field names (\com{\fldname}) to traverse nested records and look up or update part of a record.
For simplicity we assume all field names are distinct.

\subsubsection{\move Operational Semantics} %
\label{sec:sem}

The stack machine has a small-step semantics whose judgement is \com{\pevals{\procid}{i}{\stt}{\stt'}} and it is read \emph{``in module \com{\codeenv}, instruction \com{\instr} in function \com{\procid} modifies state \com{\stt} into \com{\stt'}''}.
This semantics relies two additional kinds of reductions for global and local reductions.
The first ones follow this judgement: \com{ \geval{ \instr }{ \st{M,G,S} }{ \st{M',G',S'} } } and they describe the semantics of instructions \com{\instr} in function \com{\procid} that only modify globals (either via the operand stack \com{S} or via the memory \com{M}).
The second ones follow this judgement: \com{\eval{ \st{M,L,S} }{ \instr }{ \st{M',L',S'} } } and they describe the semantics of instructions \com{\instr} that only modify locals (again, either via the operand stack or via the memory).
The list of \move instructions is in \Cref{fig:instr}, they include calling, returning, branching conditionally and unconditionally, moving a value from memory to the stack (and back), borrowing a global, checking the existence of a global, packing and unpacking a record, moving a value to the local stack (and back), copying it, borrowing it, popping a value, loading a constant, binary operations, reading (and writing) to memory and accessing a record field.
\myfig{
  \centering
  \begin{align*}
    \text{instrs.}
      &\
      \ao{\callCmd}{\procid} \mid \returnCmd \mid \ao{\branchcondCmd}{\pc}
    \\
    \mid
      &\
      \ao{\branchCmd}{\pc}
    \\
    \text{global instrs.}
      &\
      \ao{\movetoCmd}{\sname} \mid \ao{\movefromCmd}{\sname} \mid \ao{\existsCmd}{\sname}
    \\
    \mid
      &\
      \ao{\borrowglobalCmd}{\sname}
      \mid \ao{\packCmd}{\sname} \mid \ao{\unpackCmd}{\sname}
    \\
    \text{local instrs.}
      &\
      \ao{\movelocCmd}{\var} \mid \ao{\storelocCmd}{\var} \mid \ao{\copylocCmd}{\loc}
    \\
    \mid
      &\
      \ao{\borrowlocCmd}{\var}
      \mid \popCmd \mid \ao{\loadconstCmd}{\addr} \mid \stackopCmd
    \\
    \mid
      &\
      \readrefCmd \mid \writerefCmd \mid \ao{\borrowfieldCmd}{\fldname}
  \end{align*}
}{instr}{
  Instructions of the \move language.
}

Most of the rules are unsurprising and therefore omitted, we only provide the most interesting ones in \Cref{fig:sem}, i.e., those that deal with the moving of resources.
Notation-wise, we indicate accessing a map (such as the memory) as \com{\mread{M}{\loc}} and updating its content as \com{\mset{M}{\ell}{\val}}; we use the same notation for locals, globals and for looking up functions in a module.
We indicate the domain of a map \com{M} as \com{\dom{M}}.
A list of elements \com{K} is denoted with \com{\seq{K}},
and its length as \com{\card{\seq{K}}}.
We use dot notation to access sub-parts of procedures \com{\procid}, namely
\com{\procid\dotmidtype} is the module identifier of the procedure and \com{\procid\dotinty} and \com{\procid\dotretty} are the lists of inputs and return types of \com{\procid} respectively.
Function \com{\fun{instr}{\codeenv,\stt}} returns the current instruction by looking it up in the codebase \com{\codeenv} given the current function and the \com{\pc} from the top of the call stack in \com{\stt}.

\myfig{
  \centering
  \typerule{[MoveLoc]}
  {
      \lread{L}{\var} = \loc
      &
      \loc\in\dom{M}
  }{
    \eval{
      \st{M,L,S}
    }{
      \ao{\movelocCmd}{x}
    }{
      \st{\mdel{M}{\loc},\mdel{L}{\var},\stackht{\mread{M}{\loc}}{S}}
    }
  }{sem-loc-moveloc}
  \typerule{[StoreLoc]}
  {
      \com{\val} \in \StoreableValue
      &
      \com{\loc} \notin \com{\dom{M}}
      \\
      \com{M'} = \mdel{M}{\lread{L}{\var}} \text{ if } \lread{L}{\var} \in\dom{M} \text{ else } M
  }{
    \evalml{
      \st{M,L,\stackht{\val}{S}}
    }{
      \ao{\storelocCmd}{\var}
    }{
      \st{\mset{M'_{\mdel{\mc{C}}{\loc}}}{\loc}{\val}, \lset{L}{\var}{\loc},S}
    }
  }{sem-loc-storeloc}

  \typerule{[MoveFrom]}
  {
      \stype = \tup{\procid\dotmidtype, s}
      &
      \gread{G}{\tup{\addr,\stype}} = \loc
      &
      \mread{M}{\loc}=\val
  }{
    \gevalml{
      \ao{\movefromCmd}{\sname}
    }{
      \st{M,G,\stackht{\addr}{S}}
    }{
      \st{\mdel{M}{\loc},\gdel{G}{\tup{\addr,s}},\stackht{\val}{S}}
    }
  }{sem-glob-movefrom}
  \typerule{[MoveTo]}{
      \stype = \tup{\procid\dotmidtype, s}
      &
      \tup{\addr, \stype} \notin \dom{G}
      &
      \loc\notin\dom{M}
      \\
      M'=\mset{M_{\mdel{\mc{C}}{\loc}}}{\loc}{\val}
      &
      G'=\gset{G}{\tup{\addr,\stype}}{\loc}
  }{
    \geval{
      \ao{\movetoCmd}{\sname}
    }{
      \st{M,G,\stackht{\stackht{\addr}{\val}}{S}}
    }{
      \st{M',G',S}
    }
  }{sem-glob-moveto}
  \typerule{[Step-Loc]}{
      \fun{instr}{\codeenv,\stt}=\instr
      &
      \eval{\tup{M,L,S}}{i}{\tup{M',L',S'}}
  }{
    \pevalml{
      \st{\stackht{\tup{\procid,\pc,L}}{C}, M,G,S}
    }{
      \st{\stackht{\tup{\procid,\pc+1,L'}}{C}, M',G,S'}
    }
  }{sem-loc}
  \typerule{[Step-Glob]}{
      \fun{instr}{\codeenv,\stt}=\instr
      &
      \geval{
        \instr
      }{
        \tup{M,G,S}
      }{
        \tup{M',G',S'}
      }
  }{
    \pevalml{
      \st{\stackht{\tup{\procid,\pc,L}}{C}, M,G,S}
    }{
      \st{\stackht{\tup{\procid,\pc+1,L}}{C}, M',G',S'}
    }
  }{sem-glob}
}{sem}{
    Semantics of the \move language (excerpts).
}
\Cref{tr:sem-loc-moveloc} performs a destructive read of local variable \com{\var} by removing it from the domain of \com{L} and placing its value (\com{L(\var)}) on the stack.
\Cref{tr:sem-loc-storeloc} places the top of the stack in variable \com{\var}, and that variable in a fresh memory location \com{\loc}, deleting any location that \com{\var} pointed to from memory.
This rule also shows the memory allocator \mc{C}, which is a set of (fresh) locations that allocation can draw from, for simplicity we often omit \mc{C} and report it only when necessary.
\Cref{tr:sem-glob-movefrom} starts from an address (\com{\addr}) and the type \com{\stype} of the currently-executing function \com{\procid} to look up a memory location \com{\loc} and then push its content \com{\val} on the operand stack, removing the memory and global locations just read.
\Cref{tr:sem-glob-moveto} publishes a value \com{\val} to a fresh memory location \com{\loc} that is itself published to a fresh global \com{\tup{\addr,\stype}} whose type is defined by the currently-executing function \com{\procid}.
The role of \com{\stype} is key here: note that it is not programmer-supplied, but it is computed by the semantics (i.e., by the \move abstract machine) which ensures that any resource being moved belongs to the code that is moving it.
This rules out certain attacks (as resources defined in a module cannot be accessed outside it, as mentioned in \Cref{sec:move-bg}), but it still leaves the door open for confused deputy attacks, where external code tricks \tc into insecure behaviour (similar to \Cref{fig:attacker}).
Finally, \Cref{tr:sem-loc,tr:sem-glob} show how the local and global steps affect the top-level reduction judgements.

\subsubsection{Static Semantics} %
\label{sec:typ}

As we mentioned in \Cref{sec:overview,} any \move code that is executed must pass through a bytecode verifier~\cite{blackshear2020resources} to ensure that all \move code is well-typed, meaning that, e.g., operations that require a \mb{N} are supplied a \mb{N} and values that cannot be copied are not copied but only moved.
We indicate a module \com{\codeenv} to be well-typed as: $\wt{}{\codeenv}$.
In order to typecheck instructions, the verifier uses a stack of local types (\com{\ty{L}}) and of operand types (\com{\ty{S}}), which are analogous to their semantics counterpart save that instead of tracking values they track types (\com{\tau}).
The typing of \move instructions follows the judgement \com{\typingdef}, which reads \emph{``instruction \com{\instr} (in function \com{\procid}, in module \com{\codeenv}) requires locals typed \com{\ty{L}} and operands typed \com{\ty{S}} and returns locals typed \com{\ty{L'}} and operands typed \com{\ty{S'}}''}.
As for the semantics, typing is unsurprising and therefore omitted.

\subsection{Threat Model} %
\label{sec:threatmodel}

The start of our threat model is the element whose security we are interested in, and that is some \move module of interest that we call the \tc and that we denote as \com{\tcenv}.
We now describe what are the attackers to the \tc (\Cref{sec:attacker}) and invariants, i.e., the specific formulation of security properties that must hold on \tc (\Cref{sec:invariants}).
We conclude this section by describing the trace model used to formalise the trace semantics capturing the security-relevant behaviour of the \tc (\Cref{sec:tracemodel}).

\subsubsection{Attacker} %
\label{sec:attacker}
An attacker is code that is linked against the \tc so that they call each other's function (and return to each other after said calls).
We can thus identify a \emph{boundary} that separates between attacker code and \tc and that some of the notions described below rely on.

Currently, \move programs are smart contracts deployed on blockchains, and as such we focus on a blockchain-based attacker; this affects how code interacts and what security we can enforce on data, as we now discuss.
We identify two main classes of attackers based on whether the dependencies of \tc with attacker code is immutable or not and call them the \emph{immutable} attacker and the \emph{mutable} one.

When the \tc is deployed with an immutable attacker, it knows that any existing code cannot change, so the attacker is any code that gets deployed \emph{temporally after} the \tc.
This attacker can call the \tc and the \tc can return to it, but any code that the \tc calls is not attacker.
In fact, the publisher of the \tc can verify both the \tc and any of its dependency before publishing.
On the other hand, when \tc is deployed with a mutable attacker, verifying existing code is not helpful, since it can change in the future.
In this case, the attacker can both call and return to the \tc.
The mutable attacker exists also beyond blockchain settings, it is the typical attacker considered in robust safety works, since typically one does not know what code the \tc will link against, only their signatures~\cite{davidcaps,dg-rs}.
In this paper we focus primarily on immutable attackers, though we demonstrate how to attain robust safety for immutable ones too in \Cref{sec:non-batt}.

All blockchain data being public suggests that data confidentiality is not a security goal, but data integrity is (i.e., we are not interested in hiding how much money there is, but we are interested in nobody getting more money than they should).%
\footnote{
  We leave considering a different attacker and thus devising an encapsulator that enforces confidentiality of data for future work.
}

To clearly capture the power of attackers (\com{\atk}), we formalise them as pairs consisting of a code environment (\com{\codeenv}) and a main function (\com{\procid_{main}}).
We impose minimal constraints on attackers, namely that they are well-typed (i.e., $\wt{}{\atk}$) and that they define functions that do not overlap with those defined in the \tc\ \com{\tcenv}.
Any attacker function can call \tc, then immutable attacker functions cannot be called by \tc, while mutable attacker functions can be.
We call these attackers valid and denote this fact as: $\com{\tcenv} \vdash \com{\atk} : \com{atk}$.
Linking some \tc\ \com{\tcenv} against attacker \com{\atk} is denoted as $\tcenv+\atk$ and it returns a module comprising all functions defined in both \com{\tcenv} and in the module part of \com{\atk}.
With a small abuse of notation we use metavariable \com{\atk} for both an attacker and for just its code environment to differentiate it from the code of interest.
When an attacker is linked against the \tc, we assume execution starts from the function \com{\procid_{main}} defined in \com{\atk}.
We call that the starting state of the stack machine (i.e., memory, globals and stacks are all empty) and indicate it as \com{\SInit{\tcenv+\atk}}.

\subsubsection{Invariants} %
\label{sec:invariants}

Invariants contain the list of globals that point to memory locations with a logical invariant as well as the list of memory locations with a logical invariant, so they contain all of the objects with an invariant on (in the sense of \Cref{sec:example}).
For each memory location, invariants define a logical condition that must hold for the content of that location (as in \Cref{fig:coin}).
Invariants can describe relationships between structs or global storage locations in distinct modules as long as the modules have a dependency relationship.
Some invariants describe a dynamic footprint (e.g., all values of type \code{Coin}, the storage location of type \code{Bank} under every possible address) that the verification process must approximate statically~\cite{prover-new}.

We indicate invariants as \com{\inv} and leave their formal details abstract to avoid binding our formalisation to a specific implementation.
For this reason, we work with invariants axiomatically, via the functions described below.
Function \domG{\com{\inv}} returns the globals for which invariants are defined, i.e., the pairs \com{\addr,\stype} that identify globals for which an invariant is defined.
We indicate whether some field \com{\fldname} belongs to a global with an invariant as $\com{\fldname}\in\inv$.
Function \invcond{\inv, M} evaluates the condition for all locations in memory \com{M} and returns true if the condition is satisfied or false otherwise.

With these functions we can define whether a memory and a global satisfy some invariant ($\com{M,G} \vdash \com{\inv}$).
This holds if restricting all memory locations to those mapped by a global yields a memory that contains values that satisfy the conditions the invariant imposes on them.
We use notation \com{\restr{M}{\ell}} to restrict memory \com{M} (and similarly for globals and other elements) to the element \com{\ell}, which is in the domain of \com{M}.
\begin{center}
  \typerule{Invariant Satisfaction}{
    \com{G_i} = \restr{G}{\domG{\inv}}
    &
    \com{M_i} = \restr{M}{G_i}
    &
    \invcond{\inv,M_i} = \com{true}
  }{
    \com{M,G} \vdash \com{\inv}
  }{inv-sat}
\end{center}

This abstract characterisation lets us model invariants such as the one on \Cref{line:inv1} in \Cref{fig:coin}.
In fact the \com{\restr{M}{G_i}} returns all the memory locations that contain structs with an invariant on, i.e., the \code{Info} \code{struct} as well as all all minted \code{Coin} \code{struct}\code{s}.
With \invcond{\inv,M_i}, we express in an abstract fashion the condition that the first field of the former (\code{Info.total_supply}) equals the \code{sum} of all the first fields of the latter ones (\code{Coin.value}).

Invariants are defined for a code environment, which can be obtained from \com{\inv} as follows: $\fun{codeof}{\inv} = \codeenv$.
A code environment \com{\codeenv} and an invariant \com{\inv} are in agreement if the former is the code of the latter.
Formally: $\agree{\codeenv}{\inv} \isdef \fun{codeof}{\inv}=\codeenv$.

Dually, given a code environment, we can identify its subset wrt an invariant as follows: \com{\restr{\codeenv}{\inv}}.
This operation returns the code environment \com{\codeenv'} that is contained in \com{\codeenv} and that only talks about the code mentioned in \com{\inv}, without any other code that has no invariant on.
This is used to identify the sub-part of a code environment that needs to be encapsulated, as we discuss later in \Cref{sec:prec}.

Invariants uphold a key property: none of the types mentioned in their globals (i.e., in \domG{\inv}) are attacker types, i.e., all of those types are defined in the code of that the invariant refers to.
\begin{property}[Invariants are not on Attacker-Typed Globals]\label{prop:inv-type-atk}
  \begin{align*}
    \forall \st{a,\stype} \in \domG{\inv}, \stype \in \fun{declaredtypes}{\fun{codeof}{\inv}}
  \end{align*}
\end{property}

\subsubsection{Trace Model and Trace Semantics} %
\label{sec:tracemodel}
Having defined invariants, we need to collect all security-relevant events produced as computation progresses in order to check whether those invariants hold or not.
Choosing when events are produced is crucial in order to assess safety of \tc and in this work we record events where any control is passed from \tc to the attacker and back.
This way the \tc can internally violate the invariants, but so long as they are reinstated before control is passed to the attacker, no safety violation is detected.
This is intuitively ok, as explained in \Cref{sec:example}.
The only missing bit is that we need to ensure that the attacker does not tamper with our invariants -- or that if she does, this will be recorded -- and this is discussed later, in \Cref{sec:rob}.

Formally, observable events (also called actions, \com{\ac}), follow this grammar and they are concatenated in traces (\com{\trc}).
\begin{align*}
  \com{\trc} \bnfdef
    &\
    \com{\nil} \mid \com{\trc}\listsep\com{\ac}
  \\
  \com{\ac} \bnfdef
    &\
    \com{\trcl{\procid}{M,G}} \mid \com{\trcb{\procid}{M,G}} \mid \com{\trrt{M,G}} \mid \com{\trrb{M,G}}
\end{align*}
Actions include calling function \com{\procid} into \tc, calling function \com{\procid} into attacker code, returning to attacker code and returning to \tc.
  We borrow decorators $?$ and $!$ from process calculi literature in order to indicate the ``direction'' of the action i.e., from attacker to \tc ($?$) or back ($!$)~\cite{bookpi}.
Crucially, all actions record elements of the stack machine state that are relevant from a security perspective: the globals \com{G} and the memory \com{M}.
Given an action \com{\ac}, we indicate its \com{M} and \com{G} elements as \mandg{\ac}.
This lets us apply the invariant verification (i.e., \Cref{tr:inv-sat}) to the globals and memory sub-part of an action and then to a trace as:
\begin{center}
  \typerule{Action-check}{
    \mandg{\ac}\vdash\com{\inv}
  }{
    \actcheck{\com{\ac}}{\com{\inv}}
  }{event-glob-ind}
  \typerule{Trace-check}{
    \forall\com{\ac}\in\com{\trc}\ldotp
    \actcheck{\com{\ac}}{\com{\inv}}
  }{
    \tracecheck{\com{\trc}}{\com{\inv}}
  }{trace-glob-ind}
\end{center}

\paragraph{Trace Semantics}
We now define a big-step trace semantics on top of the small-step operational semantics of \Cref{sec:sem} in order to obtain the traces of some \tc of interest.
The trace semantics is structured on three levels and selected rules are presented in \Cref{fig:trsem}.
First, there is a single-step, single-labelled semantics that is responsible of generating the single actions, its judgement is \com{\tcenv \triangleright \codeenv,\procid, \instr \vdash \stt \xto{\ac} \stt'}.
The trace semantics is defined for whole programs, i.e., for \tc that is linked against some attacker and then run.
However, the trace semantics needs to remember the perspective from which the trace is being generated, i.e., which one is the \tc of interest.
This gets reflected in the judgement of the trace semantics which extends the one of the operational semantics with this information (the \com{\tcenv} on the left).
Second, there is a big-step, single-labelled semantics that is the reflexive-transitive closure of the previous one, its judgement is \com{\tcenv \triangleright \codeenv,\procid \vdash \stt \Xto{\ac} \stt'}.
Third, there is the big-step, trace-labelled semantics that concatenates all big-step single-labelled steps into a trace, its judgement is \com{\tcenv \triangleright \codeenv,\procid \vdash \stt \Xtol{\trc} \stt'}.
Lastly, in order to decorate the generated actions with $?$ or $!$, we rely on function $\tcenv\vdash C : ?/!/\mi{same}$.
This function analyses the top two elements of the call stack \com{C} and tells whether they belong to functions defined by \tc\ \com{\tcenv} and attacker ($?$), attacker and \tc ($!$), or by the same entity (\mi{same}).

\myfig{
  \centering
  \typerule{Action-No}{
    \fun{instr}{\codeenv,\stt}=\instr
    &
    \com{\peval{\stt}{\sigma'}}
    &
    \com{\stt}=\com{\tup{C,M,G,S}}
    \\
    (\instr \neq \callCmd
      \text{ and }
    \instr \neq \returnCmd)
    \text{ or }
    \\
    (\instr = \callCmd\tup{\procid_0} \text{ and } \com{\tcenv}\vdash\com{C\listsep\tup{\procid_0,0,\emptyset}} : \com{same})
    \text{ or }
    \\
    (\instr = \returnCmd \text{ and } \com{\tcenv}\vdash\com{C} : \com{same})
  }{
    \com{ \acteval{\tcenv}{\codeenv}{\stt}{\nil}{\stt'} }
  }{act-no}
  \typerule{Action-Call}{
    \fun{instr}{\codeenv,\stt} = \callCmd\tup{\procid_0}
    &
    \com{\peval{\stt}{\stt'}}
    \\
    \com{\stt}=\com{\tup{C,M,G,S}}
    &
    \com{\tcenv}\vdash\com{C\listsep\tup{\procid_0,0,\emptyset}} : \com{?}
  }{
    \com{ \acteval{\tcenv}{\codeenv}{\stt}{\trcl{\procid_0}{M,G}}{\stt'} }
  }{act-call}
  \typerule{Action-Return}{
    \fun{instr}{\codeenv,\stt}= \returnCmd
    &
    \com{\peval{\stt}{\stt'}}
    \\
    \com{\stt}=\com{\tup{C,M,G,S}}
    &
    \com{\tcenv}\vdash\com{C} : \com{!}
  }{
    \com{ \acteval{\tcenv}{\codeenv}{\stt}{\trrt{M,G}}{\stt'} }
  }{act-ret}

  \typerule{Single}{
    \com{ \actevalfat{\tcenv}{\codeenv}{\stt}{\nil}{\stt''}}
    \\
    \com{\stt''} = \com{\st{\procid,\pc,L}\listsep C,M,G,S}
    &
    \com{ \acteval{\tcenv}{\codeenv}{\stt''}{\ac}{\stt'} }
  }{
    \com{ \actevalfat{\tcenv}{\codeenv}{\stt}{\ac}{\stt'} }
  }{sing}

  \typerule{Trace-Both}{
    \com{\treval{\tcenv}{\codeenv}{\stt}{\trc}{\stt''}}
    \\
    \com{ \actevalfat{\tcenv}{\codeenv}{\stt''}{\ac?}{ \stt'''}}
    &
    \com{ \actevalfat{\tcenv}{\codeenv}{\stt'''}{\ac!}{\stt'}}
  }{
    \com{\treval{\tcenv}{\codeenv}{\stt}{\trc\listsep\ac?\listsep\ac!}{\stt'}}
  }{trace-sing2}
  \typerule{Trace-Single}{
    \com{\treval{\tcenv}{\codeenv}{\stt}{\trc}{\stt''}}
    \\
    \com{ \actevalfat{\tcenv}{\codeenv}{\stt''}{\ac?}{ \stt'''}}
    &
    \lnot(\com{ \actevalfat{\tcenv}{\codeenv}{\stt'''}{\ac!}{\stt'}})
  }{
    \com{\treval{\tcenv}{\codeenv}{\stt}{\trc\listsep\ac?}{\stt'}}
  }{trace-sing1}
}{trsem}{
  Trace semantics for \move programs (excerpts).
}
\Cref{tr:act-no} says that no action is produced if the underlying small-step reduction is caused by an instruction that is not a \com{\callCmd}, nor a \com{\returnCmd}, or the jump caused by the \com{\callCmd} or by the \com{\returnCmd} does not cross the boundary between \tc and attacker.
\Cref{tr:act-call} generates a call action in case of the attacker calls a function defined by the \tc while \Cref{tr:act-ret} generates a return action when the \tc returns to the attacker.
\Cref{tr:sing} concatenates a series of empty steps followed by an action as a single action that is then used by \Cref{tr:trace-sing2,tr:trace-sing1} to generate a trace.

Let us now explain which of these rules apply to external code calling the \code{mint} function of \Cref{fig:coin}.
First, \Cref{tr:act-call} is triggered when calling \code{mint}, producing action \com{\trcl{{mint}}{M,G}}.
\Cref{tr:act-no} handles the body of \code{mint} until it returns, where \Cref{tr:act-ret} produces action \com{{\trrt{M',G'}}}.
The globals \com{G'} now contain a new address pointing to a memory location in \com{M'} where the newly-minted coin (the one being allocated and returned on line \Cref{line:inv-rest}) is stored.
All these single actions are concatenated into a trace by \Cref{tr:trace-sing2}.

\smallskip
Given a \tc\ \com{\tcenv} and an attacker \com{\atk}, we indicate the trace \com{\trc} of \com{\tcenv} generated according to the rules above, starting from the starting state as:
\[
  \SInit{\tcenv+\atk}\sem\trc
\]

\subsection{Tools to Attain Robust Safety} %
\label{sec:tools}
Security of the \tc is attained via three tools: the bytecode verifier ensuring all code is well-typed (as presented in \Cref{sec:typ}), a prover that checks whether invariants hold locally for \tc (\Cref{sec:provers}) and an encapsulator ensuring \tc does not leak globals that have an invariant (\Cref{sec:encapsulator}).
This section focusses on the two tools whose job is purely security-oriented, as the goal of the already-presented verifier pertains to more general functional correctness.

As mentioned, the prover and the encapsulator verify two different properties on some \tc\ \com{\tcenv} to asses whether it respects invariants \com{\inv}.
Given an execution state \com{\stt}, the prover checks that the globals \com{G} and the memory \com{M} respect \com{\inv}, we call this the local property (\Cref{tr:weak-inv-inv-ppr}).
\begin{center}
  \typerule{Weak Property - Locality}{
    \com{\stt} = \com{\tup{C,M,G,S}}
    &
    \com{M,G} \vdash \com{\inv}
  }{
    \weakinvinv{\tcenv}{\stt}{\inv}
  }{weak-inv-inv-ppr}
\end{center}
On the other hand, the encapsulator takes a state and checks that the memory and globals that are reachable from the attacker do not intersect (\com{\notcap}) with those with an invariant, we call this the unreachability property (\Cref{tr:weak-inv-unchange-ppr}).
We rely on judgement $\com{\tcenv,\sigma} \vdash \com{M_a}, \com{G_a} : \com{attackerpart}$ to traverse state \com{\stt} and extract the parts of globals (\com{G_a}) and memory (\com{M_a}) that belong to the attacker, i.e., that do not belong to code defined in \com{\tcenv}.
\begin{center}
  \typerule{Weak Property - Unreachability}{
    \com{\stt} = \com{\tup{C,M,G,S}}
    &
    \com{G_{i}} = \restr{\com{G}}{\domG{\inv}}
    &
    \com{M_{i}} = \restr{\com{M}}{G_i}
    \\
    \com{\tcenv,\stt} \vdash \com{M_a}, \com{G_a} : \com{attackerpart}
    \\
    {\com{G_{i}}} \notcap {\com{G_a}}
    &
    \dom{\com{M_{i}}} \notcap \dom{\com{M_a}}
  }{
    \weakinvatk{\tcenv}{\stt}{\inv}
  }{weak-inv-unchange-ppr}
\end{center}

We call these properties weak because them alone are not sufficient to entail security of \tc.
However, a state \com{\stt} that satisfies \emph{both} properties is strong enough to be secure (\Cref{tr:stron-inv}).
\begin{center}
  \typerule{Strong Property}{
    \weakinvinv{\tcenv}{\stt}{\inv}
    &
    \weakinvatk{\tcenv}{\stt}{\inv}
  }{
    \stronginv{\tcenv}{\stt}{\inv}
  }{stron-inv}
\end{center}

These properties are the key to the robust safety theorem (\Cref{thm:rs-sketch} later on), as well as to defining the properties that the prover and the encapsulator must uphold.

\subsubsection{Prover} %
\label{sec:provers}

The prover (\com{\local}) is a tool that statically verifies that some \tc\ \com{\tcenv} satisfies invariants \com{\inv} \emph{locally}.
That is, a programmer can run the prover on her code (and its dependencies) before deploying that code and linking it against attacker code.
We denote the prover running on the \tc\ \com{\tcenv} as: $\com{\local}(\tcenv)$.

Ideally, the prover statically shows that invariants hold in a \emph{closed} world containing a fixed set of modules, but we want to ensure that invariants will continue to hold in an \emph{open} world with arbitrary modules that may be written by attackers.
Informally, we want \tc that has gone through the prover to have this property: if the \tc code starts executing in some state \com{\stt}, then when control is given back to the attacker, the memory and globals there respect the invariant.
Formally, this is denoted with $\localinv{\com{\tcenv}}{\com{\inv}}$, as captured by \Cref{def:loc-inv} below.
Given an execution in \tc that starts from a state satisfying the strong property and producing a visible action, the ending state satisfies the local property.
\begin{definition}[Local Invariant Satisfaction]\label{def:loc-inv}
  \begin{align*}
    \localinv{\com{\tcenv}}{\com{\inv}}
      &
      \isdef\
    \text{ let }
      \com{\stt} = \com{\tup{C,M,G,S}}
    \\
    \text{ if }
      &
      \stronginv{\tcenv}{\stt}{\inv}
    \text{ and }
      \com{\tcenv \triangleright \codeenv,\procid} \vdash \com{\stt \Xto{\ac!} \stt'}
    \\
    \text{ and }
      &
      \local(\tcenv)
    \text{ then }
      \actcheck{\com{\ac!}}{\com{\inv}}
    \text{ and }
      \weakinvinv{\tcenv}{\stt'}{\inv}
  \end{align*}
\end{definition}

What we mentioned this far is an abstract prover \com{\local}.
A concrete prover instance would be the \move Prover~\cite{DBLP:conf/cav/ZhongCQGBPZBD20,prover-new}, which processes a module by assuming global invariants specified by the programmer hold at the entry of each public function and ensuring that they continue to hold at the exit.
The \move Prover translates both invariants and \move bytecode into Boogie~\cite{DBLP:conf/fmco/BarnettCDJL05}, which uses Z3~\cite{DBLP:conf/tacas/MouraB08} to prove that the invariants hold or find a counterexample.
We believe the \move Prover fulfils \Cref{def:loc-inv}, but since it is not the focus of this work, we leave that result for future work.

\subsubsection{Encapsulator} %
\label{sec:encapsulator}

The encapsulator is a static analysis that verifies that no mutable reference to a global is passed to attacker code.
We first reason about an encapsulator (\com{\modul}) as an abstract entity in order to define what property it must uphold (\Cref{def:encapsulated}); we discuss a concrete encapsulator that satisfies this property later in this section.
We denote the encapsulator analysing the \tc\ \com{\tcenv} as: $\com{\modul}(\tcenv)$.
Since we formulate the encapsulator as a static analysis, these are all the parameters it needs, if it were a dynamic analysis we would have to supply runtime states.

Informally, we want \tc that is encapsulated to have this property: if the \tc code starts executing in some state \com{\stt}, then when control is given back to the attacker, she has no access to globals or memory with an invariant.
Formally, this is denoted with $\encapsulated{\com{\tcenv}}{\com{\modul}}{\com{\inv}}$, as captured by \Cref{def:encapsulated}.
Given an execution in \tc that starts from a state satisfying the strong property, the state when control is passed to the attacker satisfies the unreachability property.
\begin{definition}[Encapsulated Code Satisfaction]\label{def:encapsulated}
  \begin{align*}
    \encapsulated{\com{\tcenv}}{\com{\modul}}{\com{\inv}}
          &\isdef\
          \text{ let }
            \com{\stt} = \com{\tup{C,M,G,S}}
          \\
          \text{ if }
            &
          \stronginv{\tcenv}{\stt}{\inv}
          \text{ and }
            \com{\tcenv \triangleright \codeenv,\procid} \vdash \com{\stt \Xto{\ac!} \stt'}
          \\
          \text{ and }
            &
            \com{\modul}(\restr{\tcenv}{\inv})
          \text{ then }
        \weakinvatk{\tcenv}{\stt'}{\inv}
  \end{align*}
\end{definition}
Here we restrict the encapsulator to only run on the subset of \com{\tcenv} that is invariant-defined (i.e., \com{\restr{\tcenv}{\inv}}) since it is sometimes the case that only part of the codebase needs to be encapsulated, as we discuss later in \Cref{sec:prec}.

We now describe a concrete encapsulator, denoted with \com{\ea}, that satisfies \Cref{def:encapsulated} under the immutably attacker of \Cref{sec:attacker} (\Cref{sec:conc-ea}).
This is fulfilled in practice, since currently, most \move programs are smart contracts deployed on blockchains.
Afterwards, we describe a slightly-different encapsulator, denoted with \com{\enb}, that satisfies \Cref{def:encapsulated} with respect to the mutable attacker of \Cref{sec:attacker} (\Cref{sec:non-batt}).

\paragraph{A Concrete Encapsulator for Blockchain \move Code}\label{sec:conc-ea}
\com{\ea} is a static intraprocedural escape analysis that formalises the intuition presented in \Cref{sec:why-rs}.
The analysis abstracts the concrete values bound to local variables and stack locations using a lattice with three abstract values: $\noref$, $\okref$, $\inref$.
We indicate abstract values as \com{\hat{\val}} and abstract locals (resp. globals) as \com{\hat{L}} (resp. \com{\hat{S}}).
The lattice ordering is $\noref \sqsubseteq \inref$ and $\okref \sqsubseteq \inref$.
Intuitively, $\noref$ represents any non-reference value, $\okref$ represents a reference that does not point to resource defined in \tc, and $\inref$ represents a reference that \emph{may} point to a resource defined in \tc.
The goal of the analysis is to prevent an $\inref$ from ``leaking'' to a caller of the \tc via a $\returnCmd$.
Since \move records cannot store references, this is the only way such leaks occur.

Applying \com{\ea} to a module \com{\codeenv} (still denoted as \com{\ea(\codeenv)}) makes \com{\ea} traverse all the functions in the module of interest, and in each function it verifies all instructions that make up their bodies (\Cref{fig:enc}).
Formally, the analysis follows this judgement \com{\aeval{ \codeenv, \procid, \inv, \instr }{ \st{\hat{L},\hat{S}} }{ \st{\hat{L'}, \hat{S'}} }}, which reads \emph{``under invariant \com{\inv}, instruction \com{\instr} (in function \com{\procid}, in module \com{\codeenv}) consumes abstract locals \com{\hat{L}} and abstract globals \com{\hat{S}} and produces \com{\hat{L'}} and \com{\hat{S'}}''}.
\myfig{
  \centering
    \typerule{\com{\ea}-BorrowFld-Relevant}{
    f \in \inv
  }{
    \aeval{
      \codeenv, \procid, \inv,
      \ao{\borrowfieldCmd}{\fldname}
    }{
      \st{\hat{L},\stackht{\hat{\val}}{\hat{S}}}
    }{
      \st{\hat{L}, \stackht{\inref}{\hat{S}}}
    }
  }{encapsulator-borrowfld-attacker}
  \typerule{\com{\ea}-BorrowFld-Irrelevant}{
    f \notin \inv
  }{
    \aeval{
      \codeenv, \procid, \inv,
      \ao{\borrowfieldCmd}{\fldname}
    }{
      \st{\hat{L},\stackht{\hat{\val}}{\hat{S}}}
    }{
      \st{\hat{L}, \stackht{\hat{\val}}{\hat{S}}}
    }
  }{encapsulator-borrowfld}
  \typerule{\com{\ea}-BorrowGlobal}{
  }{
    \aeval{
      \codeenv, \procid, \inv, \ao{\borrowglobalCmd}{\sname}
    }{
      \st{\hat{L}, \stackht{\hat{v}}{\hat{S}}}
    }{
      \st{\hat{L}, \stackht{\inref}{\hat{S}}}
    }
  }{encapsulator-borrowglobal}
  \typerule{\com{\ea}-BorrowLoc}{
  }{
    \aeval{
      \codeenv, \procid, \inv, \ao{\borrowlocCmd}{x}
    }{
      \st{\hat{L},\hat{S}}
    }{
      \st{\hat{L},\stackht{\okref}{\hat{S}}}
    }
  }{encapsulator-borrowloc}
  \typerule{\com{\ea}-Return}{
    \card{\codeenv(P)\dotretty} = n
    &
    \forall i \in 1..n\ldotp
    \hat{v_i} \neq \inref
  }{
    \aeval{
      \codeenv, \procid, \inv, \returnCmd
    }{
      \st{\hat{L}, \stackht{\hat{v_1}\listsep\hat{v_n}}{\hat{S}}}
    }{
      \st{\hat{L}, \stackht{\hat{v_1}\listsep\hat{v_n}}{\hat{S}}}
    }
  }{encapsulator-return}
}{enc}{
  \com{\ea} escape analysis (excerpts).
}

\Cref{tr:encapsulator-borrowfld-attacker} states that when borrowing a field that has an invariant on, it may point to a resource defined in \tc and thus \com{\inref}.
\Cref{tr:encapsulator-borrowfld} propagates the abstract values when the field has no invariant on itself.
\Cref{tr:encapsulator-borrowglobal} applies to globals, since one such reference should never be leaked, the borrowed global is \com{\inref}.
\Cref{tr:encapsulator-borrowloc}, on the other side, applies to locals, which cannot outlive the current function, so any value retrieved this way is \com{\okref}.
Crucially, \com{\returnCmd} cannot return \com{\inref} (\Cref{tr:encapsulator-return}).
Finally, as mentioned, the analysis is intraprocedural, so we conservatively assume that each value returned by a call is the join of all function inputs of that call--i.e., if any function input is \com{\inref}, we assume any function output is alos \com{\inref}.

As mentioned, we have proven that \com{\ea} is a valid encapsulator, i.e., it satisfies \Cref{def:encapsulated} (as captured by \Cref{thm:escape-ok}).
We describe the implementation of \com{\ea} in \Cref{sec:eval}.
\begin{theorem}[\com{\ea} is a Valid Encapsulator]\label{thm:escape-ok}
  \begin{align*}
    \encapsulated{\com{\tcenv}}{\com{\ea}}{\com{\inv}}
  \end{align*}
\end{theorem}
Intuitively, this holds because code that is encapsulated with \com{\ea} cannot leak references to globals with invariants to attacker code.
The reason is that the only way to leak those references is through returns, and \Cref{tr:encapsulator-return} prevents that so long as the reference is \com{\inref}.
To ensure that any reference to a global with invariants that we load in the stack is \com{\inref}, \com{\ea} ensures that any way to load those references tags them as \com{\inref}.
This is exactly what \Cref{tr:encapsulator-borrowglobal,tr:encapsulator-borrowfld-attacker} do.

Technically, the statement of \Cref{thm:escape-ok} contains concrete states, yet applying \com{\ea} (i.e., $\com{\ea}(\restr{\tcenv}{\inv})$) operates on abstract ones.
To connect the two, we rely on two functions: $\concfun{\hat{L},\hat{G},\tcenv}$ and $\totypes{\val}$.
The first is a concretisation function that returns all possible concrete states whose locals and globals match their abstract counterparts.
The latter is an abstraction function used to generate abstract locals and globals by abstracting any value \com{\val} contained in their concrete counterparts.

\paragraph{A Concrete Encapsulator for Mutable Attackers}\label{sec:non-batt}
We now discuss how to devise an escape analysis that lets us attain robust safety in the case of a mutable attacker, as defined in \Cref{sec:attacker}.
Recall that the mutable attacker consists of code that cannot be verified, and whose functions can be called by the \tc.
Thus, to define \com{\enb}, the only change we need to introduce is that calls cannot send \com{\inref} on function calls to functions defined outside the \tc (\Cref{tr:encapsulator-nb-call}).
In the same rule, the expected returned values are obtained from calculating the abstract value from the concrete types returned by the function (\com{\codeenv(P)\dotretty}).
This is done via function \totypes{\cdot}, which maps non-reference types to \com{\noref} and reference types to \com{\okref}.
Since these values are returned by an attacker, they cannot be \com{\inref}.
\myfig{
  \typerule{\com{\enb}-Call}{
    \card{\codeenv(\procid_0).type} = n
    &
    \forall i \in 1..n\ldotp
    \hat{v_i^a} \neq \inref
    \\
    \hat{v_1^r}\cdots\hat{v_j^r} = \totypes{\codeenv(P)\dotretty}
    &
    \procid_0\notin\tcenv
  }{
    \aeval{
      \codeenv, \procid, \ao{\callCmd}{\procid_0}
    }{
       \st{\hat{L},\stackht{\hat{v_1^a}\cdots\hat{v_n^a}}{\hat{S}}}
    }{
      \st{\hat{L},\stackht{\hat{v_1^r}\cdots\hat{v_j^r}}{\hat{S}}}
    }
  }{encapsulator-nb-call}
}{enc-enb}{
  \com{\enb} escape analysis (excerpts).
}

Similarly to \com{\ea}, we have proven that \com{\enb} is a valid encapsulator (as captured by \Cref{thm:escape-ok-nb}).
\begin{theorem}[\com{\enb} is a Valid Encapsulator]\label{thm:escape-ok-nb}
  \begin{align*}
  \encapsulated{\com{\tcenv}}{\com{\enb}}{\com{\inv}}
  \end{align*}
\end{theorem}

\subsection{Robust Safety for \move} %
\label{sec:rs-ppr}

We now have all the technical setup to state (\Cref{def:rs}) and prove robust safety for \tc (\Cref{thm:rs-sketch}).
After presenting and discussing the definition, we analyse the robust aspect from a security perspective (\Cref{sec:rob}) and compare ours to existing robust safety definitions and proofs (\Cref{sec:compare-rs}).

Any \tc\ \com{\tcenv} that is verified (in the sense of \Cref{sec:typ}), proved (in the sense of \Cref{sec:provers}) and encapsulated from \com{\modul} (in the sense of \Cref{sec:encapsulator}), can interact with \emph{any} attacker \com{\atk} and its invariants \com{\inv} cannot be violated.
This means that the \tc is robustly safe, and it is indicated as $\rs{\tcenv}{\inv,\modul,\local}$.
\begin{definition}[Robust Safety for \move]\label{def:rs}
  \begin{align*}
    \rs{\tcenv}{\inv,\modul,\local}
    \isdef
    &
      \forall \com{\atk}
      \ldotp
    \text{ if }
      \com{\tcenv} \vdash \com{\atk} : \com{atk}
    \text{ and }
      \agree{\tcenv}{\inv}
    \\
    \text{ and }
      &
      \wt{}{\com{\tcenv}}
    \text{ and }
      \localinv{\com{\tcenv}}{\com{\inv}}
    \\
    \text{ and }
      &
      \encapsulated{\com{\tcenv}}{\com{\modul}}{\com{\inv}}
    \text{ and }
      \SInit{\tcenv+\atk}\sem\trc
    \\
      \text{ then }
      &
      \tracecheck{\com{\trc}}{\com{\inv}}
  \end{align*}
\end{definition}
The first premise of the definition ensures that only valid attackers are considered while the second ensures that the invariants are specified for the \tc.
The third, fourth and fifth premise ensure that the \tc is verified, proved by \com{\local} and encapsulated by \com{\modul}.
Note that the result is general, no matter what prover and encapsulator are used, so long as those prover and encapsulators are valid in the sense of \Cref{def:loc-inv,def:encapsulated}.
The final premise introduces the trace yielded by the interaction of the \tc with the attacker and the conclusion of the definition confirms that the trace does not violate the invariants.

\begin{theorem}[\move Modules are Robustly-Safe]\label{thm:rs-sketch}
  \begin{align*}
    &
    \rs{\tcenv}{\inv,\modul,\local}
  \end{align*}
\end{theorem}

\subsubsection{Robustness}
\label{sec:rob}
What makes \Cref{thm:rs-sketch} relevant for security is the universal quantification over attackers \com{\atk}, since that differentiates \emph{robust} safety from closed-world safety.
That universal quantification ensures that we are considering \emph{arbitrary} code as attacker, so that includes e.g., the code of \Cref{fig:attacker}.
Crucially, that code needs not be present at verification time.
This is particularly relevant for blockchain-deployed \move code, since attacker code is written (and published) \emph{temporally after} the \tc.
Closed-world safety requires the entire codebase for verification, which is not possible for an evolving system such as a public blockchain.

\subsubsection{Comparison with Existing Robust Safety Statements and Proofs}\label{sec:compare-rs}
To the best of our knowledge, \move is the first large language with tools that provide robust safety, early work on robust safety~\cite{sec-typ-prot,autysec,refty-sec-impl,tydisa,cca,ot4jc,catalin-rs} focussed on formal calculi without a real-world implementation.
Moreover, existing work on robust safety ties the robust safety result to the verification of \tc with a \emph{specific} tool.
In the more modern robust safety results~\cite{dg-rs,davidcaps} (which apply to calculi with complex features not amenable to smart-contract programming), this tool comes in the form of a semantic type systems built on top of the Iris separation logic~\cite{iris} -- a non-trivial tool to understand and master.

In contrast, our definition of robust safety identifies the requirements of \emph{multiple} tools, and singles out the \emph{individual guarantees} of each one tool.
By performing such separation of requirements, our work requires less complexity from each tool individually, which can build upon simple, well-established techniques (such as the escape analysis).
Moreover, as tools (and languages) evolve, this separation makes it simpler to provide new tools for robust safety, maximising proof reuse.
For example, if a new version of \move relies on a different logic of invariants, it is sufficient to change the \move Prover, and thus re-prove it fulfils \Cref{def:encapsulated}, without changing the escape analysis, nor its proofs.
Finally, the proof reuse makes is easier to prove robust safety for different attacker models, as we do with the \com{\ea} and \com{\enb} escape analyses.

\section{Evaluation}\label{sec:eval}

In this section, we present an implementation of the escape analysis \com{\ea} described in \Cref{sec:encapsulator} (\Cref{sec:impl}).
Then, we measure its performance on a large set of \move benchmarks (\Cref{sec:bench}).

We evaluate the \com{\ea} analysis according to two criteria:
\begin{itemize}
\item[\mf{Performance.}] We claim the escape analysis is fast: it adds negligible overhead over the companion verifier (\Cref{sec:perf});
\item[\mf{Precision.}] We claim the escape analysis is precise: it rarely flags \move code that is robustly safe w.r.t its safety invariants (\Cref{sec:prec}).
\end{itemize}

\subsection{Implementation}\label{sec:impl}
We have implemented the escape analysis in approximately 300 lines of Rust code on top of the \move Prover analysis framework.%
\footnote{%
  Our escape analysis is available at: \url{https://github.com/diem/diem/blob/03c30e1/language/move-prover/bytecode/src/escape_analysis.rs}
}
The framework has libraries for parsing \move bytecode, control-flow graph construction, and fixed point computation that are not included in the total above.

We use the \move Prover specification language\footnote{\url{https://github.com/diem/diem/blob/03c30e1/language/move-prover/doc/user/spec-lang.md}} as the invariant language $\inv$.
This specification language lets programmers write source code invariants similar to our example in \Cref{sec:example}.
The invariants are converted to SMT and checked by the \move Prover against the compiled Move bytecode.

Unlike our minimalistic formalism in \Cref{sec:formal-rs}, the Move bytecode languages distinguishes between mutable (\code{&mut T}) and immutable (\code{&T}) references.
A mutable reference can be either written or read, whereas an immutable reference can only be read.
In the public blockchain \move code we consider in our evaluation, attacker-controlled reads are not concerning because the entire blockchain state is world-readable by external users for auditing.
Thus, our implementation only flags functions that may return a \emph{mutable} reference to a field involved in a prover \code{spec} for the module under analysis.
If the module does not have any invariant, we conservatively flag all such functions.

\subsection{Benchmarks}\label{sec:bench}
We ran our analysis on the benchmarks shown in \Cref{fig:results}.
The benchmarks fall roughly into three categories: blockchain management logic implemented in \move (starcoin, diem, bridge), utility libraries (taohe, stdlib), and applications (mai, blackhole, alma, starswap, meteor).
All of the benchmarks contains some \move Prover specs, though we note that not all modules have specs and the density of specification varies across benchmarks.
Although our benchmark set is small, it represents a substantial fraction of the publicly available \move code on GitHub.
Move is a young language that it only beginning to gain transaction, so these benchmarks contain a representative sample of production \move code.

\begin{figure}[!htb]\small
  \begin{tabular}{l r r | r | r || r || r | r}
  \textbf{Bench} & \textbf{Mod} & \textbf{Fun} & \textbf{Rec} & \textbf{Instr} & \textbf{Err} & \textbf{T$_p$} & \textbf{T$_e$} \\
  \hline
  starcoin & 60 & 431 & 88 & 8243 & 2 & 3178 & 10  \\
  diem & 13 & 102 & 19 & 1830 & 0 & 1651 & 1  \\
  mai & 45 & 411 & 77 & 7881 & 0 & 4209 & 12  \\
  bridge & 36 & 352 & 85 & 8060 & 0 & 2428 & 8  \\
  blackhole & 36 & 324 & 72 & 6030 & 0 & 2289 & 7  \\
  alma & 35 & 333 & 67 & 6318 & 0 & 2102 & 8  \\
  starswap & 33 & 335 & 67 & 6617 & 0 & 14993 & 7  \\
  meteor & 32 & 323 & 69 & 5981 & 0 & 1641 & 7  \\
  taohe & 11 & 40 & 7 & 305 & 0 & 1022 & 1  \\
  stdlib & 9 & 66 & 5 & 933 & 1 & 1151 & 1  \\
  \hline
  \textbf{Total} & 310 & 2717 & 556 & 52198 & 3 & 34664 & 62  \\
  \end{tabular}

  \caption{
    Checking for robust safety with the escape analysis encapsulator.
    The \textbf{Mod}, \textbf{Fun}, \textbf{Rec}, and \textbf{Instr} columns show the number of modules, declared functions, declared record types, and bytecode instructions in each project and its dependencies. The \textbf{Err} column shows the number of functions flagged by the escape analysis.
    The \textbf{T$_p$} and \textbf{T$_e$} columns show the time taken to run the \textit{\textbf{p}}rover and \textit{\textbf{e}}scape analysis in milliseconds on a 2.4 GHz Intel Core i9 laptop with 64GB RAM.
  }
  \label{fig:results}
\end{figure}

\subsubsection{Evaluating Performance}\label{sec:perf}
The results in \Cref{fig:results} support our claim that the analysis is fast; it takes well under a second on all benchmarks and under 10ms on most benchmarks.
As we can see, this time is a tiny fraction of the time taken to run the prover (which is several orders of magnitude slower) on each benchmark.

Thus, we can use the escape analysis to strengthen the prover's closed-world guarantees to open-world ones with no user-visible performance degradation.
In the future, we plan to do this by incorporating the escape analysis into the prover's pipeline of pre-analyses (e.g., liveness analysis, invariant instrumentation).
This will improve performance of the escape analysis even more by sharing the steps of parsing bytecodes and building control-flow graphs with the other analyses in the pipeline.
Anecdotally, these steps take roughly half of the escape analysis running time.

\subsubsection{Evaluating Precision}\label{sec:prec}
The results also support our claim that the analysis is precise.
Only three functions (0.1\% of the total analysed) in three distinct modules (0.9\% of the total analysed) were flagged as potentially containing robust safety violations.

We manually investigated each finding to determine whether it indicated a genuine robust safety issue.
In all three cases, the function \emph{does} leak a mutable reference to module-internal state, but the reference cannot point into memory used by the module's invariants (and thus, all three are false positives).

The following code captures the essence of two reports from \texttt{starcoin} and \texttt{stdlib} (which both contain variants of the \code{Option} module).
\begin{lstlisting}[numbers=none]
module 0x1::OptionVariant {
  struct Option<T> { v: vector<T> }

  // Typically, Options are defined as None | Some (x)
  // Move does not have sum types, so encode None as an
  // empty vector and Some(x) as a vector of length 1 containing x
  spec Option { invariant len(v) <= 1; }

  // False positive flagged by analysis as unsafe, but safe
  public fun get_mut<T>(t: &mut Option<T>): &mut T {
    Vector::borrow_mut(&mut t.v, 0)
  }
}
\end{lstlisting}
In this code, the \code{get_mut} function does indeed leak an internal reference, but this is intentional--\code{Option} is a collection intended to be instantiated by clients who need to mutate the contents of the \code{Option} in-place using this function.
The analysis sees that the invariant contains the field \code{v} and conservatively reports leaks not only of \code{v}, but also of references that extend from \code{v} (\code{&v[0]}, in this case).
We note that it \emph{would} be a robust safety violation to leak a reference to \code{Option.v}, since an attacker could use this reference to violate the invariant \code{len(v) <= 1} (e.g., by adding extra elements to the vector).

The third case (from \texttt{starcoin}) is somewhat similar: a module implementing a collection type leaks a reference to its internal vector, but does so intentionally to allow client modules to add elements to the vector.
\begin{lstlisting}[numbers=none]
module 0x1::OwnedVector {
  struct OwnedVec<T> { v: vector<T>, owner: address }

  // flagged by analysis
  public fun get_mut<T>(c: &mut OwnedVec<T>, i: u64): &mut T {
    &mut c.v
  }
}
\end{lstlisting}
This module does not contain any invariant, so the analysis conservatively flags all leaks of internal references.

\paragraph{Discussion}
These examples demonstrate an interesting and perhaps counter-intuitive point--although encapsulation is generally a good idea, it is not desirable to fully encapsulate \emph{all} modules.
Modules like \code{OptionVariant} and \code{OwnedVector} are utility modules that are intended to be specialized by clients who need the flexibility to write the internal state of these modules.
For example, clients of \code{OwnedVector} would not be able to add/remove elements from the vector without the \code{get_mut} function.
Thus, although it is tempting to suggest integrating the escape analysis into \move's bytecode verifier (and thus make \emph{all} Move code robustly safe by construction), there is evidence that this would remove expressivity used by real Move programmers.

We believe it would be possible to eliminate the false positive in the \code{OptionVariant} example by using a more sophisticated abstract domain that tracks the set of \emph{access paths}~\cite{Jones-al:POPL79} associated with each reference.
The analysis could compare the leaked access paths to the access paths mentioned in specifications and only complain if there is an overlap between the paths \emph{or} their possible suffixes.
For example, the analysis could determine that the \code{get_mut} function leaks the path \code{Option.v[0]}, but the specification only mentions the incomparable path \code{Option.v.length}.
The analysis would also need to flag a leak of a path like \code{Option.v[0]} if the specification mentions a prefix of the path (e.g., \code{invariant v == vec[1,2]}).

However, this analysis would be somewhat more complex than our straightforward three-value abstract domain; e.g., we would need to introduce a widening operator~\cite{10.1145/512950.512973} because the access path domain is not finite height.
Furthermore, our results suggest that the precision gain from this improvement would be fairly small because the existing analysis is already quite precise in practice.

Finally, we note that a simpler analysis that flags any return of a (mutable) reference would be too restrictive to be practical. Returning references is common in real-world Move code because of the ubiquity of non-copyable types.

\subsection{Security Conclusions}\label{sec:evalconc}
Thus, \emph{all} of the \move modules we looked are robustly safe w.r.t their specified invariants, and the \move Prover augmented with our escape analysis can automatically prove this for $>$99\% of the modules.
This indicates that language-supported robust safety is indeed a practically achievable goal for \move programmers.

\section{Related Work} %
\label{sec:rw}

\paragraph{Smart Contract Languages}

Ensuring that key safety invariants hold even in the presence of attackers is a challenging and important task for all smart contract programmers.
The Solidity~\cite{solidity} source language and its executable Ethereum Virtual Machine (EVM)~\cite{evm} bytecode language are the most popular smart contract languages and have been studied the most extensively with security in mind .

The primary barrier to writing encapsulated code in these languages is dynamic dispatch.
When the target of a callback is determined by the contract \code{C}'s caller (which is common, e.g., every payment operation fits this pattern), the contract author cannot know statically how it will change the global state.
This is particularly pernicious when the callback supplied by the caller is \emph{re-entrant}--that is, the callback invokes one or more functions from \code{C}.
Attackers can leverage this to change the state of \code{C} in ways that the author did not anticipate. and/or to observe (and exploit, by injecting code via dynamic dispatch) the interval when a key safety invariant is violated.
For example, in the DAO~\cite{re_dao} attack, the vulnerable contract made a dynamic dispatch call while a key conservation invariant is violated, which the attacker leveraged to steal funds from the contract.

Many approaches to mitigating re-entrancy have been proposed, including design patterns~\cite{sc_best_practices}, dynamic analysis~\cite{DBLP:journals/pacmpl/GrossmanAGMRSZ18}, and static analysis~\cite{DBLP:journals/pacmpl/AlbertGRRRS20}.
Although absence of re-entrancy facilitates proving robust safety, we are not aware of any work that attempts to define robust safety of EVM code, or any tools that can prove EVM code robustly safe.

The problem of ensuring robust safety is quite different in \move and the EVM.
While \move does not have dynamic dispatch (and thus also does not have re-entrancy), it does have mutable references that can escape from the module that created them.
In the EVM, references are represented as indexes into a sequential memory that is only accessible by a single contract, so they cannot escape.
We note that precisely and efficiently verifying the absence of re-entrancy for EVM code is challenging, whereas our escape analysis for verifying the absence of leaked mutable references is precise, efficient, and relatively straightforward.

Scilla~\cite{DBLP:journals/pacmpl/SergeyNJ0TH19} is a newer smart contract language that shares some design goals with \move.
Scilla restricts dynamic dispatch by requiring a dynamic function call to be the last instruction in a procedure, which largely mitigates the re-entrancy issues afflicting the EVM.
Scilla was also designed to support automated static analysis; its toolchain includes an abstract interpretation framework that supports both built-in (e.g., determining where monetary values can flow) and user-defined analyses.
Finally, recent work~\cite{DBLP:journals/corr/abs-2109-06557} proposes a proof methodology and programming discipline for addressing the challenges of verifying class invariants in the presence of common smart contract features such as callbacks.
The emphasis on avoiding ``reference leaks'' is directly related to the properties enforced by our escape analysis.

\paragraph{Language Design for Isolation}
Language design to support safe interaction with untrusted code is not unique to smart contract languages or \move.
Typed assembly language~\cite{typed_assembly} and capability machines like CHERI~\cite{DBLP:conf/sp/WatsonWNMACDDGL15} are low-level approaches to isolating memory from untrusted code running in the same process.
The Singularity OS project~\cite{singularity} used the type system of Sing\# (a variant of C\#) to enforce strong ownership of memory that crosses trust boundaries.
The Joe-E~\cite{DBLP:conf/ndss/MettlerWC10} language defines a secure subset of Java to enable capability-based programming patterns.

WASM~\cite{DBLP:conf/pldi/HaasRSTHGWZB17} is designed to isolate untrusted applications from the trusted host system.
Recent work on WASM has studied a variety of mechanisms~\cite{DBLP:conf/pldi/HaasRSTHGWZB17,DBLP:conf/isca/DisselkoenRWGLS19} (e.g., static and dynamic checks) to enhance the system with the ability to isolate untrusted applications from each other as well as from the host.

Broadly speaking, an important difference between these previous systems and \move is that they do not satisfy key requirements for smart contract programming such as determinism, metered execution, and first-class currency.
In addition, these systems are typically concerned with low-level isolation to ensure generic safety properties (such as memory safety) rather than enforcing application-specific, programmer-specified properties like the those specified/verified using the \move Prover toolchain.

\paragraph{Robust Safety}

The robust safety property originated in the context of modular model checking~\cite{rs-orig} and has then been widely applied to reason about security protocols that interact with adversaries~\cite{sec-typ-prot,autysec,refty-sec-impl,tydisa,cca,ot4jc,catalin-rs}.
In this setting, security protocols are written in concurrent languages (often process calculi) and given a type system that enforces robust safety and therefore ensures safe interaction with untyped adversaries.
The type system of these works is analogous to the encapsulator of this paper: it is a static analysis whose goal is to prevent leaks.
In order to model the security invariants, some of these languages have explicit assertions, which are proven to never fail because of robust safety.

\citet{davidcaps}, instead, use robust safety to verify object capability patterns, a programming pattern that enables programmers to protect the private state of their
objects from corruption by untrusted code~\cite{Miller03capabilitymyths}.
Their language is richer than \move (and thus not amenable for safe smart contract programming) and lets programmers define custom assertions, which robust safety ensures to never fail.
To ensure this, their code is verified with a powerful mechanism built on top of the Iris logic~\cite{iris}.
Their (static) verification step is analogous to our encapsulator but it relies on user-defined logical assertions to describe invariants that are more complex than \move assertions, as the underlying language is also more complex.

\citet{dg-rs} use robust safety to demonstrate the end-to-end security property of sandboxing.
Sandboxing is a common technique that allows trusted and untrusted components to interact safely~\cite{moz,goog}.
This work defines invariants outside of the language, as a system call policy, and robust safety means that any program execution respects the policy.
To enforce robust safety, they rely on a type system: any well-typed program can be linked with untyped code and the resulting program is robustly-safe.

\section{Conclusion} %
\label{sec:conc}
We have formalised robust safety for the \move language and gave a precise characterisation of the security properties needed of the tools used to attain it in practice.
One of these tools is an encapsulator, which ensures no sensitive references are leaked to attacker code.
We have also implemented a valid encapsulator and evaluated its precision and performance on a representative set of \move benchmarks.
Our evaluation confirms that the encapsulator can augment existing tools like the \move prover to enable practical enforcement of robust safety for \move programmers.

\renewenvironment{description}{%
    \begin{itemize}%
}{%
    \end{itemize}%
    \ignorespacesafterend%
}
\newpage
\onecolumn
\appendix

\section{The Move Language}\label{sec:move-app}
This section contains the formalisation we borrow (\Cref{sec:syn-app,sec:typ-app}) and extend (\Cref{sec:sem-app}) from \citet{blackshear2020resources}.

\subsection{Syntax}\label{sec:syn-app}
We rely on a series of notions:
\begin{itemize}
	\item Code environments \com{\codeenv} map module ids \com{m} to modules
	\item Modules are pairs of maps: struct ids to struct definitions and procedure id to procedure definitions
	\item Memories \com{M} are maps from locations \com{\ell} to storable values, which can be addresses, tagged records, bools or nats
	\item Operand stacks \com{S} are lists of values \com{v}
	\item Local stacks \com{L} are lists of mappings from variables \com{x} to locations or references
	\item Call stacks \com{C} are lists of call stack frames
	\item Global storages \com{G} map resource ids to locations \com{\ell}
	\item Program states \com{\sigma} are tuples of a call stack, a memory, a global storage and an operand stack: \com{\st{C,M,G,S}}
\end{itemize}

\paragraph{Instructions Set}
\begin{center}
\begin{tabular}{|ll|}\hline
local var instructions &
	$\ao{\movelocCmd}{x} \mid \ao{\copylocCmd}{\loc} \mid \ao{\storelocCmd}{x} \mid \ao{\borrowlocCmd}{x}$
	\\
ref instructions &
	$\readrefCmd \mid \writerefCmd \mid \ao{\borrowfieldCmd}{f}$
	\\
record instructions &
	$\ao{\packCmd}{s} \mid \ao{\unpackCmd}{s}$
	\\
global instructions &
	$\ao{\movetoCmd}{s} \mid \ao{\movefromCmd}{s} \mid \ao{\borrowglobalCmd}{s} \mid \ao{\existsCmd}{s}$
	\\
stack instructions &
	$\popCmd \mid \ao{\loadconstCmd}{v} \mid \stackopCmd$
	\\
procedure instructions &
	$\ao{\callCmd}{p} \mid \returnCmd \mid \ao{\branchCmd}{\pc} \mid \ao{\branchcondCmd}{\pc}$
	\\
\hline
\end{tabular}
\end{center}

\subsection{Static Semantics}\label{sec:typ-app}
This section reworks the notation of \citet{blackshear2020resources} to our own liking.

We have these notions:
\begin{itemize}
	\item A stack of operand types \com{\typ{S}}, which is a list of types \com{\tau}
	\item A map of local types \com{\typ{L}}, which maps variables \com{x} to types \com{\tau}
\end{itemize}

We indicate that a code environment, a stack of operand types and a map of local types typecheck a state as follows:
\[
	\typerule{Well-typedness of states}{
		\com{\codeenv,\typ{L}} \vdash \com{\st{C,M}} : \com{wtcs}
		\\
		\forall \com{i}\in0..\len{\typ{S}}\ldotp \com{\codeenv,M}\vdash \com{S(i)} : \com{\typ{S(i)}}
		\\
		\forall \com{\st{\alpha,\tau}}\in\dom{\com{G}}\ldotp \com{\codeenv,M}\vdash \com{M(G(\st{\alpha,\tau}))} : \com{\tau}
	}{
		\com{\codeenv,\typ{S},\typ{L}} \vdash \com{\st{C,M,G,S}} : \com{wt}
	}{}
\]

We indicate that a code environment and a map of local types typecheck a call stack and a memory as follows:
\[
	\typerule{Well-typedness of callstacks - base}{}{
		\com{\codeenv,\typ{[]}} \vdash \com{\st{[],M}} : \com{wtvals}
	}{}
\]
\[
	\typerule{Well-typedness of callstacks - ind}{
		\dom{\typ{L}} = \dom{L}
		&
		\forall \com{x}\in\dom{L}\ldotp \com{\codeenv,M}\vdash \com{L(x)} : \com{\typ{L(x)}}
	}{
		\com{\codeenv,\typ{L}} \vdash \com{\st{\st{P,pc,L}\listsep C,M}} : \com{wtcs}
	}{}
\]

We indicate that a code environment and a memory typecheck a value at a type as follows:
\[
	\typerule{}{
		\text{Def 2.2 in \citet{blackshear2020resources}}
	}{
		\com{\codeenv,M}\vdash \com{v} : \com{\tau}
	}{}
\]

\subsubsection{Whole-Program Typing}
We indicate that a procedure \com{P} is well-typed in a module environment \com{\codeenv} as:
\[
	\typerule{}{
		\text{ see \citet{blackshear2020resources} }
	}{
		\wt{\codeenv}{P}
	}{}
\]

The role of \com{P} is to indicate the \code{main} procedure that begins execution of a transaction. The procedure may call procedures of other modules in \com{\codeenv} and/or perform local computation.

\subsection{Dynamic Semantics}\label{sec:sem-app}
This section reworks the notation of the work by \citet{blackshear2020resources} to our own liking and it makes two important changes:
\begin{itemize}
	\item It adds canary values to stacks to help identify what stack is from the component (see later for what we mean) and what is attacker stack.

	Formally:
	\begin{align*}
		S\in\text{OpStk} =
			&\
			S \uplus K
		\\
		K\in\text{Canary} =
			&\
			\text{startfun\ P}
	\end{align*}
	A canary $K$ tells that the next values on the stack belong to function $P$.

	\item It changes the call rule to take a single step.

	The work by \citet{blackshear2020resources} has a big-step reduction for the call to simplify proofs.

	\item It changes the call and return rule to respectively push a canary on the stack and pop it.

	The operand stack is ensured to be split between portions accessible only to the function using that portion, canaries help us understand where the boundaries of each portion lie.
\end{itemize}
The semantic changes are in \Cref{sec:interproc}.

\subsubsection{Local Semantics}
We have a notion of local state \com{\st{M,L,S}} that keeps track of a memory \com{M}, a list of local variables \com{L} and an operand stack \com{S}.
The local semantics tells how an instruction \com{i} makes a local state evolve into the next local state and so it follows this judgement:
\[
	\eval{\st{M,L,S}}{i}{\st{M',L',S'}}
\]

\paragraph{Memories}
Memory locations \com{\loc} come from an infinite, denumerable set of abstract locations \com{\mc{C}}.
The memory \com{M} is parametrised by the set \com{\mc{C}} in order to know where a fresh location should come from, we write that as \com{M_{\mc{C}}}.
For simplicity, we omit the \com{\mc{C}} unless when it is necessary, and simply write \com{M} for \com{M_{\mc{C}}}.

\begin{center}
  \typerule{[MoveLoc]}
  {
      \lread{L}{x} = \loc\sep
      &
      \loc\in\dom{M}
  }{
    \eval{
      \st{M,L,S}
    }{
      \ao{\movelocCmd}{x}
    }{
      \st{\mdel{M}{\loc},\mdel{L}{x},\stackht{\mread{M}{\loc}}{S}}
    }
  }{sem-loc-moveloc}
  \typerule{[MoveLocRef]}
  {
    \lread{L}{x} = \rft{\loc}{p}
  }{
    \eval{
      \st{M,L,S}
    }{
      \ao{\movelocCmd}{x}
    }{
      \st{M,\mdel{L}{x},\stackht{\lread{L}{x}}{S}}
    }
  }{sem-loc-movelocref}
  \typerule{[CopyLoc]}
  {
    \lread{L}{x}=\loc\sep
    &
    \loc\in\dom{M}
  }{
    \eval{
      \st{M,L,S}
    }{
      \ao{\copylocCmd}{x}
    }{
      \st{M,L,\stackht{\mread{M}{\loc}}{S}}
    }
  }{sem-loc-copyloc}
  \typerule{[CopyLocRef]}
  {
      \val=\lread{L}{x} = \rft{\loc}{p}
  }{
    \eval{
      \st{M,L,S}
    }{
      \ao{\copylocCmd}{x}
    }{
      \st{M,L,\stackht{\val}{S}}
    }
  }{sem-loc-copylocref}
  \typerule{[StoreLoc]}
  {
      \com{\val} \in \StoreableValue\sep
      &
      \com{\loc} \notin \com{\dom{M}}\sep
      &
      \com{M'} = \mdel{M}{\lread{L}{x}} \text{ if } \lread{L}{x} \in\dom{M} \text{ else } M\sep
  }{
    \eval{
      \st{M,L,\stackht{\val}{S}}
    }{
      \ao{\storelocCmd}{x}
    }{
      \st{\mset{M'_{\mdel{\mc{C}}{\loc}}}{\loc}{v}, \lset{L}{x}{\loc},S}
    }
  }{sem-loc-storeloc}
  \typerule{[StoreLoc-Ref]} %
  {
      \val \in \Reference
  }{
    \eval{
      \st{M,L,\stackht{\val}{S}}
    }{
      \ao{\storelocCmd}{x}
    }{
      \st{M,\lset{L}{x}{\val},S}
    }
  }{sem-loc-storelocref}
  \typerule{[BorrowLoc]}
  {
      \lread{L}{x}=\loc
  }{
    \eval{
      \st{M,L,S}
    }{
      \ao{\borrowlocCmd}{x}
    }{
      \st{M,L,\stackht{\rft{\loc}{[]}}{S}}
    }
  }{sem-loc-borrowloc}
  \typerule{[Borrowfld]}
  {
    \val = \rft{\loc}{p}\sep
    &
    \loc\in\dom{M}\sep
    &
    \mread{M}{\loc}[p]=\tv{\record{(f,\val_{f}),\cdots}}{t}
  }{
    \eval{
      \st{M,L,\stackht{\val}{S}}
    }{
      \ao{\borrowfieldCmd}{f}
    }{
      \st{M,L,\stackht{\rft{\loc}{\stackht{p}{f}}}{S}}
    }
  }{sem-loc-borrowfld}
  \typerule{[ReadRef]}
  {
    \val = \rft{\loc}{p}\sep
    &
    \loc \in\dom{M}\sep
    &
    \mread{M}{\loc}[p]=\val_p
  }{
    \eval{
      \st{M,L,\stackht{\val}{S}}
    }{
      \readrefCmd
    }{
      \st{M,L,\stackht{\val_p}{S}}
    }
  }{sem-loc-readref}
  \typerule{[Writeref]}
  {
    \val_2 = \rft{\loc}{p}\sep
    &
    v' =\mread{M}{\loc}\sep
  }{
    \eval{
      \st{M,L,\stackht{\val_1}{\stackht{\val_2}{S}}}
    }{
      \writerefCmd
    }{
      \st{\mset{M}{\loc}{v'[p := v_1]},L,S}
    }
  }{sem-loc-writeref}
  \typerule{[Pop]}
  {}{
    \eval{
      \st{M,L,\stackht{v}{S}}
    }{
      \popCmd
    }{
      \st{M,L,S}
    }
  }{sem-loc-pop}
  \typerule{[LoadConst]}
  {
  	v\in \AccountAddress \uplus \mathbb{B} \uplus \mathbb{N}
  }{
    \eval{
      \st{M,L,S}
    }{
      \ao{\loadconstCmd}{v}
    }{
      \st{M,L,\stackht{v}{S}}
    }
  }{sem-loc-loadconst}
  \typerule{[Op]}
  {
  	\com{v''} = \com{[[ v\ \stackopCmd\ v' ]]}
  }{
    \eval{
      \st{M,L,\stackht{v}{\stackht{v'}{S}}}
    }{
      \stackopCmd
    }{
      \st{M,L,\stackht{v''}{S}}
    }
  }{sem-loc-op}
\end{center}

For simplicity we only consider binary ops.

\subsubsection{Global Semantics}
We have a notion of global state \com{\tup{M,G,S}} that keeps track of a memory \com{M}, a list of global variables \com{G} and an operand stack \com{S}.
The global semantics tells how an instruction \com{i} in procedure \com{\procid} makes a global state evolve into the next global one according to a code environment \com{\codeenv}, so it has this judgement:
\[
	\geval{i}{\tup{M,G,S}}{\tup{M',G',S'}}
\]

\begin{center}
  \typerule{[Pack]}
  {
      \tg = gen\_tag(\codeenv(\tup{\procid.mid, s}).kind)\sep
      &
      v = \record{(f_{i},\val_{i})\mid 1 \leq i \leq n}
  }{
    \geval{
      \ao{\packCmd}{s}
    }{
      \tup{M, G, \stackht{\val_{1}\cons\cdots\cons \val_{n}}{S}}
    }{
      \tup{M, G, \stackht{\tv{\val}{\tg}}{S}}
    }
  }{sem-glob-pack}
  \typerule{[Unpack]}
  {
      \val=\tup{\record{(f_{i},\val_{i})\mid 1 \leq i \leq n}, \tg}
  }{
    \geval{
      \ao{\unpackCmd}{s}
    }{
      \tup{M, G, \stackht{\val}{S}}
    }{
      \tup{M, G, \stackht{\val_{1}\cons\cdots\cons \val_{n}}{S}}
    }
  }{sem-glob-unpack}
  \typerule{[MoveFrom]}
  {
      \stype = \tup{\procid.mid, s}\sep
      &
      \gread{G}{\tup{\addr,\stype}} = \loc\sep
      &
      \mread{M}{\loc}=\val
  }{
    \geval{
      \ao{\movefromCmd}{s}
    }{
      \st{M,G,\stackht{\addr}{S}}
    }{
      \st{\mdel{M}{\loc},\gdel{G}{\tup{a,s}},\stackht{\val}{S}}
    }
  }{sem-glob-movefrom}
  \typerule{[MoveTo]}
  {
      \stype = \tup{\procid.mid, s}\sep
      &
      \tup{\addr, \stype)} \notin \dom{G}\sep
      &
      \loc\notin\dom{M}
  }{
    \geval{
      \ao{\movetoCmd}{s}
    }{
      \st{M,G,\stackht{\stackht{\addr}{\val}}{S}}
    }{
      \st{\mset{M_{\mdel{\mc{C}}{\loc}}}{\loc}{\val},\gset{G}{\tup{a,\stype}}{\loc},S}
    }
  }{sem-glob-moveto}
  \typerule{[BorrowGlobal]}
  {
      \stype = \tup{\procid.mid, s}\sep
      &
      \gread{G}{\tup{\addr,\stype}} = \loc
  }{
    \geval{
      \ao{\borrowglobalCmd}{s}
    }{
      \st{M,G,\stackht{\addr}{S}}
    }{
      \st{M,G,\stackht{\rft{\loc}{[]}}{S}}
    }
  }{sem-glob-borrowglobal}
  \end{center}

\subsubsection{Inter-procedural Semantics}\label{sec:interproc}
The inter-procedural semantics relies on the notion of inter-procedural state \com{\st{C,M,G,S}} that extends the global state with a call stack \com{C}.
\[
	\peval{\st{C, M, G, S}}{\st{C', M', G', S'}}
\]

We change the rule for call and edit the rule for return.
\begin{center}
	\typerule{[Call]}{
		\codeenv(\procid_1).\text{code}[\pc_1]= i = \callCmd\tup{\procid_0}
		&
		\fun{sizeof}{S_{args}} = \fun{sizeof}{\codeenv(\procid_0).\text{1}}
	}{
		\callevalmultline{
			\procid_1
		}{
			\st{\stackht{\tup{\procid_1,\pc_1,L_1}}{C}, M, G, \stackht{S_{args}}{S}}
		}{
			\st{\stackht{\tup{\procid_0,0,\emptyset}}{\stackht{\tup{\procid_1,\pc_1,L_1}}{C}}, M, G, \stackht{S_{args}\listsep \text{startfun }\procid_0}{S}}
		}
	}{sem-call}
	\typerule{[Return]}{
		\codeenv(\procid_0).\text{code}[\pc_0]=\returnCmd
	}{
		\callevalmultline{
	      \procid_0
	    }{
	      \st{\stackht{\tup{\procid_0,\pc_0,L_0}}{\stackht{\tup{\procid_1,\pc_1,L_1}}{C}, M,G,S_0\listsep\text{startfun }\procid_0 \listsep S}}
	    }{
	      \st{\stackht{\tup{\procid_1,\pc_1+1,L_1}}{C},M,G,S_0\listsep S}
	    }
	}{sem-ret}
  \typerule{[Step-Loc]}
	{
	    \codeenv(\procid).code[\pc]=i\sep
	    &
	    \eval{\tup{M,L,S}}{i}{\tup{M',L',S'}}
	}{
	  \peval{
	    \st{\stackht{\tup{\procid,\pc,L}}{C}, M,G,S}
	  }{
	    \st{\stackht{\tup{\procid,\pc+1,L'}}{C}, M',G,S'}
	  }
	}{sem-loc}
	\typerule{[Step-Glob]}
	{
	    \codeenv(\procid).code[\pc]=i\sep
	    &
	    \geval{
	      i
	    }{
	      \tup{M,G,S}
	    }{
	      \tup{M',G',S'}
	    }
	}{
	  \peval{
	    \st{\stackht{\tup{\procid,\pc,L}}{C}, M,G,S}
	  }{
	    \st{\stackht{\tup{\procid,\pc+1,L}}{C}, M',G',S'}
	  }
	}{sem-glob}
        \typerule{[Branch]}
	{
	  \codeenv(\procid).code[\pc] = i = \ao{\branchCmd}{\pc'}\sep
	}{
	  \peval{
	    \st{\stackht{\tup{\procid,\pc,L}}{C}, M,G,S}
	  }{
	    \st{\stackht{\tup{\procid,\pc',L}}{C}, M,G, S}
          }
	}{sem-branch}
	\typerule{[BranchTrue]}
	{
	  \codeenv(\procid).code[\pc] = i = \ao{\branchcondCmd}{\pc'}\sep
	  &
	  v = \true
	}{
	  \peval{
	    \st{\stackht{\tup{\procid,\pc,L}}{C}, M,G,\stackht{v}{S}}
	  }{
	    \st{\stackht{\tup{\procid,\pc',L}}{C}, M,G, S}
	  }
	}{sem-ift}
	\typerule{[BranchFalse]}
	{
	  \codeenv(\procid).code[\pc]= i = \ao{\branchcondCmd}{\pc'}\sep
	  &
	  v = \false
	}{
	  \peval{
	    \st{\stackht{\tup{\procid,\pc,L}}{C}, M,G,\stackht{v}{S}}
	  }{
	    \st{\stackht{\tup{\procid,\pc+1,L}}{C}, M,G, S}
	  }
	}{sem-iff}
\end{center}

\subsubsection{Reflexive-Transitive Closure of the Semantics}
The inter-procedural semantics is the top-level one, we indicate its reflexive and transitive closure as:
\[
\smalleval{\procid}{i}{\st{C, M, G, S}}{\st{C', M', G', S'}}
\]

\newpage
\section{Robust Safety Additions}\label{sec:rs-gen}
The overall goal of this section is provide measures to reason locally about a collection of Move modules (formally, a code environment \com{\codeenv}).
The move modules of interest we call:
\[
	\text{\tc}
\]
We want to prove that given some invariants that hold for \tc \com{\codeenv} alone (i.e., locally), we can compose \com{\codeenv} with another module \com{\atk}, get a whole program \com{\codeenv^w}, and state that the same invariants now hold for \com{\codeenv^w}.
To state this theorem (\Cref{sec:rs}) we need to define:
\begin{itemize}
	\item code environment composition $\com{+} : \com{\codeenv}\times\com{\codeenv}\to\com{\codeenv}$ (\Cref{sec:modcomp});
	\item traces \com{\OB{\alpha}} of events that capture what is relevant to be monitored for robust safety (\Cref{sec:traces});
	\item global invariants \com{\inv} that indicate what are conditions that programmers specify on \tc (\Cref{sec:invs});
    \item a local static analysis that can prove that an invariant \com{\inv} holds locally for \tc \com{\codeenv} (\Cref{sec:inv-analysis});
	\item a semantics that produces traces according to the existing small-step semantics (\Cref{sec:trace-sem});
	\item a global static analysis that checks whether an invariant holds for all the actions of a trace (\Cref{sec:global});
	\item an ensapsulator \com{\modul} that proves that \tc respects some coding / structural impositions \com{\vdash\com{\codeenv} : \com{\modul}} (\Cref{sec:modul}).
\end{itemize}

\subsubsection{Blockchain and \tc}\label{sec:tc-ass-atk}
Dealing with robustness in a Blockchain scenario is unlike a standard setting due to a number of assumptions.
All code is on public on the blockchain, so if you deploy some code, you know all that is there, and cannot consider it an attacker.
An attacker is therefore code that is deployed after your code.
Because there is no dynamic dispatch, the control flow that can happen with attackers has a clear structure: attackers can only call you and you never call attackers.

\begin{property}[Attacker code cannot be called]\label{prop:atk-no-call}
  \begin{align*}
    \text{ if }
      &\
      \com{\codeenv} \vdash \com{\atk} : \com{atk}
    \\
    \text{ then } 
      &\
      \nexists \callCmd\tup{\procid}\in\com{\codeenv} \text{ where } \procid\in\com{\atk}
  \end{align*}
\end{property}

\subsection{Module Composition}\label{sec:modcomp}
We indicate the rest of the program that the \tc links against as attackers.
From the formalisation standpoint, attackers are pairs consisting of a code environment and a main function.
With a small abuse of notation we use metavariable \com{\atk} for both an attacker and for just its code environment to differentiate it from the code of interest.
\begin{align*}
	\text{Attackers } \com{\atk} \bnfdef&\ \com{\codeenv, \procid}
\end{align*}

An attacker code environment must be valid with respect to the \tc \com{\codeenv}.
Specifically, all attacker modules must be certified by the Move bytecode verifier and link successfully against \com{\codeenv}.
\[
	\com{\codeenv} \vdash \com{\atk} : \com{atk} \isdef \wt{}{\atk} \text{ and } \fun{funs}{\atk}\notcap\fun{funs}{\codeenv}
\]

\subsubsection{Linking and Starting}

We indicate the code environment resulting from the linking of two code environments as follows:
\[
	\typerule{Link}{
		\com{SD}=\com{\codeenv}.1
		&
		\com{SD'}=\com{\codeenv'}.1
		&
		\com{PD}=\com{\codeenv}.2
		&
		\com{PD'}=\com{\codeenv'}.2
		\\
		\dom{\com{SD}}\cap\dom{\com{SD'}} = \emptyset
		&
		\dom{\com{PD}}\cap\dom{\com{PD'}} = \emptyset
		\\
		\fun{freenames}{\com{SD}} \in \dom{\com{PD}}\cup\dom{\com{PD'}}\cup\dom{\com{SD'}}
		\\
		\fun{freenames}{\com{SD'}} \in \dom{\com{PD}}\cup\dom{\com{PD'}}\cup\dom{\com{SD}}
		\\
		\fun{freenames}{\com{PD}} \in \dom{\com{SD}}\cup\dom{\com{PD'}}\cup\dom{\com{SD'}}
		\\
		\fun{freenames}{\com{PD'}} \in \dom{\com{SD'}}\cup\dom{\com{PD'}}\cup\dom{\com{SD}}
	}{
		\com{\codeenv}+\com{\codeenv'} = \com{\codeenv\listsep\codeenv'}
	}{link}
\]

The domain ($\dom{\cdot}$) of a partial function is the list of \com{RecName} / \com{ProcName} defined for that function.

The \fun{freeenames}{\cdot} of a partial function are the free \com{RecName}s / \com{ProcName}s mentioned in the codomain in that function that are not defined in the domain of the function.

When linking we check that no function nor data structure is defined twice (second line).
Then we check that all free names are defined (last four lines), for a free name to be defined, it needs to be defined in the domain of any of the other definitions.

Note that this linking is purely syntactical, any constraint on the well-typedness of the elements being linked is tested in the generation of the initial state.
The initial state of a code environment and a main procedure is calculated as follows:
\[
	\typerule{Initial State}{}{
		\SInit{\com{\codeenv},\procid} =
			\com{\codeenv,\procid } %
	}{init-state}
\]
Additionally, we need to calculate the initial configuration, which is simply an empty memory, empty globals, the call stack initialised to the \mtt{main} and the default parameter (\com{0}) for main on the operand stack, followed by the canary indicating that \mtt{main} was called.
\[
	\typerule{Initial Configuration}{
		\tup{a,m} \mapsto \tup{\_, \mtt{main}\mapsto \_}\in\com{\codeenv}
		&
		\com{C} = \com{\tup{\tup{\tup{a,m},\mtt{main}},0,\nil}}
	}{
		\CInit{\codeenv} = \com{\tup{C,\nil,\nil,0\listsep\text{startfun }\mtt{main}}}
	}{init-conf}
\]

\subsection{Traces}\label{sec:traces}
We collect traces of events as computation progresses.
Events simply record the safety-relevant bit of the program state.
The goal is to check that at any point in time, events respect invariants.
The notion of `any point in time' in this case is whenever control is passed from \tc
to attacker and vice versa.
We choose this granularity in order to allow invariants to be broken while the attacker is not observing and then reinstated before the attacker starts observing again.
This flexibility is critical for enabling relational invariants over mutable data.
For example, a defender may want to enforce the invariant \code{x == y} in code that updates both \code{x} and \code{y}.
The invariant will be temporarily violated in the time between (e.g.) an update to \code{x} and \code{y}, but this as acceptable as long as the invariant is restored before control is passed to the attacker.
The quantification over all attackers ensures that if invariant-relevant resources are shared with the attacker%
, they will change on the trace, breaking any invariant.

Traces follow this grammar:
\begin{align*}
	\com{\OB{\alpha}} \bnfdef&\ \com{\nil} \mid \com{\OB{\alpha}}\listsep\com{\alpha}
	\\
	\com{\alpha} \bnfdef&\ \com{\trcl{P}{M,G}} \mid \com{\trcb{P}{M,G}} \mid \com{\trrt{M,G}} \mid \com{\trrb{M,G}}
\end{align*}

The key bit that traces record from a security perspective is the globals \com{G} and the memory \com{M} at each boundary crossing.
For proof reasons, the actions indicate more: they tell what kind of instruction generated the action (\callCmd or \returnCmd) and in case of calls, where was the call directed to.
Finally, actions are decorated with \com{?} or \com{!} depending on whether the action originated in the attacker or in the \tc respectively

\subsection{Trace Semantics}\label{sec:trace-sem}
\subsubsection{Cross-Boundary Helpers}\label{sec:cross-help}
We rely on a judgement that, knowing which functions are defined by the \tc, takes a call stack and tells whether the jump between the head and the second element crosses the boundary, and if so, in which direction.
\begin{center}
	\typerule{Cross-?}{
		\com{C} = \com{C'}\listsep\tup{\procid_2,\_,\_}\listsep\tup{\procid_1,\_,\_}
		&
		\com{\codeenv}(\procid_2) = \com{undefined}
		&
		\com{\codeenv}(\procid_1) = \tup{\_,\_,\_}
	}{
		\com{\codeenv}\vdash\com{C} : \com{?}
	}{cross-in}
	\typerule{Cross-!}{
		\com{C} = \com{C'}\listsep\tup{\procid_2,\_,\_}\listsep\tup{\procid_1,\_,\_}
		&
		\com{\codeenv}(\procid_1) = \com{undefined}
		&
		\com{\codeenv}(\procid_2) = \tup{\_,\_,\_}
	}{
		\com{\codeenv}\vdash\com{C} : \com{!}
	}{cross-out}
	\typerule{Cross-no}{
		\com{\codeenv}\nvdash\com{C} : \com{?}
		&
		\com{\codeenv}\nvdash\com{C} : \com{!}
	}{
		\com{\codeenv}\vdash\com{C} : \com{same}
	}{cross-no}
\end{center}

\subsubsection{Trace Semantics Rules}
The trace semantics uses on the operational semantics, and depending on the instruction being executed, it produces an action that is concatenated on a trace.
Most instructions generate a silent action, the only instructions that generate an action are \ao{\callCmd}{\ao{\procid}{\vec{\stype}}} and \returnCmd.
In both cases the generated action is the same, the current globals.
Then, the trace semantics relies on a state that keeps track of an additional element, the \tc, in order to decorate the actions with \com{?} and \com{!}.
\begin{center}
	\typerule{Action-No}{
		&
		\com{\peval{\sigma}{\sigma'}}
		&
		\com{\sigma}=\com{\tup{C,M,G,S}}
		\\
		(i \neq \callCmd
			\text{ and }
		i \neq \returnCmd)
		\text{ or }
		\\
		(i = \callCmd\tup{\procid_0} \text{ and } \com{\codeenv'}\vdash\com{C\listsep\tup{\procid_0,0,\emptyset}} : \com{same})
		\text{ or }
		\\
		(i = \returnCmd \text{ and } \com{\codeenv'}\vdash\com{\procid_0} : \com{same})
	}{
		\com{\codeenv' \triangleright \codeenv,\procid, i} \vdash \com{\sigma \xto{\nil} \sigma'}
	}{act-no}
	\typerule{Action-Call}{
		i = \callCmd\tup{\procid_0}
		&
		\com{\peval{\sigma}{\sigma'}}
		&
		\com{\sigma}=\com{\tup{C,M,G,S}}
		&
		\com{\codeenv'}\vdash\com{C\listsep\tup{\procid_0,0,\emptyset}} : \com{?}
	}{
		\com{\codeenv' \triangleright \codeenv,\procid, i} \vdash \com{\sigma \xto{\trcl{\procid_0}{M,G}} \sigma'}
	}{act-call}
	\typerule{Action-Callback}{
		i = \callCmd\tup{\procid_0}
		&
		\com{\peval{\sigma}{\sigma'}}
		&
		\com{\sigma}=\com{\tup{C,M,G,S}}
		&
		\com{\codeenv'}\vdash\com{C\listsep\tup{\procid_0,0,\emptyset}} : \com{!}
	}{
		\com{\codeenv' \triangleright \codeenv,\procid, i} \vdash \com{\sigma \xto{\trcb{\procid_0}{M,G}} \sigma'}
	}{act-callback}
	\typerule{Action-Return}{
		i = \returnCmd
		&
		\com{\peval{\sigma}{\sigma'}}
		&
		\com{\sigma}=\com{\tup{C,M,G,S}}
		&
		\com{\codeenv'}\vdash\com{C} : \com{!}
	}{
		\com{\codeenv' \triangleright \codeenv,\procid, i} \vdash \com{\sigma \xto{\trrt{M,G}} \sigma'}
	}{act-ret}
	\typerule{Action-Returnback}{
		i = \returnCmd
		&
		\com{\peval{\sigma}{\sigma'}}
		&
		\com{\sigma}=\com{\tup{C,M,G,S}}
		&
		\com{\codeenv'}\vdash\com{C} : \com{?}
	}{
		\com{\codeenv' \triangleright \codeenv,\procid, i} \vdash \com{\sigma \xto{\trrb{M,G}} \sigma'}
	}{act-retback}

	\typerule{Single}{
		\com{\codeenv' \triangleright \codeenv,\procid} \vdash \com{\sigma \Xto{\nil} \sigma''}
		\\
    \com{\sigma''} = \com{\st{\procid,\pc,L}\listsep C,M,G,S}
    &
		\com{\codeenv' \triangleright \codeenv,\procid, \procid(\pc)} \vdash \com{\sigma'' \xto{\alpha} \sigma'}
	}{
		\com{\codeenv' \triangleright \codeenv,\procid} \vdash \com{\sigma \Xto{\alpha} \sigma'}
	}{sing}
	\typerule{Refl}{}{
		\com{\codeenv' \triangleright \codeenv,\procid} \vdash \com{\sigma \Xto{\nil} \sigma}
	}{sing-refl}

	\typerule{Trace-Both}{
		\com{\codeenv' \triangleright \codeenv,\procid} \vdash \com{\sigma \Xtol{\OB{\alpha}} \sigma''}
		\\
		\com{\codeenv' \triangleright \codeenv,\procid} \vdash \com{\sigma'' \Xto{\alpha?} \sigma'''}
		&
		\com{\codeenv' \triangleright \codeenv,\procid} \vdash \com{\sigma''' \Xto{\alpha!} \sigma'}
	}{
		\com{\codeenv' \triangleright \codeenv,\procid} \vdash \com{\sigma \Xtol{\OB{\alpha}\listsep\alpha?\listsep\alpha!} \sigma'}
	}{trace-sing2}
	\typerule{Trace-Single}{
		\com{\codeenv' \triangleright \codeenv,\procid} \vdash \com{\sigma \Xtol{\OB{\alpha}} \sigma''}
		\\
		\com{\codeenv' \triangleright \codeenv,\procid} \vdash \com{\sigma'' \Xto{\alpha?} \sigma'''}
		&
		\lnot(\com{\codeenv' \triangleright \codeenv,\procid} \vdash \com{\sigma''' \Xto{\alpha!} \sigma'})
	}{
		\com{\codeenv' \triangleright \codeenv,\procid} \vdash \com{\sigma \Xtol{\OB{\alpha}\listsep\alpha?} \sigma'}
	}{trace-sing1}
	\typerule{Trace-Refl}{
	}{
		\com{ \treval{\tcenv}{\codeenv,\procid}{\sigma}{\nil}{\sigma} }
	}{trace-re}
\end{center}

\paragraph{Other RS Approaches: Assertions}
An alternative is to enrich the language with precise assertions.
We could both use code-based assertions~\citep{davidcaps} or logical ones that are collected in assume statements~\citep{autysec,tydisa}.
Our approach is similar to the former, but we choose to not introduce an instruction that checks invariants for several reasons.
First, we do not want to modify our language.
Second, we want to show that checks would never fail, so it is safe to skip them.

Thus, including these checks in the operational semantics is unwanted and unnecessary.

\subsection{Invariants}\label{sec:invs}
We indicate invariants with $\inv$.
Invariants contain the list of globals that point to memory locations with a logical invariant.
Then it contains the list of memory locations with a logical invariant, i.e., a map from locations \com{\loc} to conditions that hold on the content of those locations in memory.
We leave the conditions abstract and give an intuition of what they may look like via examples only.

We rely on these functions to manipulate invariants:
\begin{itemize}
	\item invariants are defined for a code environment, which can be extracted from \com{\inv} as follows: $\fun{codeof}{\inv} = \codeenv$.

	A code environment \com{\codeenv} and an invariant \com{\inv} are in agreement if the former is the code of the latter.
	Formally $\agree{\codeenv}{\inv} \isdef \fun{codeof}{\inv}=\codeenv$.

  \item Invariants can refer to fields of records, so $f \in \inv$/$f \notin \inv$ indicates that the invariant $\inv$ does/does not refer to the field $f$.

	\item \domG{\com{\inv}} returns the indices of globals for which invariants are defined.
	That is, this returns the pairs \com{a,\rho} that identify globals for which an invariant is defined.

	\item \domM{\com{\inv}} returns the memory locations reachable from \domG{\com{\inv}}, i.e., the memory locations for which invariants are defined.

	\item \invcond{\inv,M} evaluates the condition for locations \com{M} and returns true (if the condition is satisfied) or false (otherwise).

	\item We can restrict a code environment wrt an invariant as follows \com{\restr{\codeenv}{\inv}} in order to carve out the code environment \com{\codeenv'} that is contained in \com{\codeenv} and that only talks about the code mentioned in \com{\inv}, without any other code that has no invariant on.
	This is used to identify the sub-part of a code environment that needs to be encapsulated (\Cref{sec:modul}).

\end{itemize}
We rely on this property for invariants: none of the types mentioned in \domG{\inv} are attacker types.
\begin{property}[Invariants are not on Attacker Typed Globals]\label{prop:inv-type-atk}
	\begin{align*}
		\forall \st{a,\stype} \in \domG{\inv}, \stype \in \fun{declaredtypes}{\fun{codeof}{\inv}}
	\end{align*}
\end{property}

\subsection{Local Invariant-Checking}\label{sec:inv-analysis}
We introduce a judgement for a static analysis that proves a \tc \com{\codeenv} satisfies an invariant \com{\inv} \emph{locally}:
This relies on the invariants stated in \Cref{sec:invars-states}.

\begin{definition}[Local Invariant]\label{def:loc-inv}
	\begin{align*}
        \localinv{\com{\codeenv'}}{\com{\inv}}
        	&\isdef\
        	\\
        	\text{ let }
        		&
        		\com{\sigma} = \com{\tup{C,M,G,S}}
        	\\
        	\text{ if }
        		&
				\stronginv{\codeenv'}{\sigma}{\inv}
        	\\
        	\text{ and }
        		&
        		\com{\codeenv' \triangleright \codeenv,\procid} \vdash \com{\sigma \Xto{\alpha!} \sigma'}
          \\
          \text{ and }
            &
            \local(\tcenv)
        	\\
        	\text{ then }
        		&
        		\com{\alpha!}\Vdash\com{\inv}
        	\\
        	\text{ and }
        		&
        		\weakinvinv{\codeenv'}{\sigma}{\inv}
	\end{align*}
\end{definition}

\paragraph{Local Checking in Practice}
The verifier checks invariants locally on some \tc using a suitable static analysis tool (e.g., the Move Prover~\cite{DBLP:conf/cav/ZhongCQGBPZBD20}).

\subsection{Global Invariant-Checking}\label{sec:global}
A trace respects an invariant \emph{globally} if all of its actions respect the invariants.
\begin{center}
	\typerule{Global-check-base}{
	}{
        \com{\nil}\Vdash \com{\inv} : \com{global}
    }{trace-glob-base}
	\typerule{Global-check-ind}{
		\com{\alpha}\Vdash\com{\inv}
		&
		\com{\OB{\alpha}} \Vdash \com{\inv} : \com{global}
	}{
        \com{\OB{\alpha}\listsep\alpha}\Vdash \com{\inv} : \com{global}
    }{trace-glob-ind}
\end{center}
An action respects an invariant if all the invariants applied to the action are true.
\begin{center}
	\typerule{Event-check-base}{}{
		\com{\alpha}\Vdash\com{\nil}
	}{event-glob-base}
	\typerule{Event-check-ind}{
		\com{\alpha} = \_ \com{M,G}
		&
		\com{M,G} \vdash \com{\inv}
	}{
		\com{\alpha}\Vdash\com{\inv}
	}{event-glob-ind}
\end{center}

Intuitively, \(\com{M,G} \vdash \com{\inv}\) holds if the part of the memory \com{M} restricted to just the addresses mentioned in \com{G} respects \com{\inv}.
\begin{center}
	\typerule{Invariant Satisfaction}{
		\com{G_i} = \restr{G}{\domG{\inv}}
		&
		\com{M_i} = \restr{M}{G_i}
		&
		\invcond{\inv,M_i} = \com{true}
	}{
		\com{M,G} \vdash \com{\inv}
	}{inv-sat}
\end{center}

\paragraph{Almost-Everywhere}
Having invariants that hold at every line of code is impractical and meaningless: such invariants do not let us express interesting programming patterns.
Instead, invariants should be broken temporarily, so long as they are reinstated before an attacker can notice.
The noticing point is whenever the attacker is executing, so \tc can violate an invariant so long as it is executing, but it will reinstate it before moving control to the attacker.

However, since invariants are only on \tc relevant globals (and memory), we can enforce that they always hold at any step of attacker computation.
The way to do this is the universal quantification over attackers in the classical RS setup (which we follow).

In fact, suppose there is an attacker that breaks an invariant of the \tc.
Since attackers are universally quantified, there also exists the attacker that right after breaking the invariant returns control to the \tc.
The second attacker generates a trace action with its globals and memory that violate the invariant.
However, by definition of RS, that trace does violate the invariant.
Thus, we reached a contradiction and the premise that the attacker could violate the invariant temporarily in its code is therefore wrong.

\subsection{Encapsulator}\label{sec:modul}
We indicate that the code \com{\codeenv} follows the indications of the encapsulator \com{\modul} as follows.
Right now the encapsulator is an abstract entity, whose behaviour we abstract away as $\com{\modul}(\codeenv')$.
We only pass these parameters because this is a static analysis, for a dynamic analysis we could also pass \com{\sigma} and \com{\sigma'}.

\begin{definition}[Encapsulated Code]\label{def:encapsulated}
	\begin{align*}
		\encapsulated{\com{\codeenv'}}{\com{\modul}}{\com{\inv}}
        	&\isdef\
        	\\
        	\text{ let }
        		&
        		\com{\sigma} = \com{\tup{C,M,G,S}}
        	\\
        	\text{ if }
        		&
				\stronginv{\codeenv'}{\sigma}{\inv}
      	  	\\
        	\text{ and }
        		&
        		\com{\codeenv' \triangleright \codeenv,\procid} \vdash \com{\sigma \Xto{\alpha!} \sigma'}
          \\
          \text{ and }
            &
            \com{\modul}(\restr{\codeenv'}{\inv})
        	\\
        	\text{ then }
        		&
				\weakinvatk{\codeenv'}{\sigma}{\inv}
	\end{align*}
\end{definition}
Intuitively, the encapsulator ensures that all writes to fields of the declared types of a given module can only occur inside the proceures of that module.
If this condition holds, the internal state of a module is encapsulated and cannot be mutated by any other module.
Similarly, the internal state of a code environment is encapsulated when no modules outside of the environment can mutate it.
The local invariant checking performed by the \Cref{sec:inv-analysis} is necessary, but not sufficient, to ensure that an invariant holds when a code environment is composed with attacker code.
The encapsulator captures some generic restrictions that ensure the local invariant will continue to hold even in an adversarial setting.

Crucially, we only call the encapsulator on the part of the \tc that is mentioned by the invariants, thus the restriction \com{\restr{\codeenv'}{\inv}}.

\subsection{RS Theorem}\label{sec:rs}
The RS theorem below (\Thmref{thm:rs-sketch}) is stated in a classical way.
It tells us that some \tc \com{\codeenv} of interest, with a starting program \com{\procid}, is \com{RS} wrt an invariant \com{\inv} and an encapsulator \com{\modul} if
	no matter what valid attacker that \tc is linked against,
	so long as the \tc and the attacker are well-typed
	and the \tc respects the invariants locally
	and it is programmed respecting the encapsulator
	then no matter what traces it generates when running alongside the attacker
	the traces never break the invariants.

In the statement below, we have a slight abuse of notation to make the statement simpler, without projections: $\atk$ is both $\codeenv$ and $\procid$.

The goal to prove is the following:
\begin{theorem}[Well-typed modules are Robustly-Safe]\label{thm:rs-sketch}
	\begin{align*}
		\rs{\codeenv}{\inv,\modul,\local} \isdef &\
		\forall \com{\atk},\com{\procid},\com{\OB{\alpha}},\com{\sigma'}\ldotp
		\\
		\text{ if } &
			\com{\codeenv} \vdash \com{\atk} : \com{atk}
		\\
		\text{ and } &
			\agree{\codeenv}{\inv}
		\\
		\text{ and } &
			\wt{}{\codeenv}
		\\
		\text{ and } &
			\localinv{\com{\codeenv}}{\com{\inv}}
                \\
		\text{ and } &
			\encapsulated{\com{\codeenv}}{\com{\modul}}{\com{\inv}}
		\\
		\text{ and } &
			\com{\codeenv \triangleright \SInit{\codeenv+\atk}} \vdash \com{\CInit{\codeenv+\atk} \Xtol{\OB{\alpha}} \sigma'}
			\\
			\text{ then } &
			\com{\OB{\alpha}}\Vdash\com{\inv} : \com{global}
	\end{align*}
\end{theorem}
\begin{proof}[Proof of \Thmref{thm:rs-sketch}]\hfill

	This holds by \Thmref{thm:rs-general} and \Thmref{thm:sat-initial}.
\end{proof}
\BREAK

\subsubsection{Auxiliary Functions}\label{sec:aux-fun}
\paragraph{State Executing Attacker Code}
\begin{center}
	\typerule{Attacker Code Running}{
		\com{\sigma} = \com{\tup{C,M,G,S}}
		&
		\com{C} = \com{\tup{P,pc,L}\listsep C'}
		&
		\com{\codeenv(P)} = \text{undefined}
	}{
		\com{\codeenv}\vdash\com{\sigma}:\com{atkcode}
	}{atkcode}
	\typerule{Module Code Running}{
		\com{\sigma} = \com{\tup{C,M,G,S}}
		&
		\com{C} = \com{\tup{P,pc,L}\listsep C'}
		&
		\com{\codeenv(P)} = \_
	}{
		\com{\codeenv}\vdash\com{\sigma}:\com{modcode}
	}{modcode}
\end{center}
We can tell whether attacker code is running by looking at the top of the stack: if the procedure that executes is not defined in the \tc, then it is the attacker.
Conversely, the \tc is running if its code is on the top of the stack.

\paragraph{Call Stack of the Attacker}
\begin{center}
	\typerule{Attacker Stack - No}{
		\fun{atkstk}{\com{\codeenv},\com{C}} = \com{C'}
		&
		\com{\codeenv(P)} = \_
	}{
		\fun{atkstk}{\com{\codeenv},\com{\tup{P,pc,L}\listsep C}} = \com{C'}
	}{atkstk-n}
	\typerule{Attacker Stack - Yes}{
		\fun{atkstk}{\com{\codeenv},\com{C}} = \com{C'}
		&
		\com{\codeenv(P)} = \text{undefined}
	}{
		\fun{atkstk}{\com{\codeenv},\com{\tup{P,pc,L}\listsep C}}= \com{\tup{P,pc,L}\listsep C'}
	}{atkstk-y}
	\typerule{Attacker Stack - Base}{}{
		\fun{atkstk}{\com{\codeenv},\come}= \come
	}{atkstk-b}
\end{center}
According to the \tc, given a global call stack \com{C}, \com{C'} is the part of the stack that only talks about attacker functions.

\paragraph{Locations in a Call Stack}
We indicate locations as $\com{\loc}$ and lists of locations as $\com{\ell}$.
The function below traverses a call stack and extracts all locations in its locals.
We structure this function to eventually work on lists of values, so we can apply that both to lists of call stacks, to locals and to lists of values.
\begin{center}
	\typerule{Locations-base}{}{
		\fun{locsof}{\nil} = \nil
	}{locsof}
	\typerule{Locations-ind}{
		\fun{locsof}{L} = \com{\ell}
		&
		\fun{locsof}{\com{C}} = \com{\ell'}
	}{
		\fun{locsof}{\com{\tup{L}\listsep C}} = \com{\ell\listsep\ell'}
	}{locsof-i}

	\typerule{Locations-locals-ind}{
	}{
		\fun{locsof}{x\mapsto v \listsep L} = \fun{locsof}{v\listsep \fun{locsof}{\com{L}}}
	}{locsof-l-i-y}

	\typerule{Locations-values-ind yes}{
		\fun{locsof}{\com{V}} = \com{\ell}
	}{
		\fun{locsof}{c\listsep V} = \com{\loc\listsep\ell}
	}{locsof-v-i-y}
	\typerule{Locations-locals-ind - no}{
		\fun{locsof}{\com{V}} = \com{\ell}
		&
		\com{v}\neq\com{\loc}
	}{
		\fun{locsof}{ v \listsep V} = \com{\ell}
	}{locsof-v-i-n}
\end{center}

\paragraph{OpStack of the Attacker}
The function below traverses an operand stack and filters out all the sections that do not belong to attacker code (as defined from the viewpoint of \tc).
\begin{center}
	\typerule{atkops-main}{}{
		\fun{atkops}{\codeenv,S\listsep \text{startfun }\mtt{main}} = \fun{atkops}{\codeenv,S\listsep \text{startfun }\mtt{main},\top}
	}{atkops}

	\typerule{atkops-atk}{
		\codeenv{\procid} = \text{undefined}
	}{
		\fun{atkops}{\codeenv,S\listsep \text{startfun }\procid,\top} = \fun{atkops}{\codeenv,S,\top} \listsep \text{startfun }\procid
	}{atkops-atk}
	\typerule{atkops-atk-val}{
		\codeenv{\procid} = \text{undefined}
	}{
		\fun{atkops}{\codeenv,S\listsep V,\top} = \fun{atkops}{\codeenv,S,\top} \listsep V
	}{atkops-atk-val}
	\typerule{atkops-code}{
		\codeenv{\procid} = \_
	}{
		\fun{atkops}{\codeenv,S\listsep \text{startfun }\procid,\top} = \fun{atkops}{\codeenv,S,\bot}
	}{atkops-code}
	\typerule{atkops-code-val}{
		\codeenv{\procid} = \text{undefined}
	}{
		\fun{atkops}{\codeenv,S\listsep V,\bot} = \fun{atkops}{\codeenv,S,\bot}
	}{atkops-code-val}
\end{center}

\paragraph{Types of Attackers}
Indicate with \com{T} a list of types.
Assume given a function \fun{declaredTypes}{\com{\codeenv}} that returns all the types declared in \com{\codeenv}.
This function is used to retrieve all types defined by the \tc.

We can overapproximate the types defined by the attacker by collecting all types that are valid and not mentioned by the \tc.
\begin{center}
	\typerule{Attacker Types}{
		\com{T} = \myset{ \com{\tau} }{ \com{\tau}\in\text{Type} \text{ and } \com{\tau}\notin\fun{declaredTypes}{\com{\codeenv}} }
	}{
		\fun{atktypes}{\codeenv} = \com{T}
	}{atk-types}
\end{center}

This definition has this form because we do not carry around the attacker definition, so we do not have its type definitions in the statement of the weak property where this function is used.
Instead we just have the \tc, so we rely on that information to know what types may the attacker declare..
An alternative (for where this function is used. i.e., for the weak property) is to take the globals, take the types defined by the module and remove them.
This alternative would make this definition of \fun{atktypes}{\cdot} not necessary -- we still need the \fun{declaredtypes}{\cdot} though.

\paragraph{Memory and Global Restriction}
Given a memory \com{M} and a set of locations \com{\ell}, we indicate the sub-memory restricted to just the locations of \com{\ell} as:
\[
	\restr{\com{M}}{\com{\ell}}
\]

Given globals \com{G} and a set of addresses \com{A} or of types \com{T} (such as those derivable by \dom{\inv}), we indicate the sub-globals restricted to just the domain of \com{A} or \com{T} as:
\[
	\restr{\com{G}}{\com{A}}
	\qquad
	\restr{\com{G}}{\com{T}}
\]
To perform a restriction according to both parameters, we write \(\restr{\restr{\com{G}}{\com{A}}}{\com{T}}\) (the order of \com{A} and \com{T} is irrelevant here).

\subsubsection{State Properties for Proofs}\label{sec:invars-states}
The RS proof is based upon the operational state maintaining a strong property, which is that if we take the globals, we can partition them in 2 parts: invariant related and not.
The invariant related are never alterable from attacker code and at boundaries crossing, applying the invariants is always true.
This last part is a bit strong, it is exactly what we need.
The strong property is preserved if the attacker is executing, because the attacker cannot vary the invaraint-related globals.
Then the strong property is preserved if the \tc is reducing because \tc is verified and encspaulated.
Intuitively, verification and encapsulation each give us a weaker property that, once combined, yield the strong one.
If some code is locally verified, then starting from a state with the strong property, we do a !-action and at the end we have a state with a weaker property: for the globals that are
invariant related, applying the invariant is always true.
Then we need to apply the encapsulator: if some code is encapsulated, then starting from a state with the invariant, we do a !-action and at the end we have a state with a weaker property: for the globals that are invariant related, they are not alterable from attacker code.

\begin{center}
	\typerule{Strong Property}{
		\weakinvinv{\codeenv}{\sigma}{\inv}
		&
		\weakinvatk{\codeenv}{\sigma}{\inv}
	}{
		\stronginv{\codeenv}{\sigma}{\inv}
	}{stron-inv}
\end{center}
According to the module code \com{\codeenv}, the program state \com{\sigma} respects invariants \com{\inv} strongly.
That is, the program state respects the invariant weakly and makes a part of the global unreachable from attacker code.
For the weak properties we rely on knowledge of what code belongs to the module to understand which stack frame are not its own, and therefore they belong to the attacker.
\begin{center}
	\typerule{Attacker Part of State}{
		\com{\sigma} = \com{\tup{C,M,G,S}}
		\\
		\com{C_a} = \fun{atkstk}{\com{\codeenv},\com{C}}
		&
		\com{L_a} = \com{C_a}.\text{locals}
		&
		\com{S_a} = \fun{atkops}{\codeenv,S}
		&
		\com{T} = \fun{atktypes}{ \com{\codeenv} }
		\\
		\com{G_a} = \restr{G}{\com{T}}
		&
		\fun{locsof}{\com{L_a}} \cup \fun{locsof}{G_a} \cup \fun{locsof}{\com{S_a}} = \com{\ell}
		&
		\com{M_a} = \restr{\com{M}}{\com{\ell}}
	}{
		\com{\codeenv,\sigma} \vdash \com{M_a}, \com{G_a} : \com{attackerpart}
	}{atk-part-state}
\end{center}
The relevant parts of a state for weak properties are attacker memory and attacker globals.
Attacker memory \com{M_a} is the memory whose locations are found: in attacker locals (\com{L_a}), in attacker globals (\com{G_a}) and in attacker operand stack frames (\com{S_a}).
Attacker globals are those whose type is attacker-defined (\com{T}).
Attacker locals are the locals extracted from attacker call stacks frames (\com{C_a}).
Attacker operand stack frames are those that belong to the attacker as extracted from the stack.

In the following, given two sets \com{A} and \com{B}, we write that they are disjoint as $\com{A\notcap B}$, so $\com{A\notcap B}\isdef \com{A\cap B = \emptyset}$.
\begin{center}
	\typerule{Weak Property - Inv}{
		\com{\sigma} = \com{\tup{C,M,G,S}}
		&
		\com{M,G} \vdash \com{\inv}
	}{
		\weakinvinv{\codeenv}{\sigma}{\inv}
	}{weak-inv-inv}
	\typerule{Weak Property - Atk Changes}{
		\com{\sigma} = \com{\tup{C,M,G,S}}
		&
		\com{G_{i}} = \restr{\com{G}}{\domG{\inv}}
		&
		\com{M_{i}} = \restr{\com{M}}{\domM{\inv}}
		\\
		\com{\codeenv,\sigma} \vdash \com{M_a}, \com{G_a} : \com{attackerpart}
		\\
		{\com{G_{i}}} \notcap {\com{G_a}}
		&
		\dom{\com{M_{i}}} \notcap \dom{\com{M_a}}
	}{
		\weakinvatk{\codeenv}{\sigma}{\inv}
	}{weak-inv-unchange}
\end{center}
As stated, \Cref{tr:weak-inv-inv} means that the memory and globals part of a state respect the invariant, as per \Cref{tr:inv-sat}.

Instead, \Cref{tr:weak-inv-unchange} first extracts the attacker memory and globals as per \Cref{tr:atk-part-state}.
Then it checks that no attacker global matches one with an invariant on, and it checks that no attacker memory location matches a location with an invariant on.
It is ok for the attacker to temporarily have a global address which has an invariant on, because when that address is used by attacker code, the type part will be auto-filled by the semantics with an attacker-supplied type, and this ensures the global accessed with that address does not have an invariant on.

\begin{lemma}[The Two Weak Properties Imply the Strong One]\label{thm:wk-inv-impl-str}
	\begin{align*}
		\text{ if }
		&\
		\weakinvinv{\codeenv'}{\sigma}{\inv}
		\\
		\text{ and }
		&\
		\weakinvatk{\codeenv'}{\sigma}{\inv}
		\\
		\text{ then }
		&\
		\stronginv{\codeenv'}{\sigma}{\inv}
	\end{align*}
\end{lemma}
\begin{proof}[Proof of \Thmref{thm:wk-inv-impl-str}]\hfill

	By definition in \Cref{tr:stron-inv}.
\end{proof}
\BREAK

\subsubsection{Auxiliary Lemmas}\label{sec:aux-lemmas}
We rely on the following auxiliary lemmas in order to prove the generalised RS statement (\Thmref{thm:rs-general}).

First, initial states respect the strong property.
\begin{lemma}[The Initial State Respects the Strong Property]\label{thm:sat-initial}
	\begin{align*}
		\text{ if }
		&\
		\vdash\com{\codeenv+\atk} : \com{ok}
		\\
		\text{ then }
		&\
		\stronginv{\codeenv}{\SInit{\codeenv+\atk}}{\inv}
	\end{align*}
\end{lemma}
\begin{proof}[Proof of \Thmref{thm:sat-initial}]\hfill

	Trivial.
\end{proof}
\BREAK

The next lemma tells that reductions in the attacker (\Thmref{thm:rs-att-red}) respect the property.
\begin{lemma}[Attacker Reductions Respect the Strong Property]\label{thm:rs-att-red}
	\begin{align*}
		\text{ if }
		&\
		\com{\codeenv' \triangleright \codeenv,\procid} \vdash \com{\sigma \Xto{\alpha?} \sigma'}
		\\
		\text{ and } &
			\agree{\codeenv'}{\inv}
		\\
		\text{ and }
		&\
		\com{\codeenv'}\vdash\com{\sigma}:\com{atkcode}
		\\
		\text{ and }
		&\
		\stronginv{\codeenv}{\sigma}{\inv}
		\\
		\text{ then }
		&\
		\stronginv{\codeenv'}{\sigma'}{\inv}
		\\
		\text{ and }
		&\
		\com{\alpha?}\Vdash\com{\inv}
	\end{align*}
\end{lemma}
\begin{proof}[Proof of \Thmref{thm:rs-att-red}]\hfill

	By \Cref{tr:sing} we need to prove:
	\begin{enumerate}
		\item $\com{\codeenv' \triangleright \codeenv,\procid} \vdash \com{\sigma \Xto{\nil} \sigma''}
		$

		This follows from \Thmref{thm:rs-att-silent-red} and we also get that $
		\stronginv{\codeenv'}{\sigma''}{\inv}$ (HS2).

		\item $\com{\codeenv' \triangleright \codeenv,\procid, i} \vdash \com{\sigma'' \xto{\alpha} \sigma'}$

		There are two cases here:
		\begin{itemize}
			\item[\Cref{tr:act-call}:]

			By HS2, we have $\weakinvinv{\codeenv}{\sigma''}{\inv}$ (HW1) and $\weakinvatk{\codeenv}{\sigma''}{\inv}$ (HW2).

			We need to prove $\weakinvinv{\codeenv}{\sigma'}{\inv}$ (TW1) and $\weakinvatk{\codeenv}{\sigma'}{\inv}$ (TW2).

			We have \com{\sigma''} = \com{\tup{C,M,G,S}} and \com{\sigma'} = \com{\tup{C',M,G,S'}}.

			Thus, since \com{M} and {G} are not changed, TW1 follows from HW1.

			Regarding TW2: \com{C} contains one more binding than \com{C'}, but it is not in attacker functions and thus not considered by the weak property.

			Similarly and \com{S} also contains more locals than \com{S'} but they are not in attacker functions.

			Thus, the attacker-related entries in \com{\sigma''} and in \com{\sigma'} are the same and TW2 holds by HW2.s

			\item[\Cref{tr:act-retback}:]

			By \Thmref{thm:strong-prop-monot}.

		\end{itemize}

		In both cases, by HW1, the second conjunct of the thesis holds, now we focus on the first one.
	\end{enumerate}
\end{proof}
\BREAK

This lemma tells us that if we know the weak invariant property on a state, and that state takes a step decreasing its call stack and/or its operand stack, the new state also upholds the weak invariant property.
\begin{lemma}[Monotonicity for Weak Property INV ]\label{thm:weak-prop-inv-monot}
	\begin{align*}
		\text{ if }
		&\
		\weakinvinv{\codeenv'}{\st{C, M, G, S}}{\inv}
		\\
		\text{ and }
		&\
		\peval{\st{C, M, G, S}}{\st{C', M, G, S'}}
		\\
		\text{ and }
		&\
		\com{C'}\subseteq\com{C}
		\\
		\text{ and }
		&\
		\com{S'}\subseteq\com{S}
		\\
		\text{ then }
		&\
		\weakinvinv{\codeenv'}{\st{C, M, G, S}}{\inv}
	\end{align*}
\end{lemma}
\begin{proof}
	By \Cref{tr:weak-inv-inv}.
\end{proof}
\BREAK

This lemma tells us that if we know the weak unreachable property on a state, and that state takes a step decreasing its call stack and/or its operand stack, the new state also upholds the weak unreachable property.
\begin{lemma}[Monotonicity for Weak Property Unreachable ]\label{thm:weak-prop-unch-monot}
	\begin{align*}
		\text{ if }
		&\
		\weakinvatk{\codeenv'}{\st{C, M, G, S}}{\inv}
		\\
		\text{ and }
		&\
		\peval{\st{C, M, G, S}}{\st{C', M, G, S'}}
		\\
		\text{ and }
		&\
		\com{C'}\subseteq\com{C}
		\\
		\text{ and }
		&\
		\com{S'}\subseteq\com{S}
		\\
		\text{ then }
		&\
		\weakinvatk{\codeenv'}{\st{C', M, G, S'}}{\inv}
	\end{align*}
\end{lemma}
\begin{proof}
	\com{C} contains less attacker-related bindings than \com{C'}, and \com{S} also contains less attacker frames than \com{S'}.

	Let \com{C_a} and \com{S_a} be the call stack and operand stack parts of \Cref{tr:atk-part-state} applied to the state with \com{C} and \com{S}.

	Let \com{C_a'} and \com{S_a'} be the call stack and operand stack parts of \Cref{tr:atk-part-state} applied to the state with \com{C'} and \com{S'}.

	By assumptions HP3 and HP4 we thus have HPC $\com{C_a'}\subseteq\com{C_a}$ and HPS $\com{S_a'}\subseteq\com{S_a}$

	By assumption HP1 we have HPNS $\com{G_{i}} \notin \com{S_a}$ and HPNC $\com{G_{i}} \notin \com{C_a}$.

	By HPC with HPNC and HPS with HPNS we can conclude that $\com{G_{i}} \notin \com{S_a'}$ and $\com{G_{i}} \notin \com{C_a'}$, so the thesis holds.
\end{proof}
\BREAK

This lemma tells us that if we know the strong property on a state, and that state takes a step decreasing its call stack and/or its operand stack, the new state also upholds the strong property.
\begin{lemma}[Monotonicity for Strong Property ]\label{thm:strong-prop-monot}
	\begin{align*}
		\text{ if }
		&\
		\stronginv{\codeenv'}{\st{C, M, G, S}}{\inv}
		\\
		\text{ and }
		&\
		\peval{\st{C, M, G, S}}{\st{C', M, G, S'}}
		\\
		\text{ and }
		&\
		\com{C'}\subseteq\com{C}
		\\
		\text{ and }
		&\
		\com{S'}\subseteq\com{S}
		\\
		\text{ then }
		&\
		\stronginv{\codeenv'}{\st{C, M, G, S}}{\inv}
	\end{align*}
\end{lemma}
\begin{proof}
	Indicate \com{\st{C, M, G, S}} as \com{\sigma} and \com{\st{C', M, G, S'}} as \com{\sigma'}.

	By HP1, we have $\weakinvinv{\codeenv}{\sigma}{\inv}$ (HW1) and $\weakinvatk{\codeenv}{\sigma}{\inv}$ (HW2).

	By \Cref{tr:stron-inv} we need to prove
	\begin{enumerate}
		\item $\weakinvinv{\codeenv}{\sigma'}{\inv}$ (TW1)

		Since \com{M} and {G} are not changed we apply \Thmref{thm:weak-prop-inv-monot} with HW1 and conclude TW1.

		\item  $\weakinvatk{\codeenv}{\sigma'}{\inv}$ (TW2)

		This follows from \Thmref{thm:weak-prop-unch-monot}.
	\end{enumerate}
\end{proof}
\BREAK

This lemma tells us that in order to derive the strong property on a state, we can split all its components in 2 parts and check the strong property on the 2 sub-parts individiually, so long as the second part has a very specific form: all new additions to globals and locals point to a memory location whose content was part of the state with the strong property.
\begin{lemma}[Compositionality of Strong Property]\label{thm:strong-prop-compos}
	\begin{align*}
		\text{ if }
		&\
		\com{\sigma} = \com{\st{\st{\procid,\pc,L'\listsep L}\listsep C,M'\listsep M,G'\listsep G,S}}
		\\
		\text{ and }
		&\
		\stronginv{\codeenv}{\st{\st{\procid,\pc,L}\listsep C,M,G,S}}{\inv}
		\\
		\text{ and }
		&\
		\stronginv{\codeenv}{\st{\st{\procid,\pc,L'}\listsep C,M',G',\nil}}{\inv}
		\\
		\text{ and }
		&\
		\forall \com{a}\in\img{L'}\ldotp \com{a}\in\dom{M'}
		\\
		\text{ and }
		&\
		\forall \com{\st{a,\stype}}\in\dom{G'}\ldotp \com{G'(\st{a,\stype})}\in\dom{M'}
		\\
		\text{ and }
		&\
		\forall \com{v}\in\img{M'}\ldotp \com{v}\in\com{S}
		\\
		\text{ and }
		&\
		\dom{M'}\cap\dom{M}=\emptyset
		\\
		\text{ then }
		&\
		\stronginv{\codeenv}{\sigma}{\inv}
	\end{align*}
\end{lemma}
\begin{proof}
	By \Cref{tr:stron-inv} we have to prove:
	\begin{enumerate}
		\item $\weakinvinv{\codeenv}{\sigma}{\inv}$

		From HP7 and HP 5 we have that \com{M} and \com{M'} as well as \com{G} and \com{G'} have disjoint domains.

		From HP2 we have the thesis for the \com{M} and \com{G} subparts.

		From HP3 we have the thesis for the \com{M'} and \com{G'} subparts.

		\item $\weakinvatk{\codeenv}{\sigma}{\inv}$

		We need to show that the \com{G_a} part of \com{G'} is not intersecting \com{G_i}, which holds from HP3.

		Then we need to show that locations in \com{L'} do not intersect with \com{M_i}.

		This follows from HP6 since from HP4 those addresses and locations are in \com{M'} and from HP6 we have that they were in \com{S}, which from HP2 we know to satisfy the definition of \Cref{tr:weak-inv-unchange}.

	\end{enumerate}
\end{proof}
\BREAK

This lemma tells that in a state whose call stack and operand stack can pe decomposed in 2 parts, but where M and G remain the same it is sufficient to know the strong property on part 1 of C and S and just the weak unreachable property on part 2 to derive the strong property on the full state.

\begin{lemma}[Compositionality of Strong Property and Weak Property Conditions]\label{thm:strong-weak-unch-impl-strong-nom-nog}
	\begin{align*}
		\text{ if }
		&\
		\com{\sigma} = \com{\st{\st{\procid,\pc,L'\listsep L}\listsep C,M,G,S'\listsep S}}
		\\
		\text{ and }
		&\
		\stronginv{\codeenv}{\st{C,M,G,S}}{\inv}
		\\
		\text{ and}
		&\
		\forall \com{\loc} \in \fun{locsof}{L'} \cup \fun{locsof}{S'}\ldotp  \com{\dom{M_i} \notcap c } %
		\\
		\text{ then }
		&\
		\stronginv{\codeenv}{\sigma}{\inv}
	\end{align*}
\end{lemma}
\begin{proof}

	From \Cref{tr:stron-inv} we have to prove
	\begin{enumerate}
		\item $\weakinvinv{\codeenv}{\sigma}{\inv}$

		Since there is no change to \com{M} nor \com{G}, this holds from HP2.

		\item $\weakinvatk{\codeenv}{\sigma}{\inv}$

		From \Cref{tr:atk-part-state} we have that \com{L_a} = $\com{L_a'} \cup \com{L_a''}$ where \com{L_a'} come from \com{L} and \com{L_a''} come from \com{L'}.

		Similarly \com{S_a} = $\com{S_a'} \cup \com{S_a''}$ where \com{S_a'} come from \com{S} and \com{S_a''} come from \com{S'}.

		Also, \com{M_a} = $\com{M_a'} \cup \com{M_a''}$ where \com{M_a'} come from $\fun{locsof}{L_a'}\cup\fun{locsof}{S_a'}$ and \com{M_a''} come from $\fun{locsof}{L_a''}\cup\fun{locsof}{S_a''}$.

		Thus, \com{A_{mem}} = $\com{A_{mem}'} \cup \com{A_{mem}''}$ where \com{A_{mem}'} come from \com{M_a'} and \com{A_{mem}''} come from \com{M_a''}.

		We need to prove that all the elements respect \Cref{tr:weak-inv-unchange}.

		For all $'$ elements, this follows from HP2.

		For all $''$ elements, we have the following:
		\begin{enumerate}
			\item we need to show that ${\com{G_{i}}} \notcap {\com{G_a}}$ this is trivial since \com{G_a} is empty
			\item we need to show that $\dom{\com{M_{i}}} \notcap \dom{\com{M_a''}}$, this follows from HP3
		\end{enumerate}

	\end{enumerate}
\end{proof}
\BREAK

\begin{lemma}[Attacker Silent Reductions Respect the Strong Property]\label{thm:rs-att-silent-red}
	\begin{align*}
		\text{ if }
		&\
		\com{\codeenv' \triangleright \codeenv,\procid} \vdash \com{\sigma \Xto{\nil} \sigma'}
		\\
		\text{ and } &
			\agree{\codeenv'}{\inv}
		\\
		\text{ and }
		&\
		\com{\codeenv'}\vdash\com{\sigma}:\com{atkcode}
		\\
		\text{ and }
		&\
		\stronginv{\codeenv}{\sigma}{\inv}
		\\
		\text{ then }
		&\
		\stronginv{\codeenv'}{\sigma'}{\inv}
	\end{align*}
\end{lemma}
\begin{proof}[Proof of \Thmref{thm:rs-att-silent-red}]\hfill

	This proof proceeds by induction on $\Xto{}$:
	\begin{itemize}
		\item[Base, \Cref{tr:sing-refl}:]
			This trivially holds.
		\item[Inductive, \Cref{tr:sing}:]
			so we have:
			\begin{itemize}
				\item[$\com{\codeenv' \triangleright \codeenv,\procid} \vdash \com{\sigma \Xto{\nil} \sigma''}$]

				By IH we have $\stronginv{\codeenv'}{\sigma''}{\inv}$ (HS2).

				\item[$\com{\codeenv' \triangleright \codeenv,\procid, i} \vdash \com{\sigma'' \xto{\alpha} \sigma'}$]
			\end{itemize}

			Since $\alpha=\nil$ we can only apply \Cref{tr:act-no}, so the proof now proceeds by case analysis on $\com{\peval{\sigma''}{\sigma'}}$:
			\begin{itemize}
				\item[\Cref{tr:sem-ift}]

				By \Thmref{thm:strong-prop-monot}.

				\item[\Cref{tr:sem-iff}]

				By \Thmref{thm:strong-prop-monot}.

				\item[\Cref{tr:sem-ret}]

				Since no label is created this is a return to an attacker function (\Cref{tr:cross-no}).

				By \Thmref{thm:strong-prop-monot}.

				\item[\Cref{tr:sem-call}]

				Since no label is created this is a call to an attacker function (\Cref{tr:cross-no}).

				The only change to the program state is a new element on \com{C} which contains no bindings.

				So, this holds by a reasoning similar to \Thmref{thm:strong-prop-monot}.

				\item[\Cref{tr:sem-glob}]
				We perform case analysis on global reductions:
				\begin{itemize}
					\item[\Cref{tr:sem-glob-borrowglobal}]

					We have \com{i = \ao{\borrowglobalCmd}{s}} in
					\begin{center}
						\AxiomC{
							$
							\stype = \tup{\procid.mid, s}\sep
					      	\gread{G}{\tup{\addr,\stype}} = \loc
					      	$
						}
						\UnaryInfC{
							$
							\geval{
								\ao{\borrowglobalCmd}{s}
							}{
								\st{M,G,\stackht{\addr}{S}}
							}{
								\st{M,G,\stackht{\rft{\loc}{[]}}{S}}
							}
						    $
						}
						\UnaryInfC{
							$
							\peval{
							    \st{\stackht{\tup{\procid,\pc,L}}{C}, M,G,\stackht{\addr}{S}}
							}{
							    \st{\stackht{\tup{\procid,\pc+1,L}}{C}, M,G,\stackht{\rft{\loc}{[]}}{S}}
							}
							$
						}
						\DisplayProof
					\end{center}

					We need to prove:
					\begin{itemize}
						\item $\stronginv{\codeenv'}{\st{\stackht{\tup{\procid,\pc+1,L}}{C}, M,G,\stackht{\rft{\loc}{[]}}{S}}}{\inv}$
					\end{itemize}

					By \Thmref{thm:strong-weak-unch-impl-strong-nom-nog} it suffices to prove for $\com{L'} = \nil$ and $\com{S'} = \rft{\loc}{\nil}$:
					\begin{enumerate}
						\item\label{proof:pt1bg} $\stronginv{\codeenv'}{\st{\tup{\procid,\pc+1,L}{C},M,G,S}}{\inv}$, which follows from \Thmref{thm:strong-prop-monot} with HS2 since \pc+1 does not play a role in the properties.

						\item\label{proof:pt4bg} $\forall \com{\loc} \in \fun{locsof}{L'} \cup \fun{locsof}{S'}\ldotp \com{\dom{M_i} \notcap c }$ %

						By HS2 we have that $\com{G_i}.1\notcap\com{a,\stype}.1$, so \com{\loc} is collected by \fun{locsof}{\cdot} and put in \com{M_a}, so we have  HPG $\com{G_i}.1\notcap\com{M(\loc)}$. %

						By \Cref{tr:atk-part-state} we have that \com{\loc} is collected by \fun{locsof}{}, so it ends up in \com{S_a}.

						By \Cref{tr:weak-inv-unchange} we need to prove that $\com{M_i}\notcap\com{\loc}$ which holds by HPC. %

					\end{enumerate}

					\item[\Cref{tr:sem-glob-moveto}]

					We have \com{i = \ao{\movetoCmd}{s}} in
					\begin{center}
						\AxiomC{
							$
							\stype = \tup{\procid.mid, s}\sep
							\tup{\addr, \stype)} \notin \dom{G}\sep
							\loc\notin\dom{M}
					      	$
						}
						\UnaryInfC{
							$
							\geval{
								\ao{\movetoCmd}{s}
							}{
								\st{M,G,\stackht{\stackht{\addr}{\val}}{S}}
							}{
								\st{\mset{M_{\mdel{\mc{C}}{\loc}}}{\loc}{\val},\gset{G}{\tup{a,\stype}}{\loc},S}
							}
						    $
						}
						\UnaryInfC{
							$
							\peval{
							    \st{\stackht{\tup{\procid,\pc,L}}{C}, M,G,\stackht{\stackht{\addr}{\val}}{S}}
							}{
							    \st{\stackht{\tup{\procid,\pc+1,L}}{C}, \mset{M_{\mdel{\mc{C}}{\loc}}}{\loc}{\val},\gset{G}{\tup{a,\stype}}{\loc},S}
							}
							$
						}
						\DisplayProof
					\end{center}

					We need to prove:
					\begin{itemize}
						\item $\stronginv{\codeenv'}{\st{\stackht{\tup{\procid,\pc+1,L}}{C}, \mset{M}{\loc}{\val}, \gset{G}{\tup{a,\stype}}{\loc}, S}}{\inv}$
					\end{itemize}

					By \Thmref{thm:strong-prop-monot} it suffices to prove:
					\begin{itemize}
						\item $\stronginv{\codeenv'}{\st{\stackht{\tup{\procid,\pc+1,L}}{C}, \mset{M}{\loc}{\val}, \gset{G}{\tup{a,\stype}}{\loc}, \stackht{v}{S}}}{\inv}$
					\end{itemize}

					By \Thmref{thm:strong-prop-compos} with HPL $\com{L'}=\nil$, HPM $\com{M'}=\com{\loc\mapsto \val}$, HPG $\com{G'}=\com{\tup{a,\stype}\mapsto \loc}$ it suffices to prove:
					\begin{enumerate}
						\item $\stronginv{\codeenv'}{\st{\tup{\procid,\pc+1,L}{C},M,G,\stackht{v}{S}}}{\inv}$, which follows from \Thmref{thm:strong-prop-monot} with HS2 since \pc+1 does not play a role in the properties.

						\item $\stronginv{\codeenv'}{\st{\st{\procid,\pc+1,\nil},{\loc}\mapsto{\val}, {\tup{a,\stype}}\mapsto{\loc},\nil}}{\inv}$

						By \Cref{tr:stron-inv} we have to prove:
						\begin{enumerate}
							\item $\weakinvinv{\codeenv'}{\st{\st{\procid,\pc+1,\nil},{\loc}\mapsto{\val}, {\tup{a,\stype}}\mapsto{\loc},\nil}}{\inv}$

							Since the reduction is in attacker code, (\procid is attacker-defined), we have HPA: $\stype$ is an attacker-defined type.

							Since $\tup{a,\stype}$ is fresh by HPA, $\tup{a,\stype}$ cannot be in \domG{\inv}.
							Since \com{\loc} is fresh, no existing \com{\tup{a',\stype'}} can ever have pointed to it, so \com{\loc} cannot be in \domM{\inv}.

							Thus this case holds.

							\item $\weakinvatk{\codeenv'}{\st{\st{\procid,\pc+1,\nil},{\loc}\mapsto{\val}, {\tup{a,\stype}}\mapsto{\loc},\nil}}{\inv}$

							By \Cref{tr:weak-inv-unchange} we need to prove
							\begin{enumerate}
								\item $\com{G_i}\notcap\com{\tup{a,\stype}}$

								Since $\stype$ is an attacker-defined type, it matches \fun{atktypes}{\codeenv}.

								From \Cref{prop:inv-type-atk} we know no element in \com{G_i} can mention $\stype$, so this case holds.

								\item $\com{M_i}\notcap\com{\loc}$

								This follows from the freshess of \com{\loc}.

							\end{enumerate}
						\end{enumerate}

						\item $\forall \com{a}\in\img{L'}\ldotp \com{a}\in\dom{M'}$

						this is trivial from HPL

						\item $\forall \com{\st{a,\stype}}\in\dom{G'}\ldotp \com{G'(\st{a,\stype})}\in\dom{M'}$

						this is trivial from HPG

						\item $\forall \com{v}\in\img{M'}\ldotp \com{v}\in\com{S}$

						this is true from HPM and definition of \com{S}.

						\item $\dom{M}\cap c = \emptyset$

						This is trivially true.
					\end{enumerate}

					\item[\Cref{tr:sem-glob-movefrom}]

					We have \com{i = \ao{\movefromCmd}{s}} in
					\begin{center}
						\AxiomC{
							$
							\stype = \tup{\procid.mid, s}\sep
							\gread{G}{\tup{\addr,\stype}} = \loc\sep
							\mread{M}{\loc}=\val
					      	$
						}
						\UnaryInfC{
							$
							\geval{
							\ao{\movefromCmd}{s}
							}{
							\st{M,G,\stackht{\addr}{S}}
							}{
							\st{\mdel{M}{\loc},\gdel{G}{\tup{a,s}},\stackht{\val}{S}}
							}
						    $
						}
						\UnaryInfC{
							$
							\peval{
							    \st{\stackht{\tup{\procid,\pc,L}}{C}, M,G,\stackht{\addr}{S}}
							}{
							    \st{\stackht{\tup{\procid,\pc+1,L}}{C}, \mdel{M}{\loc},\gdel{G}{\tup{a,s}},\stackht{\val}{S}}
							}
							$
						}
						\DisplayProof
					\end{center}

					This follows the same structure as the case for BorrowGlobal except for the details that we describe below.

					We need to prove \Cref{proof:pt4bg} for $\com{L'} = \nil$ and $\com{S'} = \val$.

					\Cref{proof:pt4bg} is trivial since \com{\val} cannot be a location.

					\item[\Cref{tr:sem-glob-unpack}]

					We have \com{i = \ao{\unpackCmd}{s}} in
					\begin{center}
						\AxiomC{
							$
							\val=\tup{\record{(f_{i},\val_{i})\mid 1 \leq i \leq n}, \tg}
					      	$
						}
						\UnaryInfC{
							$
							\geval{
							\ao{\unpackCmd}{s}
							}{
							\tup{M, G, \stackht{\val}{S}}
							}{
							\tup{M, G, \stackht{\val_{1}\cons\cdots\cons \val_{n}}{S}}
							}
						    $
						}
						\UnaryInfC{
							$
							\peval{
							    \st{\stackht{\tup{\procid,\pc,L}}{C}, M,G,\stackht{\val}{S}}
							}{
							    \st{\stackht{\tup{\procid,\pc+1,L}}{C}, M, G,\stackht{\val_{1}\cons\cdots\cons \val_{n}}{S}}
							}
							$
						}
						\DisplayProof
					\end{center}

					This follows the same structure as the case for BorrowGlobal except for the details that we describe below.

					We need to prove \Cref{proof:pt4bg} for $\com{L'} = \nil$ and $\com{S'} = \val_{1}\cons\cdots\cons \val_{n}$.

					\Cref{proof:pt4bg} is simple since all locations were already in attacker locations from HS2.

					\item[\Cref{tr:sem-glob-pack}]

					We have \com{i = \ao{\unpackCmd}{s}} in
					\begin{center}
						\AxiomC{
							$
							\tg = gen\_tag(\codeenv(\tup{\procid.mid, s}).kind)\sep
      						v = \record{(f_{i},\val_{i})\mid 1 \leq i \leq n}
					      	$
						}
						\UnaryInfC{
							$
							\geval{
							\ao{\packCmd}{s}
							}{
							\tup{M, G, \stackht{\val_{1}\cons\cdots\cons \val_{n}}{S}}
							}{
							\tup{M, G, \stackht{\tv{\val}{\tg}}{S}}
							}
						    $
						}
						\UnaryInfC{
							$
							\peval{
							    \st{\stackht{\tup{\procid,\pc,L}}{C}, M,G,\stackht{\val_{1}\cons\cdots\cons \val_{n}}{S}}
							}{
							    \st{\stackht{\tup{\procid,\pc+1,L}}{C}, M, G,\stackht{\tv{\val}{\tg}}{S}}
							}
							$
						}
						\DisplayProof
					\end{center}

					This follows the inverse reasoning of the case above, the tag is irrelevant to our proof.

				\end{itemize}
				\item[\Cref{tr:sem-loc}]
				We perform case analysis on local reductions:
				\begin{itemize}
					\item[\Cref{tr:sem-loc-op}]
					We have \com{i = \stackopCmd} in
					\begin{center}
						\AxiomC{
							$
							\com{v''} = \com{[[ v\ \stackopCmd\ v' ]]}
							$
						}
						\UnaryInfC{
							$
							\eval{
						      \st{M,L,\stackht{v}{\stackht{v'}{S}}}
						    }{
						      \stackopCmd
						    }{
						      \st{M,L,\stackht{v''}{S}}
						    }
						    $
						}
						\UnaryInfC{
							$
							\peval{
							    \st{\stackht{\tup{\procid,\pc,L}}{C}, M,G,\stackht{\val}{\stackht{\val'}{S}}}
							}{
							    \st{\stackht{\tup{\procid,\pc+1,L}}{C}, M,G,\stackht{\val''}{S}}
							}
							$
						}
						\DisplayProof
					\end{center}

					We need to prove:
					\begin{itemize}
						\item $\stronginv{\codeenv'}{ \st{\stackht{\tup{\procid,\pc+1,L}}{C}, M,G,\stackht{\val''}{S}} }{\inv}$
					\end{itemize}

					By \Thmref{thm:strong-weak-unch-impl-strong-nom-nog} it suffices to prove for $\com{L'} = \nil$ and $\com{S'} = \val''$:
					\begin{enumerate}
						\item $\stronginv{\codeenv'}{\st{\tup{\procid,\pc+1,L}{C},M,G,S}}{\inv}$, which follows from \Thmref{thm:strong-prop-monot} with HS2 since \pc+1 does not play a role in the properties.

						\item\label{proof:pt4op} $\forall \com{\loc} \in \fun{locsof}{L'} \cup \fun{locsof}{S'}\ldotp \com{\dom{M_i} \notcap c }$ %

						This is not possible since the return type is not a location.

					\end{enumerate}

					\item[\Cref{tr:sem-loc-loadconst}]

					This follows the same structure as the case for Op except for the details that we describe below.

					We need to prove \Cref{proof:pt4op} for $\com{L'} = \nil$ and $\com{S'} = v$.

					This is trivial since $v$ can be a Nat, a Bool or a global address only.

					\Cref{proof:pt4bg} is trivial since \com{\val} cannot be a location.

					\item[\Cref{tr:sem-loc-pop}]

					We have \com{i = \popCmd} in
					\begin{center}
						\AxiomC{
							$
							\eval{
								\st{M,L,\stackht{v}{S}}
							}{
								\popCmd
							}{
								\st{M,L,S}
							}
						    $
						}
						\UnaryInfC{
							$
							\peval{
							    \st{\stackht{\tup{\procid,\pc,L}}{C}, M,G,\stackht{v}{S}}
							}{
							    \st{\stackht{\tup{\procid,\pc+1,L}}{C}, M,G,S}
							}
							$
						}
						\DisplayProof
					\end{center}

					We need to prove:
					\begin{itemize}
						\item $\stronginv{\codeenv'}{\st{\stackht{\tup{\procid,\pc+1,L}}{C}, M, G, S}}{\inv}$
					\end{itemize}

					This follows from \Thmref{thm:strong-prop-monot} with HS2.

					\item[\Cref{tr:sem-loc-writeref}]

					We have \com{i = \writerefCmd} in
					\begin{center}
						\AxiomC{
							$
							\val_2 = \rft{\loc}{p}\sep
						    v' =\mread{M}{\loc}\sep
							$
						}
						\UnaryInfC{
							$
							\eval{
						      \st{M,L,\stackht{\val_1}{\stackht{\val_2}{S}}}
						    }{
						      \writerefCmd
						    }{
						      \st{\mset{M}{\loc}{v'[p := v_1]},L,S}
						    }
						    $
						}
						\UnaryInfC{
							$
							\peval{
							    \st{\stackht{\tup{\procid,\pc,L}}{C}, M,G,\stackht{\val_1}{\stackht{\val_2}{S}}}
							}{
							    \st{\stackht{\tup{\procid,\pc+1,L}}{C},\mset{M}{\loc}{v'[p := v_1]},G,S}
							}
							$
						}
						\DisplayProof
					\end{center}

					We need to prove:
					\begin{itemize}
						\item $\stronginv{\codeenv'}{ \st{\stackht{\tup{\procid,\pc+1,L}}{C},\mset{M}{\loc}{v'[p := v_1]},G,S} }{\inv}$
					\end{itemize}

					By \Thmref{thm:strong-prop-monot} it suffices to prove:
					\begin{itemize}
						\item $\stronginv{\codeenv'}{ \st{\stackht{\tup{\procid,\pc+1,L}}{C},\mset{M}{\loc}{v'[p := v_1]},G,\stackht{\val_1}{S}} }{\inv}$
					\end{itemize}

					By \Thmref{thm:strong-prop-compos} with HPL $\com{L'}=\nil$, HPM $\com{M'}=\com{\loc\mapsto \val'[p:=\val_1]}$, HPG $\com{G'}=\nil$ it suffices to prove
					\begin{enumerate}
						\item $\stronginv{\codeenv'}{\st{\tup{\procid,\pc+1,L}{C},{M},G,\stackht{\val}{S}}}{\inv}$, which follows from \Thmref{thm:strong-prop-monot} with HS2 since \pc+1 does not play a role in the properties.

						\item $\stronginv{\codeenv'}{\st{\tup{\procid,\pc+1,\nil},\com{\loc\mapsto \val'[p:=\val_1]},\nil,\nil}}{\inv}$

						By \Cref{tr:stron-inv} we have to prove:
						\begin{enumerate}
							\item $\weakinvinv{\codeenv'}{\st{\tup{\procid,\pc+1,\nil},\com{\loc\mapsto \val'[p:=\val_1]},\nil,\nil}}{\inv}$

							By \Cref{tr:weak-inv-inv} we have to prove that $\com{\com{\loc\mapsto \val'[p:=\val_1]}}\vdash \inv$.

							From HS2 we get that \com{\loc} is not an address with an invariant, so this point holds since \Cref{tr:inv-sat} only considers locations with an invariant.

							\item $\weakinvatk{\codeenv'}{\st{\tup{\procid,\pc+1,\nil},\com{\loc\mapsto \val'[p:=\val_1]},\nil,\nil}}{\inv}$

							By \Cref{tr:weak-inv-unchange} we need to prove
							\begin{enumerate}
								\item $\com{G_i}\notcap\com{G_a}$

								This is trivial since \com{G_a} did not change from HS2.

								\item  $\com{M_i\notcap c}$.

								From HS2 we have $\com{M_i\notcap M_a}$.

								Since \com{\loc} was already in \com{M}, HS2 tells us that \com{\loc\in M_a}, so this case holds.
							\end{enumerate}
						\end{enumerate}

						\item $\forall \com{a}\in\img{L'}\ldotp \com{a}\in\dom{M'}$

						this is trivial from HPL

						\item $\forall \com{\st{a,\stype}}\in\dom{G'}\ldotp \com{G'(\st{a,\stype})}\in\dom{M'}$

						this is trivial from HPG

						\item $\forall \com{v}\in\img{M'}\ldotp \com{v}\in\com{S}$

						this is true from HPM and definition of \com{S}.

						\item $\dom{M}\cap c = \emptyset$

						This is trivially true.
					\end{enumerate}

					\item[\Cref{tr:sem-loc-readref}]
					We have \com{i = \readrefCmd} in
					\begin{center}
						\AxiomC{
							$
							\val = \rft{\loc}{p}\sep
							\loc \in\dom{M}\sep
							\mread{M}{\loc}[p]=\val_p
							$
						}
						\UnaryInfC{
							$
							\eval{
								\st{M,L,\stackht{\val}{S}}
							}{
								\readrefCmd
							}{
								\st{M,L,\stackht{\val_p}{S}}
							}
						    $
						}
						\UnaryInfC{
							$
							\peval{
							    \st{\stackht{\tup{\procid,\pc,L}}{C}, M,G,\stackht{\val}{S}}
							}{
							    \st{\stackht{\tup{\procid,\pc+1,L}}{C}, M,G,\stackht{\val_p}{S}}
							}
							$
						}
						\DisplayProof
					\end{center}

					We need to prove:
					\begin{itemize}
						\item $\stronginv{\codeenv'}{ \st{\stackht{\tup{\procid,\pc+1,L}}{C}, M,G,\stackht{\val_p}{S}} }{\inv}$
					\end{itemize}

					By \Thmref{thm:strong-weak-unch-impl-strong-nom-nog} it suffices to prove for $\com{L'} = \nil$ and $\com{S'} = \val_p$:
					\begin{enumerate}
						\item $\stronginv{\codeenv'}{\st{\tup{\procid,\pc+1,L}{C},M,G,S}}{\inv}$, which follows from \Thmref{thm:strong-prop-monot} with HS2 since \pc+1 does not play a role in the properties.

						\item\label{proof:pt4rr} $\forall \com{\loc} \in \fun{locsof}{L'} \cup \fun{locsof}{S'}\ldotp \com{\dom{M_i} \notcap c }$ %

						This is not possible since \com{\val_p} is stored in memory.

					\end{enumerate}

					\item[\Cref{tr:sem-loc-borrowfld}]

					We have \com{i = \borrowfieldCmd} in
					\begin{center}
						\AxiomC{
							$
							\val = \rft{\loc}{p}\sep
						    \loc\in\dom{M}\sep
						    \mread{M}{\loc}[p]=\tv{\record{(f,\val_{f}),\cdots}}{t}
							$
						}
						\UnaryInfC{
							$
							\eval{
						      \st{M,L,\stackht{\val}{S}}
						    }{
						      \ao{\borrowfieldCmd}{f}
						    }{
						      \st{M,L,\stackht{\rft{\loc}{\stackht{p}{f}}}{S}}
						    }
						    $
						}
						\UnaryInfC{
							$
							\peval{
							    \st{\stackht{\tup{\procid,\pc,L}}{C}, M,G,\stackht{\val}{S}}
							}{
							    \st{\stackht{\tup{\procid,\pc+1,L}}{C}, M,G,\stackht{\rft{\loc}{\stackht{p}{f}}}{S}}
							}
							$
						}
						\DisplayProof
					\end{center}

					This follows the same structure as the case for ReadRef except for the details that we describe below.

					We need to prove \Cref{proof:pt4rr} for $\com{L'} = \nil$ and $\com{S'} = \rft{\loc}{\stackht{p}{f}}$.

					From HP on \com{\val} and from HS2 and \Cref{tr:weak-inv-unchange} we know that all attacker memory locations (\com{M_a}) do not have an invariant on ($\com{M_a}\cap\com{M_i}=\emptyset$).

					So $c$ is not a memory address with an invariant on and \Cref{proof:pt4rr} holds.

					\item[\Cref{tr:sem-loc-borrowloc}]

					We have \com{i = \borrowlocCmd} in
					\begin{center}
						\AxiomC{
							$
							\lread{L}{x} = \loc
							$
						}
						\UnaryInfC{
							$
							\eval{
						      \st{M,L,S}
						    }{
						      \ao{\borrowlocCmd}{x}
						    }{
						      \st{M,L,\stackht{\rft{\loc}{[]}}{S}}
						    }
						    $
						}
						\UnaryInfC{
							$
							\peval{
							    \st{\stackht{\tup{\procid,\pc,L}}{C}, M,G,S}
							}{
							    \st{\stackht{\tup{\procid,\pc+1,L}}{C}, M,G,\stackht{\rft{\loc}{[]}}{S}}
							}
							$
						}
						\DisplayProof
					\end{center}

					This follows the same structure as the case for ReadRef except for the details that we describe below.

					We need to prove \Cref{proof:pt4rr} for $\com{L'} = \nil$ and $\com{S'} = \rft{\loc}{[]}$.

					From HP, we have $\lread{L}{x} = \loc$ and from HS2 and \Cref{tr:weak-inv-unchange} we know that all attacker memory locations (\com{M_a}) do not have an invariant on ($\com{M_a}\cap\com{M_i}=\emptyset$).

					So $c$ is not a memory address with an invariant on and \Cref{proof:pt4rr} holds.

					\item[\Cref{tr:sem-loc-storelocref}]

					We have \com{i = \storelocCmd} in
					\begin{center}
						\AxiomC{
							$
							\val\in\Reference
							$
						}
						\UnaryInfC{
							$
							\eval{
						      \st{M,L,\stackht{\val}{S}}
						    }{
						      \ao{\storelocCmd}{x}
						    }{
						      \st{M,\lset{L}{x}{\val},S}
						    }
						    $
						}
						\UnaryInfC{
							$
							\peval{
							    \st{\stackht{\tup{\procid,\pc,L}}{C}, M,G,\stackht{\val}{S}}
							}{
							    \st{\stackht{\tup{\procid,\pc+1,\lset{L}{x}{\val}}}{C}, {M},G,{S}}
							}
							$
						}
						\DisplayProof
					\end{center}

					This follows the same structure as the case for ReadRef except for the details that we describe below.

					We need to prove \Cref{proof:pt4rr} for $\com{L'} = x\mapsto \val$ and $\com{S'} = \nil$.

					\Cref{proof:pt4rr} is simple since \com{\val} is a location that was considered in HS2, so it is disjoint from \com{M_i}.

					\item[\Cref{tr:sem-loc-storeloc}]

					We have \com{i = \storelocCmd} in
					\begin{center}
						\AxiomC{
							$
							\com{\val} \in \StoreableValue\sep
						    \com{\loc} \notin \com{\dom{M}}\sep
						    \com{M'} = \mdel{M}{\lread{L}{x}} \text{ if } \lread{L}{x} \in\dom{M} \text{ else } M\sep
							$
						}
						\UnaryInfC{
							$
							\eval{
								\st{M,L,\stackht{\val}{S}}
							}{
								\ao{\storelocCmd}{x}
							}{
								\st{\mset{M'_{\mdel{\mc{C}}{\loc}}}{\loc}{v}, \lset{L}{x}{\loc},S}
							}
						    $
						}
						\UnaryInfC{
							$
							\peval{
							    \st{\stackht{\tup{\procid,\pc,L}}{C}, M,G,\stackht{\val}{S}}
							}{
							    \st{\stackht{\tup{\procid,\pc+1,\lset{L}{x}{\loc}}}{C}, \mset{M'_{\mdel{\mc{C}}{\loc}}}{\loc}{v},G,{S}}
							}
							$
						}
						\DisplayProof
					\end{center}

					We need to prove:
					\begin{itemize}
						\item $\stronginv{\codeenv'}{ \st{\stackht{\tup{\procid,\pc+1,\lset{L}{x}{\loc}}}{C}, \mset{M'_{\mdel{\mc{C}}{\loc}}}{\loc}{v},G,{S}} }{\inv}$
					\end{itemize}

					By \Thmref{thm:strong-prop-monot} it suffices to prove:
					\begin{itemize}
						\item $\stronginv{\codeenv'}{ \st{\stackht{\tup{\procid,\pc+1,\lset{L}{x}{\loc}}}{C}, \mset{M'_{\mdel{\mc{C}}{\loc}}}{\loc}{v},G,\stackht{\val}{S}} }{\inv}$
					\end{itemize}

					By \Thmref{thm:strong-prop-compos} with HPL $\com{L'}=\com{x\mapsto \loc}$, HPM $\com{M'}=\com{\loc\mapsto \val}$, HPG $\com{G'}=\nil$ it suffices to prove
					\begin{enumerate}
						\item $\stronginv{\codeenv'}{\st{\tup{\procid,\pc+1,\mdel{L}{x\mapsto \loc}}{C},\mdel{M}{\loc\mapsto v},G,\stackht{\val}{S}}}{\inv}$, which follows from \Thmref{thm:strong-prop-monot} with HS2 since \pc+1 does not play a role in the properties.

						\item $\stronginv{\codeenv'}{\st{\tup{\procid,\pc+1,x\mapsto \loc},\nil,\loc\mapsto v,\nil}}{\inv}$

						By \Cref{tr:stron-inv} we have to prove:
						\begin{enumerate}
							\item $\weakinvinv{\codeenv'}{\st{\tup{\procid,\pc+1,x\mapsto \loc},\nil,\loc\mapsto v,\nil}}{\inv}$

							By \Cref{tr:weak-inv-inv} we have to prove that $\com{\loc\mapsto v}\vdash \inv$.

							By \Cref{tr:inv-sat} we have that $\com{\loc}\notin\com{M_i}$ because there are no globals pointing to \com{\loc} since \com{\loc} is fresh, so this case holds.

							\item $\weakinvatk{\codeenv'}{\st{\tup{\procid,\pc+1,x\mapsto \loc},\nil,\loc\mapsto v,\nil}}{\inv}$

							By \Cref{tr:weak-inv-unchange} we need to prove
							\begin{enumerate}
								\item $\com{G_i\notcap G_a}$ this is trivial since \com{G_a} does not change from HS2.

								\item $\com{M_i}\notcap\com{\loc}$

								This follows from the freshess of \com{\loc}.

							\end{enumerate}
						\end{enumerate}

						\item $\forall \com{a}\in\img{L'}\ldotp \com{a}\in\dom{M'}$

						this is trivial from HPL

						\item $\forall \com{\st{a,\stype}}\in\dom{G'}\ldotp \com{G'(\st{a,\stype})}\in\dom{M'}$

						this is trivial from HPG

						\item $\forall \com{v}\in\img{M'}\ldotp \com{v}\in\com{S}$

						this is true from HPM and definition of \com{S}.

						\item $\dom{M}\cap c = \emptyset$

						This is trivially true.
					\end{enumerate}

					\item[\Cref{tr:sem-loc-copylocref}]

					We have \com{i = \copylocCmd} in
					\begin{center}
						\AxiomC{
							$
							\val=\lread{L}{x} = \rft{\loc}{p}
							$
						}
						\UnaryInfC{
							$
							\eval{
						      \st{M,L,S}
						    }{
						      \ao{\copylocCmd}{x}
						    }{
						      \st{M,L,\stackht{\val}{S}}
						    }
						    $
						}
						\UnaryInfC{
							$
							\peval{
							    \st{\stackht{\tup{\procid,\pc,L}}{C}, M,G,S}
							}{
							    \st{\stackht{\tup{\procid,\pc+1,L}}{C}, {M},G,\stackht{\val}{S}}
							}
							$
						}
						\DisplayProof
					\end{center}

					This follows the same structure as the case for ReadRef except for the details that we describe below.

					We need to prove \Cref{proof:pt4rr} for $\com{L'} = \nil$ and $\com{S'} = \val$.

					\Cref{proof:pt4rr} is simple because its locations were considered in HS2.

					\item[\Cref{tr:sem-loc-copyloc}]

					We have \com{i = \copylocCmd} in
					\begin{center}
						\AxiomC{
							$
							\lread{L}{x}=\loc
              \sep
							\loc\in\dom{M}
							$
						}
						\UnaryInfC{
							$
							\eval{
							  \st{M,L,S}
							}{
							  \ao{\copylocCmd}{x}
							}{
							  \st{M,L,\stackht{\mread{M}{\loc}}{S}}
							}
						    $
						}
						\UnaryInfC{
							$
							\peval{
							    \st{\stackht{\tup{\procid,\pc,L}}{C}, M,G,S}
							}{
							    \st{\stackht{\tup{\procid,\pc+1,L}}{C}, {M},G,\stackht{\mread{M}{\loc}}{S}}
							}
							$
						}
						\DisplayProof
					\end{center}

					This follows the same structure as the case for ReadRef except for the details that we describe below.

					We need to prove \Cref{proof:pt4rr} for $\com{L'} = \nil$ and $\com{S'} = M(\loc)$.

					\Cref{proof:pt4rr} is trivial since there are no locations.

					\item[\Cref{tr:sem-loc-movelocref}]

					We have \com{i = \movelocCmd} in
					\begin{center}
						\AxiomC{
							$
							\lread{L}{x} = \rft{\loc}{p}
							$
						}
						\UnaryInfC{
							$
							\eval{
						      \st{M,L,S}
						    }{
						      \ao{\movelocCmd}{x}
						    }{
						      \st{M,\mdel{L}{x},\stackht{\lread{L}{x}}{S}}
						    }
						    $
						}
						\UnaryInfC{
							$
							\peval{
							    \st{\stackht{\tup{\procid,\pc,L}}{C}, M,G,S}
							}{
							    \st{\stackht{\tup{\procid,\pc+1,\mdel{L}{x}}}{C}, {M},G,\stackht{\lread{L}{x}}{S}}
							}
							$
						}
						\DisplayProof
					\end{center}

					This follows the same structure as the case for ReadRef except for the details that we describe below.

					We need to prove \Cref{proof:pt4rr} for $\com{L'} = \nil$ and $\com{S'} = L(x)$.

					\Cref{proof:pt4rr} is simple because its locations were considered in HS2.

					\item[\Cref{tr:sem-loc-moveloc}]

					We have \com{i = \movelocCmd} in
					\begin{center}
						\AxiomC{
							$
							\lread{L}{x} = \loc\sep
							\loc\in\dom{M}
							$
						}
						\UnaryInfC{
							$
							\eval{
						      \st{M,L,S}
						    }{
						      \ao{\movelocCmd}{x}
						    }{
						      \st{\mdel{M}{\loc},\mdel{L}{x},\stackht{\mread{M}{\loc}}{S}}
						    }
						    $
						}
						\UnaryInfC{
							$
							\peval{
							    \st{\stackht{\tup{\procid,\pc,L}}{C}, M,G,S}
							}{
							    \st{\stackht{\tup{\procid,\pc+1,\mdel{L}{x}}}{C}, \mdel{M}{\loc},G,\stackht{\mread{M}{\loc}}{S}}
							}
							$
						}
						\DisplayProof
					\end{center}

					This follows the same structure as the case for ReadRef except for the details that we describe below.

					We need to prove \Cref{proof:pt4rr} for $\com{L'} = \nil$ and $\com{S'} = M(\loc)$.

					\Cref{proof:pt4rr} is trivial since there are no locations.
				\end{itemize}
			\end{itemize}
	\end{itemize}
\end{proof}

High-level intuition: when the attacker creates new memory locations (Moveto, Storeloc), they are fresh and therefore there cannot be any existing invariant on them.
When the attacker creates a new global (Moveto), it can only be with an attacker type, and thus it cannot be a global location with an invariant on.
Any literal that the attacker creates (Op, LoadConst) can be the address part of a global, and that is ok.
When that literal will be used, its type part makes it so that the global accessed with that address does not have invariants on.
All other rules move existing values around, package them in structures and access the fields of those structures, so they do not alter the attacker capabilities.

The attacker also cannot craft locations nor globals that the \tc will set an invariant on: the globals created by \tc will not have attacker types but \tc types, and the only way to pass a location is via globals, but the contents of globals are never read direcly, globals are just used as a proxy to those memory locations they point to, so the attacker-created location cannot reach \tc because it is not stored in the memory.

One concern may be: what if an attacker crafts an address, and then passes it as a parameter to the \tc, which then uses it naively and violates its invariants?
This is addressed by the local verification and the encapsulation.

\BREAK

The next lemma tells that reductions in the \tc (\Thmref{thm:rs-comp-red}) respect the strong property.
\begin{lemma}[\tc Reductions Respect the Strong Property]\label{thm:rs-comp-red}
	\begin{align*}
		\text{ if }
		&\
		\com{\codeenv' \triangleright \codeenv,\procid} \vdash \com{\sigma \Xto{\alpha!} \sigma'}
		\\
		\text{ and } &
			\agree{\codeenv'}{\inv}
		\\
		\text{ and }
		&\
		\com{\codeenv'}\vdash\com{\sigma}:\com{modcode}
		\\
		\text{ and }
		&\
		\stronginv{\codeenv}{\sigma}{\inv}
		\\
		\text{ then }
		&\
		\stronginv{\codeenv'}{\sigma'}{\inv}
		\\
		\text{ and }
		&\
		\com{\alpha!}\Vdash\com{\inv}
	\end{align*}
\end{lemma}
\begin{proof}[Proof of \Thmref{thm:rs-comp-red}]\hfill

	By \Cref{def:loc-inv} we have $\weakinvinv{\codeenv'}{\sigma_1}{\inv}$ (HW1) and $\com{\alpha!}\Vdash\com{\inv}$, the second conjunct of the thesis, which is now proven and we need to prove just the first one.

	By \Cref{def:encapsulated} we have $\weakinvatk{\codeenv'}{\sigma_2}{\inv}$ (HW2).

	By \Thmref{thm:sem-det} we conclude that the final states are the same so $\com{\sigma_1}=\com{\sigma_2}=\com{\sigma'}$.

	By \Thmref{thm:wk-inv-impl-str} with HW1 and HW2 we have $\stronginv{\codeenv'}{\sigma'}{\inv}$.
\end{proof}
\BREAK

An additional result we need is that semantics are deterministic.
\begin{lemma}[Determinism of the Semantics]\label{thm:sem-det}
	\begin{align*}
		\text{ if }
			&
			\com{\codeenv' \triangleright \codeenv,\procid} \vdash \com{\sigma \Xto{\alpha!} \sigma'}
		\\
		\text{ and }
			&
			\com{\codeenv' \triangleright \codeenv,\procid} \vdash \com{\sigma \Xto{\alpha!} \sigma''}
		\\
		\text{ then }
			&
			\com{\sigma'}=\com{\sigma''}
	\end{align*}
\end{lemma}
\begin{proof}[Proof of \Cref{thm:sem-det}]\hfill

	Follows directly from the semantics being deterministic.
\end{proof}
\BREAK

\subsubsection{Main Result}
\begin{lemma}[Generalised RS]\label{thm:rs-general}
	\begin{align*}
		&\
		\forall \com{\atk},\com{\procid},\com{\OB{\alpha}},\com{\sigma'}\ldotp
		\\
		\text{ if } &
			\com{\codeenv} \vdash \com{\atk} : \com{atk}
		\\
		\text{ and } &
			\agree{\codeenv'}{\inv}
		\\
		\text{ and } &
			\stronginv{\codeenv}{\sigma}{\inv}
		\\
		\text{ and } &
			\com{\codeenv' \triangleright \codeenv,\procid} \vdash \com{\sigma \Xtol{\OB{\alpha}} \sigma'}
			\\
			\text{ then } &
			\stronginv{\codeenv}{\sigma'}{\inv}
			\\
			\text{ and }
			&
			\com{\OB{\alpha}}\Vdash\com{\inv}:\com{global}
	\end{align*}
\end{lemma}
The conditions on $\procid$, $i$ and $\codeenv'$ are implicit: if the configuration steps accroring to $\Xtol{}$, then $\procid$ is a valid procedure in $\codeenv'$ and $i$ is a valid pc location in $\procid$ (i.e., it is either the start of a function or the address right after a call).
\begin{proof}[Proof of \Thmref{thm:rs-general}]\hfill

	This proof proceeds by induction on $\Xtol{}$.
	\begin{itemize}
		\item[Base, \Cref{tr:trace-re}:]
			This case trivially holds by \Cref{tr:trace-glob-base} since the generated action is \com{\nil}.
		\item[Inductive:]
			Depending on the rule, we have two cases:
			\begin{itemize}
				\item[\Cref{tr:trace-sing2}]
					so we have:
					\begin{itemize}
						\item[$\com{\codeenv' \triangleright \codeenv,\procid} \vdash \com{\sigma \Xtol{\OB{\alpha'}} \sigma''}$]

						By IH we have $\stronginv{\codeenv'}{\sigma''}{\inv}$ (HB) and $\com{\OB{\alpha'}}\Vdash\com{\inv}:\com{global}$ (HT).

						\item[$\com{\codeenv' \triangleright \codeenv,\procid} \vdash \com{\sigma'' \Xto{\alpha?} \sigma'''}$]

						By \Thmref{thm:rs-att-red} with HB we have $\stronginv{\codeenv'}{\sigma'''}{\inv}$ (HA) and $\com{\alpha?}\Vdash\com{\inv}$ (HT1).

						\item[$\com{\codeenv' \triangleright \codeenv,\procid} \vdash \com{\sigma''' \Xto{\alpha!} \sigma'}$]

						By \Thmref{thm:rs-comp-red} with HA we have $\stronginv{\codeenv'}{\sigma'}{\inv}$ and $\com{\alpha!}\Vdash\com{\inv}$ (HT2).

						By two applications of \Cref{tr:trace-glob-ind} with HT and HT1 first, and then with HT2, the thesis holds.
					\end{itemize}
				\item[\Cref{tr:trace-sing1}]
					This case is analogous to half of the previous one.
			\end{itemize}
	\end{itemize}
\end{proof}
\BREAK

\newpage %
\subsection{Examples}

\begin{lstlisting}
module M {
  // invariant: f > 0
  resource struct Counter { f: u64 }

  public fun create(): Counter {
    Counter { f: 1 }
  }

  public fun add(account: &signer, c: Counter) {
    move_to(account, c)
  }

  public fun remove(a: address): Counter {
    move_from(a)
  }

  public fun read(c: &Counter): &u64 {
    &c.f
  }

  public fun increment(c: &mut Counter) {
    c.f = *c.f + 1
  }

  // Violates the encapsulator!
  public fun read_mut(c: &mut Counter): &mut u64 {
    &mut c.f
  }
}
\end{lstlisting}

In the code above, a modular static analysis can prove that the invariant \code{f > 1} holds locally for all \code{Counter} instances,
but the invariant does \emph{not} hold in the presence of attacker code because of the \code{read_mut} function. Specifically, the attacker can do

\begin{lstlisting}
let c = M::create();
let x = M::read_mut(&mut c);
*x = 0;
\end{lstlisting}

The encapsulator must reject \code{read_mut} because it ``leaks'' the internal state of the \code{Counter} resource. If this function is
removed, the \code{f > 1} invariant will hold globally.

\newpage
\section{Encapsulator Instances as Intraprocedural Escape Analyses}\label{sec:encap-app}

In this section, we define a simple intraprocedural escape analysis that satisfies \Cref{def:encapsulated}.
The analysis abstracts the concrete values bound to local variables and stack locations using a lattice with three abstract values: $\noref$, $\okref$, $\inref$.
We define $\noref \sqsubseteq \inref$ and $\okref \sqsubseteq \inref$.
Inuitively, $\noref$ represents any non-reference value, $\okref$ represents a reference that does not point inside of a resource defined in the current module, and $\inref$ represents a reference that \emph{may} point inside a resource defined in the current module.
The purpose of the analysis is to prevent an $\inref$ from ``leaking'' to a caller of the module via a $\returnCmd$.
Because Move structs cannot store references, this is the only way such a leak can occur.

\subsubsection*{High-Level Description}
The analysis does the following:
\begin{itemize}
\item Initialize each reference parameter to $\okref$ and each non-reference parameter to $\noref$
\item Applying $\borrowfieldCmd$ to a $\okref$ value produces an $\inref$ value
\item A $\borrowglobalCmd$ produces an $\inref$ value
\item A $\borrowlocCmd$ produces an $\okref$ value
\item Applying a $\returnCmd$ instruction to an $\inref$ is an error
\item A $\callCmd$ marks all non-reference return values as $\noref$. If the call accepts any $\inref$ arguments, all reference return values are marked as $\inref$. Otherwise, all reference return values are marked as $\okref$.
\item All other instructions either propagate their input values or introduce $\noref$ values
\end{itemize}

We note that returning references to locals (e.g., \code{let x = param; return &x} and globals (e.g., \code{return borrow_global<T>(a)}) is already ruled out by the Move bytecode verifier and thus cannot leak an internal reference.
However, we define our escape analysis as if these behaviors were possible in order to minimize dependence on the bytecode verifier invariants in our proof.

\subsection{Encapsulator $\ea$: an Escape Analysis for Integrity}\label{sec:ea-rules}
We use the symbol $\ea$ to indicate an escape analysis that can be used as an encapsulator.
In the rules below, we omit the $\codeenv,\procid,\inv$ on the left of the $\vdash$ when they are not necessary.
This allows us to focus the escape analysis only on fields that are relevant to the key safety invariants.
\begin{center}
  \typerule{BorrowFld-InvRelevant}{
    f \in \inv
  }{
    \aeval{
      \inv,
      \ao{\borrowfieldCmd}{f}
    }{
      \st{\hat{L},\stackht{\hat{v}}{\hat{S}}}
    }{
      \st{\hat{L}, \stackht{\inref}{\hat{S}}}
    }
  }{encapsulator-borrowfld-attacker}
  \typerule{BorrowFld-InvIrrelevant}{
    f \notin \inv
  }{
    \aeval{
      \inv,
      \ao{\borrowfieldCmd}{f}
    }{
      \st{\hat{L},\stackht{\hat{v}}{\hat{S}}}
    }{
      \st{\hat{L}, \stackht{\hat{v}}{\hat{S}}}
    }
  }{encapsulator-borrowfld}
  \typerule{BorrowGlobal}{
  }{
    \aeval{
      \ao{\borrowglobalCmd}{s}
    }{
      \st{\hat{L}, \stackht{\hat{v}}{\hat{S}}}
    }{
      \st{\hat{L}, \stackht{\inref}{\hat{S}}}
    }
  }{encapsulator-borrowglobal}
  \typerule{BorrowLoc}{
  }{
    \aeval{
      \ao{\borrowlocCmd}{x}
    }{
      \st{\hat{L},\hat{S}}
    }{
      \st{\hat{L},\stackht{\okref}{\hat{S}}}
    }
  }{encapsulator-borrowloc}
  \typerule{Return}{
    \card{\codeenv(P).3} = n
    &
    \forall i \in 1..n\ldotp
    \hat{v_i} \neq \inref
  }{
    \aeval{
      \returnCmd
    }{
      \st{\hat{L}, \stackht{\hat{v_1}\listsep\hat{v_n}}{\hat{S}}}
    }{
      \st{\hat{L}, \stackht{\hat{v_1}\listsep\hat{v_n}}{\hat{S}}}
    }
  }{encapsulator-return}
  \typerule{Call}{
    \card{\codeenv(\procid_0).type} = n
    &
    \card{\codeenv(\procid_0).3} = j
    \\
    \forall i \in 1 .. j\ldotp
    \hat{v_i^r} =
       (\noref \text { if } \codeenv(\procid).type(i) \neq \texttt{ref})
       \\ \vee
    \hat{v_i^r} =
       (\okref \text { if } \codeenv(\procid).type(i) = \texttt{ref} \wedge \forall m\in 1 .. n, v_m^a \neq \inref )
       \\ \vee
    \hat{v_i^r} =
       (\inref \text { if } \codeenv(\procid).type(i) = \texttt{ref} \wedge \exists m\in 1 .. n, v_m^a = \inref )
  }{
    \aeval{
      \codeenv, \procid, \inv, \ao{\callCmd}{\procid_0}
    }{
       \st{\hat{L},\stackht{\hat{v_1^a}\cdots\hat{v_n^a}}{\hat{S}}}
    }{
      \st{\hat{L},\stackht{\hat{v_1^r}\cdots\hat{v_j^r}}{\hat{S}}}
    }
  }{encapsulator-call}
  \typerule{MoveLoc}
  {
      \lread{\hat{L}}{x} = \hat{v}
  }{
    \aeval{
      \ao{\movelocCmd}{x}
    }{
      \st{\hat{L},\hat{S}}
    }{
      \st{\mdel{\hat{L}}{x},\stackht{\hat{v}}{\hat{S}}}
    }
  }{encapsulator-moveloc}
  \typerule{CopyLoc}
  {
    \lread{\hat{L}}{x}=\hat{v}
  }{
    \aeval{
      \ao{\copylocCmd}{x}
    }{
      \st{\hat{L},\hat{S}}
    }{
      \st{\hat{L},\stackht{\hat{v}}{\hat{S}}}
    }
  }{encapsulator-copyloc}
  \typerule{StoreLoc}
  {
  }{
    \aeval{
      \ao{\storelocCmd}{x}
    }{
      \st{\hat{L},\stackht{\hat{\val}}{\hat{S}}}
    }{
      \st{\lset{\hat{L}}{x}{\hat{v}},\hat{S}}
    }
  }{encapsulator-storeloc}
  \typerule{WriteRef}
  {
  }{
    \aeval{
      \writerefCmd
    }{
      \st{\hat{L},\stackht{\hat{\val_1}}{\stackht{\hat{\val_2}}{\hat{S}}}}
    }{
      \st{\hat{L},\hat{S}}
    }
  }{encapsulator-writeref}
  \typerule{ReadRef}
  {
  }{
    \aeval{
      \readrefCmd
    }{
      \st{\hat{L},\stackht{\hat{\val}}{\hat{S}}}
    }{
      \st{\hat{L},\stackht{\textsc{Ok}}{\hat{S}}}
    }
  }{encapsulator-readref}
  \typerule{Pop}
  {}{
    \aeval{
      \popCmd
    }{
      \st{\hat{L},\stackht{\hat{v}}{\hat{S}}}
    }{
      \st{\hat{L},\hat{S}}
    }
  }{encapsulator-pop}
  \typerule{LoadConst}
  {
  }{
    \aeval{
      \ao{\loadconstCmd}{v}
    }{
      \st{\hat{L},\hat{S}}
    }{
      \st{\hat{L},\stackht{\noref}{\hat{S}}}
    }
  }{encapsulator-loadconst}
  \typerule{Op}
  {
  }{
    \aeval{
      \stackopCmd
    }{
      \st{\hat{L},\stackht{\hat{v_1}}{\stackht{\hat{v_2}}{\hat{S}}}}
    }{
      \st{\hat{L},\stackht{\noref}{\hat{S}}}
    }
  }{encapsulator-op}
  \typerule{MoveTo}
  {
  }{
    \aeval{
      \ao{\movetoCmd}{s}
    }{
      \st{\hat{L},\stackht{\hat{v_1}}{\stackht{\hat{v_2}}{\hat{S}}}}
    }{
      \st{\hat{L}, \hat{S}}
    }
  }{encapsulator-moveto}
  \typerule{MoveFrom}
  {
  }{
    \aeval{
      \ao{\movefromCmd}{s}
    }{
      \st{\hat{L},\stackht{\hat{v}}{\hat{S}}}
    }{
      \st{\hat{L},\stackht{\noref}{\hat{S}}}
    }
  }{encapsulator-movefrom}
  \typerule{Pack}
  {
  }{
    \aeval{
      \ao{\packCmd}{s}
    }{
      \st{\hat{L},\stackht{\hat{\val_{1}}\cons\cdots\cons \hat{\val_{n}}}{\hat{S}}}
    }{
      \st{\hat{L},\stackht{\noref}{\hat{S}}}
    }
  }{encapsulator-pack}
  \typerule{Unpack}
  {
  }{
    \aeval{
      \ao{\unpackCmd}{s}
    }{
      \st{\hat{L},\stackht{\hat{v}}{\hat{S}}}
    }{
      \st{\hat{L},\stackht{\hat{\val_{1}}\cons\cdots\cons \hat{\val_{n}}}{\hat{S}}}
    }
  }{encapsulator-unpack}
  \typerule{Branch}
  {
  }{
    \aeval{
      \ao{\branchCmd}{\pc}
    }{
      \st{\hat{L},\hat{S}}
    }{
      \st{\hat{L},\hat{S}}
    }
  }{encapsulator-branch}
  \typerule{BranchCond}
  {
  }{
    \aeval{
      \ao{\branchcondCmd}{\pc}
    }{
      \st{\hat{L},\stackht{\hat{v}}{\hat{S}}}
    }{
      \st{\hat{L},\hat{S}}
    }
  }{encapsulator-branch-cond}

  \smallskip
  \hrule
  \smallskip
  \hrule

  \typerule{Module-instruction list}{
    \aeval{\codeenv, \procid, \procid(\pc)}{\st{\hat{L}, \hat{S}}}{\st{\hat{L''}, \hat{S''}}}
    &
    \aevalfat{\codeenv, \procid}{\st{\pc+1,\hat{L''}, \hat{S''}}}{\st{\pc',\hat{L'}, \hat{S'}}}
  }{
    \aevalfat{\codeenv, \procid}{\st{\pc,\hat{L}, \hat{S}}}{\st{\pc',\hat{L'}, \hat{S'}}}
  }{encapsulator-trans}

  \smallskip
  \hrule
  \smallskip
  \hrule

  \typerule{Module-top}{
    \forall \procid \in \codeenv'.
    \aevalfat{\codeenv', \procid}{\st{0, \totypes{\codeenv(\procid).1} , \hat{\nil}}}{\st{ \card{\codeenv(\procid).code} ,\hat{L'}, \totypes{ \codeenv(\procid).2 } }}
  }{
    \com{\modul}(\codeenv')
  }{enc-mod}
\end{center}
We indicate the size of the list of instructions in a procedure declaration as \card{\cdot}.
Recall that in a procedure declaration, the first type is the input type and the second one is the return type.

The analysis in the branch case goes through the whole code, essentially first passing over the else branch, and then over the then branch.
Since here the control flow does not do a branch, our analysis' linear pass checks both branches.

This auxiliary functions returns the abstract values for all possible types.
\begin{align*}
  \totypes{z} =&\ \noref   & z\in\mi{GroundType} \vee z\in\mi{RecordType} \vee z\in\mi{StorableType}
  \\
  \totypes{a} =&\ \okref   & a\in\mi{ReferenceType}
\end{align*}

Below is the concretisation function that takes an abstract state \com{\st{\hat{S}, \hat{L}}} and returns a set of possible states that realise that state.
This, in turn, relies on a concretisation function for abstract values \com{\hat{v}} to a set of possible values.
Intuitively, $\noref$ concretizes to a non-reference value, $\okref$ concretizes to a reference value that does not contain a field offset defined in the current module, and $\inref$ concretizes to a reference value with one or more offsets defined in the current module.
\begin{align*}
  \concfun{\st{\hat{S}, \hat{L}},\codeenv} =
    &\
    \myset{ \st{C,M,G,S} }{
      \begin{aligned}
        &
        C = \st{\procid,\pc,\concfun{\hat{L},M,G,\codeenv}}\listsep C'
        \\
        &
        \text{ and } S = \concfun{\hat{S},M,G,\codeenv}
      \end{aligned}
    }
  \\
    \concfun{\hat{\nil},M,G,\codeenv} =
    &\
    \nil
  \\
  \concfun{\hat{v}\listsep \hat{S},M,G,\codeenv} =
    &\
    \concfun{\hat{v},M,G,\codeenv}\listsep\concfun{\hat{S},M,G,\codeenv}
  \\
  \concfun{x\mapsto\hat{v}\listsep\hat{L},M,G,\codeenv} =
    &\
    x\mapsto\concfun{\hat{v},M,G,\codeenv}\listsep\concfun{\hat{L},M,G,\codeenv}
  \\
    \concfun{\noref,M,G,\codeenv} =
    &\
    \myset{ v }{ v \in \mi{StorableValue}}
  \\
  \concfun{\okref,M,G,\codeenv} =
    &\
    \myset{ \loc }{
      \begin{aligned}
        &
        \loc \in \mi{Reference} \text{ and } \loc \in \dom{M} \text{ and }
        \\
        &
        M(\loc) = a \ldotp \forall \rho\in\codeenv.types . \st{a,\rho}\notin G
      \end{aligned}
      }
  \\
  \concfun{\inref,M,G,\codeenv} =
    &\
    \myset{ \loc }{
      \begin{aligned}
        &
        \loc \in \mi{Reference} \text{ and } \loc \in \dom{M} \text{ and }
        \\
        &
        M(\loc) = a \ldotp \exists \rho\in\codeenv.types . \st{a,\rho}\in G
      \end{aligned}
    }
\end{align*}

\subsubsection{Auxiliaries for the Encapsulator}\label{sec:aux-encaps}

\paragraph{State \pc}
This auxiliary function returns the current \com{\pc}:
\begin{align*}
  \com{\sigma.\pc} = n & \text{ if } \com{\sigma} = \com{\st{\st{P,n,L}\listsep C,M,G,S}}
\end{align*}

\paragraph{Entry-Exit Reductions}

We say that a reduction is \com{entry/exit} if it starts from the entry point of a function until the \emph{next} exit point (\Cref{tr:entryexit}).
We identify entry points (\Cref{tr:entry}) to be:
\begin{itemize}
	\item the first instruction;
	\item any instruction after a call;
\end{itemize}
while exit points (\Cref{tr:exit}) are:
\begin{itemize}
	\item call instructions;
	\item the last instruction.
\end{itemize}
Recall that we indicate the size of a stack (e.g., the call stack \com{C}) with \com{\card{C}}.

\begin{center}
	\typerule{Entry}{
		(\sigma.\pc = 0) \vee
		(\codeenv(\sigma.\pc-1) = \ao{\callCmd}{p})
	}{
		\codeenv,\procid \vdash \com{\sigma} : \com{entry}
	}{entry}
	\typerule{Exit}{
		(\sigma.\pc = \card{\codeenv(\procid).\mi{code}}) \vee
		(\codeenv(\sigma.\pc) = \ao{\callCmd}{p}) \vee
		(\codeenv(\sigma.\pc) = \returnCmd)
	}{
		\codeenv,\procid \vdash \com{\sigma} : \com{exit}
	}{exit}
	\typerule{Entry/Exit}{
		\forall \sigma''\ldotp
		\com{\codeenv' \triangleright \codeenv,\procid} \vdash \com{\sigma \Xto{\nil} \sigma'' \Xto{\_} \sigma' }
		\\
		\card{\sigma''.C} = \card{\sigma.C} = \card{\sigma'.C}
		\\
		\codeenv,\procid \vdash \com{\sigma} : \com{entry}
		&
		\codeenv,\procid \vdash \com{\sigma'} : \com{exit}
	}{
		\vdash \com{\codeenv' \triangleright \codeenv,\procid} \vdash \com{\sigma \Xto{\_} \sigma' } : \com{entry/exit}
	}{entryexit}
	\typerule{Exit/Entry}{
		\codeenv,\procid \vdash \com{\sigma} : \com{exit}
		&
		\codeenv,\procid \vdash \com{\sigma'} : \com{entry}
	}{
		\vdash \com{\codeenv' \triangleright \codeenv,\procid} \vdash \com{\sigma \xto{\nil} \sigma' } : \com{exit/entry}
	}{exitentry}
\end{center}

\begin{definition}[Same-Domain Reductions can be Split into Entry/Exit Reductions]\label{def:fat-red-split-entryexit-red}
	\begin{align*}
		\text{ if }
			&\
			\com{\codeenv' \triangleright \codeenv,\procid} \vdash \com{\sigma \Xto{\alpha!} \sigma'}
		\\
		\text{ then }
			&\
			\exists n \ldotp
			\com{\codeenv' \triangleright \codeenv,\procid} \vdash \com{\sigma_0' \Xto{\nil} \sigma_1 \xto{\nil} \sigma_1' \Xto{\nil} \cdots \xto{\nil} \sigma_n' \Xto{\nil} \sigma_{n+1} \xto{\alpha!}\sigma'}
		\\
		\text{ and }
			&\
			\com{\sigma} = \com{\sigma_0'}
		\\
		\text{ and }
			&\
			\forall i \in 0 .. n/2 \ldotp
			\vdash \com{\codeenv' \triangleright \codeenv,\procid} \vdash \com{\sigma_{2i}' \Xto{\nil} \sigma_{2i+1} } : \com{entry/exit}
		\\
		\text{ and }
			&\
			\vdash \com{\codeenv' \triangleright \codeenv,\procid} \vdash \com{\sigma_{2i+1} \xto{\nil} \sigma_{2i+1}'' } : \com{exit/entry}
	\end{align*}
\end{definition}

\paragraph{Entry-Exit Encapsulation}
Given that a code environment is encapsulated, we know that all its functions are also encapsulated from their entry to their exit points.

As for semantic reductions, we define what it means for encapsulator reductions to be \com{entry/exit}.
This holds if the reduction starts from a state that concretizes to an entry point of a function and ends in a state that concretizes into the \emph{next} exit point (\Cref{tr:entryexit-enc}).
\begin{center}
	\typerule{Entry/Exit Encapsulator}{
		\forall \st{\pc'',\hat{L''},\hat{S''}} \ldotp
		\aevalfat{\codeenv, \procid}{\st{\pc,\hat{L}, \hat{S}}}{ \st{\pc'',\hat{L''},\hat{S''}} \rightsquigarrowfat  \st{\pc',\hat{L'}, \hat{S'}}}
		\\
		\nvdash \codeenv,\procid \vdash \com{\st{\pc'',\hat{L''},\hat{S''}}} : \com{exit}
		\\
		\codeenv,\procid \vdash \com{\st{\pc,\hat{L},\hat{S}}} : \com{entry}
		&
		\codeenv,\procid \vdash \com{\st{\pc',\hat{L'},\hat{S'}}} : \com{exit}
	}{
		\vdash \aevalfat{\codeenv, \procid}{\st{\pc,\hat{L}, \hat{S}}}{\st{\pc',\hat{L'}, \hat{S'}}} : \com{entry/exit}
	}{entryexit-enc}
\end{center}

\begin{definition}[Whole-Program Encapsulation Entails Entry/Exit Encapsulation]\label{def:whole-enc-entryexit-enc}
	\begin{align*}
		\text{ if }
			&\
    		\aevalfat{\codeenv, \procid}{\st{\pc, \hat{L} , \hat{S}}}{\st{ \pc' ,\hat{L'}, \hat{S'} }}
    	\\
    	\text{ then }
    		&\
    		\exists m \ldotp
    		\aevalfat{\codeenv, \procid}{\st{\pc_0, \hat{L_0} , \hat{S_0}}}{ \st{\pc_1,\hat{L_1},\hat{S_1}} \rightsquigarrowfat\cdots\rightsquigarrowfat \st{ \pc_m ,\hat{L_m}, \hat{S_m} }\rightsquigarrowfat \st{ \pc' ,\hat{L'}, \hat{S'} } }
    	\\
    	\text{ and }
    		&\
    		\st{\pc, \hat{L} , \hat{S}} = \st{\pc_0, \hat{L_0} , \hat{S_0}}
    	\\
    	\text{ and }
    		&\
    		\forall j \in 0 .. m/2 \ldotp
    		\vdash \aevalfat{\codeenv, \procid}{\st{\pc_j,\hat{L_j}, \hat{S_j}}}{\st{\pc_{2j+1},\hat{L_{2j+1}}, \hat{S_{2j+1}}}} : \com{entry/exit}
	\end{align*}
\end{definition}

\subsubsection{Properties of \com{\ea}}

\begin{theorem}[The Escape Analysis is a Sound Encapsulator (\encapsulated{\com{\codeenv'}}{\ea}{\inv})]\label{thm:escape-ok}
	\begin{align*}
		\text{ if }
			&
			\stronginv{\codeenv'}{\sigma}{\inv}
      	\\
      	\text{ and }
      		&
  		  	\com{\codeenv' \triangleright \codeenv,\procid} \vdash \com{\sigma \Xto{\alpha!} \sigma'}
        \\
      	\text{ and }
        	&
        	\com{\ea}(\restr{\codeenv'}{\inv})
      	\\
      	\text{ then }
      		&
        	\weakinvatk{\codeenv'}{\sigma'}{\inv}
	\end{align*}
\end{theorem}
\begin{proof}
	By \Cref{tr:stron-inv} with HP0 we have HPw $\weakinvatk{\codeenv'}{\sigma}{\inv}$.

	By the semantics rules with HP2 we have:
	\begin{itemize}
		\item HPRF $\com{\codeenv' \triangleright \codeenv,\procid} \vdash \com{\sigma \Xto{\nil} \sigma''}$
		\item HPRA $\com{\codeenv' \triangleright \codeenv,\procid} \vdash \com{\sigma'' \xto{\alpha!} \sigma'}$
	\end{itemize}
	Additionally, by \Cref{tr:entry} and \Cref{tr:exit} we have:
	\begin{itemize}
		\item HPEN $\codeenv,\procid \vdash \com{\sigma} : \com{entry}$
		\item HPEX $\codeenv,\procid \vdash \com{\sigma''} : \com{exit}$
	\end{itemize}

	By HP3 with \Thmref{def:whole-enc-entryexit-enc} we can split the encapsulation of the whole code into entry/exit encapsulations, so we have:
	\begin{itemize}
		\item $\vdash \aevalfat{\codeenv, \procid}{\st{\sigma.\pc,\hat{L}, \hat{S}}}{\st{\sigma_1.\pc,\hat{L_1}, \hat{S_1}}} : \com{entry/exit}$

		This gives us HPS $\aevalfat{\codeenv, \procid}{\st{\sigma.\pc,\hat{L}, \hat{S}}}{\st{\sigma_1.\pc,\hat{L_1}, \hat{S_1}}}$

		\item $\vdash \aevalfat{\codeenv, \procid}{\st{\sigma_2.\pc,\hat{L_2}, \hat{S_2}}}{\st{\sigma''.\pc,\hat{L''}, \hat{S''}}} : \com{entry/exit}$

		This gives us HPS2 $\aevalfat{\codeenv, \procid}{\st{\sigma_2.\pc,\hat{L_2}, \hat{S_2}}}{\st{\sigma''.\pc,\hat{L''}, \hat{S''}}}$
	\end{itemize}

	By well-typedness of code and \Cref{tr:enc-mod} we get HPC: $\com{\sigma \in \concfun{ \st{\hat{S}, \hat{L}} }}$.

	By \Thmref{thm:gen-enc} with HPw, HPRF, HPEN, HPEX, HPS, HPS2, HPC, HP3 we have
	\begin{itemize}
		\item HPP $\weakinvatk{\codeenv'}{\sigma''}{\inv}$.
		\item HPC $\com{\sigma'' \in \concfun{ \st{\hat{S''}, \hat{L''}} }}$
	\end{itemize}

	By \Thmref{thm:bang-act-pres} with HPP, HPRA, HPS2, HPC we have what is needed.

\end{proof}
\BREAK

\begin{lemma}[Generalised Encapsulation]\label{thm:gen-enc}
	\begin{align*}
		\text{ if }
			&\
			\weakinvatk{\codeenv'}{\sigma}{\inv}
		\\
		\text{ and }
			&\
			\com{\codeenv' \triangleright \codeenv,\procid} \vdash \com{\sigma \Xto{\nil} \sigma'}
		\\
		\text{ and }
			&\
			\codeenv,\procid \nvdash \com{\sigma} : \com{entry}
		\\
		\text{ and }
			&\
			\codeenv,\procid \nvdash \com{\sigma'} : \com{exit}
		\\
		\text{ and }
			&\
			\com{\sigma \in \concfun{ \st{\hat{S}, \hat{L}} }}
		\\
		\text{ and }
			&\
			\aevalfat{\codeenv, \procid}{\st{\sigma.\pc,\hat{L}, \hat{S}}}{\st{\sigma''.\pc,\hat{L''}, \hat{S''}}}
		\\
		\text{ and }
			&\
			\aevalfat{\codeenv, \procid}{\st{\sigma'''.\pc,\hat{L'''}, \hat{S'''}}}{\st{\sigma'.\pc,\hat{L'}, \hat{S'}}}
		\\
  	\text{ and }
    	&
    	\com{\ea}(\restr{\codeenv'}{\inv})
  	\\
		\text{ then }
			&\
			\weakinvatk{\codeenv'}{\sigma'}{\inv}
		\\
		\text{ and }
			&\
			\com{\sigma' \in \concfun{ \st{\hat{S'}, \hat{L'}} }}
	\end{align*}
\end{lemma}
\begin{proof}
	Take the reductions in HP2: they are all for functions in \com{\codeenv.}
	Also, they can be split into sequences of reductions from an entry point to the next exit point for each function.
	(intuitively, from the start of the reduction, code can call other functions in \com{\codeenv} until the attacker code is called or returned to).

	By \Thmref{def:fat-red-split-entryexit-red} with HP2 we have HPRN:
	\[
	\com{\codeenv' \triangleright \codeenv,\procid} \vdash \com{\sigma \Xto{\nil} \sigma_1 \xto{\nil} \sigma_1' \Xto{\nil} \cdots \xto{\nil} \sigma_n' \Xto{\nil} \sigma'}
	\]
	and $\forall i \in 0 .. n/2 \ldotp$, HPRI:
	\[
		\vdash \com{\codeenv' \triangleright \codeenv,\procid} \vdash \com{\sigma_{2i} \Xto{\nil} \sigma_{2i+1} } : \com{entry/exit}
	\]
	and HPRE:
	\[
	\vdash \com{\codeenv' \triangleright \codeenv,\procid} \vdash \com{\sigma_{2i+1} \xto{\nil} \sigma_{2i+1}'' } : \com{exit/entry}
	\]

	We proceed by induction over the amount of entry-to-entry reductions within the \com{\codeenv'} execution:
    \begin{description}
    	\item[Base, $i=0$:]

    	We have:
    	\begin{itemize}
    		\item $\com{\codeenv' \triangleright \codeenv,\procid, i} \vdash \com{\sigma \Xto{\nil} \sigma'}$
    		\item HPR $\vdash \com{\codeenv' \triangleright \codeenv,\procid} \vdash \com{\sigma \Xto{\nil} \sigma' } : \com{entry/exit}$ (by definition)
    		\item HPE $\aevalfat{\codeenv, \procid}{\st{\sigma.\pc,\hat{L}, \hat{S}}}{\st{\sigma'.\pc,\hat{L'}, \hat{S'}}}$

    		Since the reduction is entry/exit, \com{\sigma'}=\com{\sigma''} and \com{\sigma'''}=\com{\sigma}
    	\end{itemize}
    	We need to prove:
    	\begin{itemize}
    		\item $\weakinvatk{\codeenv'}{\sigma'}{\inv}$.
    		\item $\com{\sigma' \in \concfun{ \st{\hat{S'}, \hat{L'}} }}$
    	\end{itemize}

    	By \Thmref{thm:entry-exit-pres-enc} with HP1, HPR, HPE, HP5 we have what is needed.

    	\item[Inductive, $i=i+1$:]

    	In this case we have $i$ steps:
    	\[
    		\com{\codeenv' \triangleright \codeenv,\procid} \vdash \com{\sigma \Xto{\nil} \sigma_1 \xto{\nil} \sigma_1' \Xto{\nil} \cdots \Xto{\nil} \sigma_{i} \xto{\nil} \sigma_i'}
    	\]
    	which proceed as HPRD
    	\[
    		\com{\codeenv' \triangleright \codeenv,\procid} \vdash \com{\sigma_i' \Xto{\nil} \sigma_{i+1} \xto{\nil} \sigma_{i+1}' \Xto{\nil} \sigma' }
    	\]

    	By IH we have
    	\begin{itemize}
    		\item IHSi $\weakinvatk{\codeenv'}{\sigma_i'}{\inv}$
    		\item IHCi $\com{\sigma_i' \in \concfun{ \st{\hat{S_i}, \hat{L_i}} }}$
    	\end{itemize}
    	and we need to prove:
    	\begin{itemize}
    		\item $\weakinvatk{\codeenv'}{\sigma'}{\inv}$
    		\item $\com{\sigma' \in \concfun{ \st{\hat{S'}, \hat{L'}} }}$
    	\end{itemize}

    	We also have the related encapsulation from HP8, since HP8 implies that all the code from entry to exit points are encapsulated:
    	\begin{itemize}
    	 	\item HPEI $\aevalfat{\codeenv, \procid}{\st{\sigma_i'.\pc,\hat{L_i'}, \hat{S_i'}}}{\st{\sigma_{i_1}.\pc,\hat{L_{i+1}}, \hat{S_{i+1}}}} $
    	 	\item HPE+1 $\aevalfat{\codeenv, \procid}{\st{\sigma_{i+1}'.\pc,\hat{L_{i+1}'}, \hat{S_{i+1}'}}}{\st{\sigma'.\pc,\hat{L'}, \hat{S'}}} $
    	\end{itemize}

    	By \Thmref{thm:entry-entry-pres-enc} with IHSi, HPRD (for the next 3 hps), HPEI, HPE+1, HPCi, HP8, we have:
    	\begin{itemize}
    		\item HPS+1 $\weakinvatk{\codeenv'}{\sigma_{i+1}'}{\inv}$
    		\item HPC+1 $\com{\sigma_{i+1}' \in \concfun{ \st{\hat{S_{i+1}'}, \hat{L_{i+1}'}} }}$
    	\end{itemize}

    	By \Thmref{thm:entry-exit-pres-enc} with HPS+1, HPRD (last fat red), HPE+1, HPC+1 we have what is needed.

    \end{description}
\end{proof}
\BREAK

\begin{lemma}[!Actions from Concretized States Preserve the Weak Invariant]\label{thm:bang-act-pres}
	\begin{align*}
		\text{ if }
			&\
			\weakinvatk{\codeenv}{\sigma}{\inv}
		\\
		\text{ and }
			&\
			\com{\codeenv' \triangleright \codeenv,\procid} \vdash \com{\sigma \xto{\alpha!} \sigma' }
		\\
		\text{ and }
			&\
			\aevalfat{\codeenv, \procid}{\st{\sigma.\pc,\hat{L}, \hat{S}}}{\st{\sigma''.\pc,\hat{L''}, \hat{S''}}}
		\\
		\text{ and }
    		&
            \com{\sigma \in \concfun{\st{\hat{S}, \hat{L}}}}
        \\
        \text{ then }
    		&
        	\weakinvatk{\codeenv}{\sigma}{\inv}
	\end{align*}
\end{lemma}
\begin{proof}
	\begin{description}
			\item[\Cref{tr:act-call}:] where the function being called is not defined in the \tc \com{\codeenv'}

      This is a contradiction with \Cref{prop:atk-no-call}.

			\item[\Cref{tr:act-retback}:] to a function defined outside \com{\codeenv'}

      By \Cref{tr:encapsulator-return} we know that none of the parameters are \inref.

      From the definition of concretization for \okref and \noref, we derive that there are no globals associated with the returned values.

		\end{description}
\end{proof}
\BREAK

\begin{lemma}[Entry to Enty Reductions Preserve Encapsulation and Weak Invariant]\label{thm:entry-entry-pres-enc}
	\begin{align*}
		\text{ if }
			&\
			\weakinvatk{\codeenv'}{\sigma}{\inv}
		\\
		\text{ and }
			&\
			\vdash \com{\codeenv' \triangleright \codeenv,\procid} \vdash \com{\sigma \Xto{\nil} \sigma' } : \com{entry/exit}
		\\
		\text{ and }
			&\
			\vdash \com{\codeenv' \triangleright \codeenv,\procid} \vdash \com{\sigma'' \Xto{\nil} \sigma''' } : \com{entry/exit}
		\\
		\text{ and }
			&\
			\vdash \com{\codeenv' \triangleright \codeenv,\procid} \vdash \com{\sigma' \xto{\nil} \sigma'' } : \com{exit/entry}
		\\
		\text{ and }
			&\
			\vdash \aevalfat{\codeenv, \procid}{\st{\sigma.\pc,\hat{L}, \hat{S}}}{\st{\sigma'.\pc,\hat{L'}, \hat{S'}}} : \com{entry/exit}
		\\
		\text{ and }
			&\
			\vdash \aevalfat{\codeenv, \procid}{\st{\sigma''.\pc,\hat{L''}, \hat{S''}}}{\st{\sigma'''.\pc,\hat{L'''}, \hat{S'''}}} : \com{entry/exit}
		\\
		\text{ and }
		&\
        \com{\sigma \in \concfun{\st{\hat{S}, \hat{L}}}}
    \\
    \text{ and }
      &
      \com{\ea}(\restr{\codeenv'}{\inv})
    \\
    \text{ then }
    	&\
    	\weakinvatk{\codeenv'}{\sigma''}{\inv}
    \\
    \text{ and }
    	&\
    	\com{\sigma'' \in \concfun{\st{\hat{S''}, \hat{L''}}}}
	\end{align*}
\end{lemma}
\begin{proof}
	By \Thmref{thm:entry-exit-pres-enc} with HP1, HP2, HP5, HP7 we get
	\begin{itemize}
		\item HPS $\com{\sigma' \in \concfun{\st{\hat{S'}, \hat{L'}}}}$
		\item HPW $\weakinvatk{\codeenv'}{\sigma'}{\inv}$
	\end{itemize}

	By \Cref{tr:entry} and \Cref{tr:exit} there are only these combinations here:
	\begin{enumerate}
		\item \com{\sigma'} performs a \ao{\callCmd}{p} to \com{\sigma'' }, which is the starting address of a function.

    From \Cref{tr:enc-mod} we know that \com{\sigma''} is encapsulated with abstract types that match its signature (\totypes{\cdot}).

    The well-typedness of the code ensures that the parameters being passed match the signature.

    The definition of \totypes{\cdot} and of \concfun{\cdot} ensure TH2.

    Since \totypes{\cdot} never returns an \inref, from the definition of concretization for \okref and \noref, we derive that there are no globals associated with the returned values.

		\item \com{\sigma'} performs a \returnCmd to \com{\sigma'' }, which is the address after a call done in the past (this also covers the case where \com{\sigma'} is the last instruction of a function, which is a return)

    By HP8 we know that \com{\sigma''} is the state after a call, so we know there exists \com{\sigma?} which is the state where the call is done such that
    \[
      \aeval{\codeenv, \procid, \ao{\callCmd}{\procid_0}}{\st{\sigma?.\pc,\hat{L?}, \hat{S?}}}{\st{\sigma''.\pc,\hat{L''}, \hat{S''}}}
    \]

    We now case-analyse the returned values in \com{\sigma''} to see if their abstract value matches their concretisation.
    As before, we rely on the well-typedness to know that the returned values match their signature.

    As before, the definition of \totypes{\cdot} and of \concfun{\cdot} ensure TH2.

    For TH1 we have 2 cases:
    \begin{enumerate}
      \item If the returned value is not an \inref, so from the definition of concretization for \okref and \noref, we derive that there are no globals associated with the returned values.

      \item If the returned value is \inref, notice that the called function \com{\procid_0} is defined in \com{\codeenv}, so nothing is passed to the attacker and TH1 holds.
    \end{enumerate}

	\end{enumerate}
\end{proof}
\BREAK

\begin{lemma}[Entry/Exit Reductions Preserve Encapsulation and Weak Invariant]\label{thm:entry-exit-pres-enc}
	\begin{align*}
		\text{ if }
			&\
			\weakinvatk{\codeenv'}{\sigma}{\inv}
		\\
		\text{ and }
			&\
			\vdash \com{\codeenv' \triangleright \codeenv,\procid} \vdash \com{\sigma \Xto{\nil} \sigma' } : \com{entry/exit}
		\\
		\text{ and }
			&\
			\vdash \aevalfat{\codeenv, \procid}{\st{\sigma.\pc,\hat{L}, \hat{S}}}{\st{\sigma'.\pc,\hat{L'}, \hat{S'}}} : \com{entry/exit}
		\\
		\text{ and }
    		&
            \com{\sigma \in \concfun{\st{\hat{S}, \hat{L}}}}
        \\
        \text{ then }
        	&
        	\weakinvatk{\codeenv'}{\sigma'}{\inv}
        \\
        \text{ and }
    		&
        	\com{\sigma' \in \concfun{\st{\hat{S'}, \hat{L'}}}}
	\end{align*}
\end{lemma}
\begin{proof}
	This proof proceeds by induction over \com{\Xto{}} with \Thmref{thm:single-red-pres-enc}.
\end{proof}
\BREAK

\begin{lemma}[Single Reductions Preserve Encapsulation and Weak Invariant]\label{thm:single-red-pres-enc}
	\begin{align*}
        \text{ if }
        	&
			\weakinvatk{\codeenv'}{\sigma}{\inv}
      	  	\\
        \text{ and }
        	&
        	\com{\codeenv' \triangleright \codeenv,\procid, i} \vdash \com{\sigma \xto{\nil} \sigma'}
        	\\
        \text{ and }
    		&
            \com{\aeval{\codeenv', \procid, i}{\st{\hat{L}, \hat{S}}}{\st{\hat{L'}, \hat{S'}}}}
        	\\
        \text{ and }
    		&
            \com{\sigma \in \concfun{\st{\hat{S}, \hat{L}}}}
            \\
        \text{ and }
        	&\
        	\codeenv,\procid \nvdash \com{\sigma} : \com{exit}
        	\\
        \text{ and }
        	&\
        	\codeenv,\procid \nvdash \com{\sigma'} : \com{entry}
        	\\
        \text{ then }
        	&
			\weakinvatk{\codeenv'}{\sigma'}{\inv}
            \\
        \text{ and }
        	&
            \com{\sigma' \in \concfun{\st{\hat{S'}, \hat{L'}}}}
	\end{align*}
\end{lemma}
\begin{proof}

  This proof proceeds by case analysis on $\xto{\nil}$.
  For simplicity, we treat all cases together, though some would be the local or global sub-cases.
  \begin{description}
    \item[\Cref{tr:sem-loc-moveloc}, $\com{i}=\ao{\movelocCmd}{x}$]

    Here  the value is not a reference, so its abstract type is an \noref.

    The encapsulator reduction is \Cref{tr:encapsulator-moveloc}

    Since this reduction does not change \com{G} nor \com{M}, TH1 follows from \Thmref{thm:g-and-m-invar-pres-unreach}.

    TH2 follows from the definition of \concfun{\cdot}.

    \item[\Cref{tr:sem-loc-movelocref}, $\com{i}=\ao{\movelocCmd}{x}$]

    As above, the reduction does not touch \com{G} nor \com{M} and the abstract type concretizes to the expected value.

    \item[\Cref{tr:sem-loc-copyloc}, $\com{i}=\ao{\copylocCmd}{\loc}$]

   	As above, the reduction does not touch \com{G} nor \com{M} and the abstract type concretizes to the expected value.

    \item[\Cref{tr:sem-loc-copylocref}, $\com{i}=\ao{\copylocCmd}{\loc}$]

    As above, the reduction does not touch \com{G} nor \com{M} and the abstract type concretizes to the expected value.

    \item[\Cref{tr:sem-loc-storeloc}, $\com{i}=\ao{\storelocCmd}{x}$]

    As above, the reduction does not touch \com{G} nor \com{M} and the abstract type concretizes to the expected value.

    \item[\Cref{tr:sem-loc-storelocref}, $\com{i}=\ao{\storelocCmd}{x}$]

    As above, the reduction does not touch \com{G} nor \com{M} and the abstract type concretizes to the expected value.

    \item[\Cref{tr:sem-loc-borrowloc}, $\com{i}=\ao{\borrowlocCmd}{x}$]

    We have:
    \begin{center}
		\AxiomC{
			$
			\lread{L}{x} = \loc
			$
		}
		\UnaryInfC{
			$
			\eval{
		      \st{M,L,S}
		    }{
		      \ao{\borrowlocCmd}{x}
		    }{
		      \st{M,L,\stackht{\rft{\loc}{[]}}{S}}
		    }
		    $
		}
		\UnaryInfC{
			$
			\peval{
			    \st{\stackht{\tup{\procid,\pc,L}}{C}, M,G,S}
			}{
			    \st{\stackht{\tup{\procid,\pc+1,L}}{C}, M,G,\stackht{\rft{\loc}{[]}}{S}}
			}
			$
		}
		\DisplayProof
	\end{center}
    and the encapsulator reduction is \Cref{tr:encapsulator-borrowloc}
    \begin{center}
		\AxiomC{
			$
			\lread{L}{x} = \loc
			$
		}
		\UnaryInfC{
			$
			\aeval{\codeenv', \procid, i}{
				\st{ \hat{L} , \hat{S}}
			}{
				\st{ \hat{L}, \stackht{\okref}{\hat{S}}}
			}
			$
		}
		\DisplayProof
	\end{center}

    Since this reduction does not change \com{G} nor \com{M}, TH1 follows from \Thmref{thm:g-and-m-invar-pres-unreach}.

    TH2 is $\st{\stackht{\tup{\procid,\pc+1,L}}{C}, M,G,\stackht{\rft{\loc}{[]}}{S}} \in \concfun{ \st{ \hat{L}, \stackht{\okref}{\hat{S}}} }$.

    The \com{L} is unchanged, so that part follows from HP4.

    The \com{S} part is extended with a \okref, so apart from HP4, we need to show that $\com{L(x)}\in\concfun{\okref}$, which holds from the semantics assumption that \com{L(x) = \loc} since \com{\loc} is a valid location in \com{M} are required by \concfun{\okref}.

    \item[\Cref{tr:sem-loc-borrowfld}, $\com{i}=\ao{\borrowfieldCmd}{f}$]

    We have:
    \begin{center}
		\AxiomC{
			$
			\val = \rft{\loc}{p}\sep
		    \loc\in\dom{M}\sep
		    \mread{M}{\loc}[p]=\tv{\record{(f,\val_{f}),\cdots}}{t}
			$
		}
		\UnaryInfC{
			$
			\eval{
		      \st{M,L,\stackht{\val}{S}}
		    }{
		      \ao{\borrowfieldCmd}{f}
		    }{
		      \st{M,L,\stackht{\rft{\loc}{\stackht{p}{f}}}{S}}
		    }
		    $
		}
		\UnaryInfC{
			$
			\peval{
			    \st{\stackht{\tup{\procid,\pc,L}}{C}, M,G,\stackht{\val}{S}}
			}{
			    \st{\stackht{\tup{\procid,\pc+1,L}}{C}, M,G,\stackht{\rft{\loc}{\stackht{p}{f}}}{S}}
			}
			$
		}
		\DisplayProof
	\end{center}
    and we have two cases for the encapsulator reductions.

    Since this reduction does not change \com{G} nor \com{M}, TH1 follows from \Thmref{thm:g-and-m-invar-pres-unreach}, so we prove TH2 in both cases below.
    \begin{enumerate}
    	\item \Cref{tr:encapsulator-borrowfld-attacker}
		\begin{center}
			\AxiomC{
        $f\in\inv$
			}
			\UnaryInfC{
				$
				\aeval{
          \inv,
					\ao{\borrowfieldCmd}{f}
				}{
					\st{\hat{L},\stackht{v}{\hat{S}}}
				}{
					\st{\hat{L}, \stackht{\inref}{\hat{S}}}
				}
				$
			}
			\DisplayProof
		\end{center}

		For TH2, \com{L} is unchanged and \com{S} is extended with \com{\rft{\loc}{\stackht{p}{f}}}, which we need to prove $\in \concfun{\inref}$.

		The additional premise is that \com{\val}, so \com{\st{c,p}}, comes from a field with an invariant.

		Since the new field may point to a value that is a global with an invariant on, the abstract value being \inref covers this case too.

		\item  \Cref{tr:encapsulator-borrowfld}
		\begin{center}
			\AxiomC{
        $f\notin\inv$
			}
			\UnaryInfC{
				$
				 \aeval{
          \inv,
				  \ao{\borrowfieldCmd}{f}
				}{
				  \st{\hat{L},\stackht{\hat{v}}{\hat{S}}}
				}{
				  \st{\hat{L}, \stackht{\hat{v}}{\hat{S}}}
				}
				$
			}
			\DisplayProof
		\end{center}

		For TH2, \com{L} is unchanged and \com{S} is extended with \com{\rft{\loc}{\stackht{p}{f}}}, which we need to prove $\in \concfun{\hat{v}}$.

    We know that the abstract value of \com{\val} is \com{\hat{\val}}, and that the field being read is $\notin\inv$, so the read value is not a reference pointing to invariants.

    So the reference will be in \concfun{\okref}, and we can conclude that the read value is in \concfun{\hat{v}} since $\okref \sqsubseteq \inref$.

    \end{enumerate}

    \item[\Cref{tr:sem-loc-readref}, $\com{i}=\readrefCmd$]

    As above, the reduction does not touch \com{G} nor \com{M} and the abstract type concretizes to the expected value.

    \item[\Cref{tr:sem-loc-writeref}, $\com{i}=\writerefCmd$]

    The encapsulator reduction is \Cref{tr:encapsulator-writeref}.

    As above, the reduction does not touch \com{G} nor \com{M} (since the type restriction prevents \Cref{tr:atk-part-state} from changing) and the abstract type concretizes to the expected value.

    \item[\Cref{tr:sem-loc-pop}, $\com{i}=\popCmd$]

    As above, the reduction does not touch \com{G} nor \com{M} and the abstract type concretizes to the expected value.

    \item[\Cref{tr:sem-loc-loadconst}, $\com{i}=\ao{\loadconstCmd}{v}$]

    As above, the reduction does not touch \com{G} nor \com{M} and the abstract type concretizes to the expected value.

    \item[\Cref{tr:sem-loc-op}, $\com{i}=\stackopCmd$]

    As above, the reduction does not touch \com{G} nor \com{M} and the abstract type concretizes to the expected value.

    \item[\Cref{tr:sem-glob-pack}, $\com{i}=\ao{\packCmd}{s}$]

    As above, the reduction does not touch \com{G} nor \com{M} and the abstract type concretizes to the expected value.

    \item[\Cref{tr:sem-glob-unpack}, $\com{i}=\ao{\unpackCmd}{s}$]

    As above, the reduction does not touch \com{G} nor \com{M} and the abstract type concretizes to the expected value.

    \item[\Cref{tr:sem-glob-movefrom}, $\com{i}=\ao{\movefromCmd}{s}$]

    As above, the reduction does not touch \com{G} nor \com{M} and the abstract type concretizes to the expected value.

    \item[\Cref{tr:sem-glob-moveto}, $\com{i}=\ao{\movetoCmd}{s}$]

    The encapsulator reduction is \Cref{tr:encapsulator-moveto}.

    The semantics rule ensures we are creating a new global and a new memory location, which are not under attacker influence, so TH1 follows from \Thmref{thm:g-and-m-invar-pres-unreach}.

    TH2 is a trivial consequence since we only eliminate bindings.

    \item[\Cref{tr:sem-glob-borrowglobal}, $\com{i}=\ao{\borrowglobalCmd}{s}$]

    We have:
    \begin{center}
		\AxiomC{
			$
			\stype = \tup{\procid.mid, s}\sep
	      	\gread{G}{\tup{\addr,\stype}} = \loc
	      	$
		}
		\UnaryInfC{
			$
			\geval{
				\ao{\borrowglobalCmd}{s}
			}{
				\st{M,G,\stackht{\addr}{S}}
			}{
				\st{M,G,\stackht{\rft{\loc}{[]}}{S}}
			}
		    $
		}
		\UnaryInfC{
			$
			\peval{
			    \st{\stackht{\tup{\procid,\pc,L}}{C}, M,G,\stackht{\addr}{S}}
			}{
			    \st{\stackht{\tup{\procid,\pc+1,L}}{C}, M,G,\stackht{\rft{\loc}{[]}}{S}}
			}
			$
		}
		\DisplayProof
	\end{center}
    and the encapsulator reduction is \Cref{tr:encapsulator-borrowglobal}
    \begin{center}
		\AxiomC{
		}
		\UnaryInfC{
			$
			\aeval{
		      \ao{\borrowglobalCmd}{s}
		    }{
		      \st{\hat{L}, \stackht{\hat{v}}{\hat{S}}}
		    }{
		      \st{\hat{L}, \stackht{\inref}{\hat{S}}}
		    }
			$
		}
		\DisplayProof
	\end{center}

	Since this reduction does not change \com{G} nor \com{M}, TH1 follows from \Thmref{thm:g-and-m-invar-pres-unreach}.

    TH2 is $\st{\stackht{\tup{\procid,\pc+1,L}}{C}, M,G,\stackht{\rft{\loc}{[]}}{S}} \in \concfun{ \st{\hat{L}, \stackht{\inref}{\hat{S}}} }$.

    The \com{L} is unchanged, so that part follows from HP4.

    The \com{S} part is extended with a \inref, so apart from HP4, we need to show that $\com{ \rft{\loc}{[]} } \in \concfun{\inref}$ .

    By definition of \concfun{\inref} this is true if \com{\loc} is a valid location in \com{M} that can point to a global, which holds.

    \item[\Cref{tr:sem-call}, $\com{i}=\ao{\callCmd}{p}$]

    This is not possible since \com{\sigma} is not an exit state.

    \item[\Cref{tr:sem-ret}, $\com{i}=\returnCmd$]

    This is not possible since \com{\sigma'} is not an exit state.

    \item[\Cref{tr:sem-branch}, $\com{i}=\ao{\branchCmd}{\pc}$]

    As above, the reduction is analogous to the call case and it does not touch \com{G} nor \com{M} and the abstract type concretizes to the expected value.

    \item[\Cref{tr:sem-ift}, $\com{i}=\ao{\branchcondCmd}{\pc}$]

    As above, the reduction does not touch \com{G} nor \com{M} and the abstract type concretizes to the expected value.

    \item[\Cref{tr:sem-iff}, $\com{i}=\ao{\branchcondCmd}{\pc}$]

    As above, the reduction does not touch \com{G} nor \com{M} and the abstract type concretizes to the expected value.
  \end{description}
\end{proof}
\BREAK

\begin{lemma}[Globals and Memory Invariance Preserve Unreachability]\label{thm:g-and-m-invar-pres-unreach}
	\begin{align*}
		\text{ if }
			&
			\weakinvatk{\codeenv'}{\sigma}{\inv}
		\\
		\text{ and }
			&
			\com{\codeenv,\sigma} \vdash \com{M_a}, \com{G_a} : \com{attackerpart}
		\\
		\text{ and }
			&
			\com{\codeenv,\sigma'} \vdash \com{M_a}, \com{G_a} : \com{attackerpart}
		\\
		\text{ then }
			&
			\weakinvatk{\codeenv'}{\sigma'}{\inv}
	\end{align*}
\end{lemma}
\begin{proof}
	Trivial unfolding of \Cref{tr:weak-inv-unchange}.
\end{proof}
\BREAK

\subsection{ Encapsulator \com{\enb}: Escape Analysis for Mutable Attackers}

\com{\enb} is defined as \com{\ea} from \Cref{sec:ea-rules} except that \Cref{tr:encapsulator-call} is replaced with \Cref{tr:encapsulator-nb-call} below:
\begin{center}
  \typerule{\com{\enb}-Call}{
    \card{\codeenv(\procid_0).type} = n
    &
    \forall i \in 1..n\ldotp
    \hat{v_i^a} \neq \inref
    \\
    \hat{v_1^r}\cdots\hat{v_j^r} = \totypes{\codeenv(P)\dotretty}
    &
    \procid_0\notin\tcenv
  }{
    \aevalmut{
      \codeenv, \procid, \ao{\callCmd}{\procid_0}
    }{
       \st{\hat{L},\stackht{\hat{v_1^a}\cdots\hat{v_n^a}}{\hat{S}}}
    }{
      \st{\hat{L},\stackht{\hat{v_1^r}\cdots\hat{v_j^r}}{\hat{S}}}
    }
  }{encapsulator-nb-call}
\end{center}

\begin{theorem}[The \com{\enb} Escape Analysis is a Sound Encapsulator (\encapsulated{\com{\codeenv'}}{\enb}{\inv})]\label{thm:escape-ok-nb}
  \begin{align*}
    \text{ if }
      &
      \stronginv{\codeenv'}{\sigma}{\inv}
        \\
        \text{ and }
          &
          \com{\codeenv' \triangleright \codeenv,\procid} \vdash \com{\sigma \Xto{\alpha!} \sigma'}
        \\
        \text{ and }
          &
          \com{\enb}(\restr{\codeenv'}{\inv})
        \\
        \text{ then }
          &
          \weakinvatk{\codeenv'}{\sigma'}{\inv}
  \end{align*}
\end{theorem}
\begin{proof}
  The proof is analogous to that of \Thmref{thm:escape-ok}, except that the usage of \Thmref{thm:bang-act-pres} gets replaced with \Thmref{thm:bang-act-pres-nb}.
\end{proof}

\begin{lemma}[!Actions from Concretized \com{\enb} States Preserve the Weak Invariant]\label{thm:bang-act-pres-nb}
  \begin{align*}
    \text{ if }
      &\
      \weakinvatk{\codeenv}{\sigma}{\inv}
    \\
    \text{ and }
      &\
      \com{\codeenv' \triangleright \codeenv,\procid} \vdash \com{\sigma \xto{\alpha!} \sigma' }
    \\
    \text{ and }
      &\
      \aevalfatmut{\codeenv, \procid}{\st{\sigma.\pc,\hat{L}, \hat{S}}}{\st{\sigma''.\pc,\hat{L''}, \hat{S''}}}
    \\
    \text{ and }
        &
            \com{\sigma \in \concfun{\st{\hat{S}, \hat{L}}}}
        \\
        \text{ then }
        &
          \weakinvatk{\codeenv}{\sigma}{\inv}
  \end{align*}
\end{lemma}
\begin{proof}
  \begin{description}
      \item[\Cref{tr:act-call}:] where the function being called is not defined in the \tc \com{\codeenv'}.

      By \Cref{tr:encapsulator-nb-call} we know that none of the parameters are \inref.

      From the definition of concretization for \okref and \noref, we derive that there are no globals associated with the returned values.

      \item[\Cref{tr:act-retback}:] to a function defined outside \com{\codeenv'}

      By \Cref{tr:encapsulator-return} we know that none of the parameters are \inref.

      From the definition of concretization for \okref and \noref, we derive that there are no globals associated with the returned values.

    \end{description}
\end{proof}
\BREAK

\newpage
\twocolumn

\bibliographystyle{plainnat}
\bibliography{./../biblio.bib}

\begin{thebibliography}{46}
\providecommand{\natexlab}[1]{#1}
\providecommand{\url}[1]{\texttt{#1}}
\expandafter\ifx\csname urlstyle\endcsname\relax
  \providecommand{\doi}[1]{doi: #1}\else
  \providecommand{\doi}{doi: \begingroup \urlstyle{rm}\Url}\fi

\bibitem[Abadi(1999)]{sec-typ-prot}
Mart\'{\i}n Abadi.
\newblock Secrecy by typing in security protocols.
\newblock \emph{J. ACM}, 46\penalty0 (5):\penalty0 749--786, September 1999.
\newblock ISSN 0004-5411.
\newblock \doi{10.1145/324133.324266}.
\newblock URL \url{https://doi.org/10.1145/324133.324266}.

\bibitem[Albert et~al.(2020)Albert, Grossman, Rinetzky,
  Rodr{\'{\i}}guez{-}N{\'{u}}{\~{n}}ez, Rubio, and
  Sagiv]{DBLP:journals/pacmpl/AlbertGRRRS20}
Elvira Albert, Shelly Grossman, Noam Rinetzky, Clara
  Rodr{\'{\i}}guez{-}N{\'{u}}{\~{n}}ez, Albert Rubio, and Mooly Sagiv.
\newblock Taming callbacks for smart contract modularity.
\newblock \emph{Proc. {ACM} Program. Lang.}, 4\penalty0 ({OOPSLA}):\penalty0
  209:1--209:30, 2020.
\newblock \doi{10.1145/3428277}.
\newblock URL \url{https://doi.org/10.1145/3428277}.

\bibitem[Amsden et~al.(2019)Amsden, Arora, Bano, Baudet, Blackshear, Bothra,
  Cabrera, Catalini, Chalkias, Cheng, Ching, Chursin, Danezis, Giacomo, Dill,
  Ding, Doudchenko, Gao, Gao, Garillot, Gorven, Hayes, Hou, Hu, Hurley, Lewi,
  Li, Li, Malkhi, Margulis, Maurer, Mohassel, de~Naurois, Nikolaenko, Nowacki,
  Orlov, Perelman, Pott, Proctor, Qadeer, Rain, Russi, Schwab, Sezer, Sonnino,
  Venter, Wei, Wernerfelt, Williams, Wu, Yan, Zakian, and
  Zhou]{libra_blockchain_white}
Zachary Amsden, Ramnik Arora, Shehar Bano, Mathieu Baudet, Sam Blackshear,
  Abhay Bothra, George Cabrera, Christian Catalini, Konstantinos Chalkias, Evan
  Cheng, Avery Ching, Andrey Chursin, George Danezis, Gerardo~Di Giacomo,
  David~L. Dill, Hui Ding, Nick Doudchenko, Victor Gao, Zhenhuan Gao, François
  Garillot, Michael Gorven, Philip Hayes, J.~Mark Hou, Yuxuan Hu, Kevin Hurley,
  Kevin Lewi, Chunqi Li, Zekun Li, Dahlia Malkhi, Sonia Margulis, Ben Maurer,
  Payman Mohassel, Ladi de~Naurois, Valeria Nikolaenko, Todd Nowacki, Oleksandr
  Orlov, Dmitri Perelman, Alistair Pott, Brett Proctor, Shaz Qadeer, Rain,
  Dario Russi, Bryan Schwab, Stephane Sezer, Alberto Sonnino, Herman Venter,
  Lei Wei, Nils Wernerfelt, Brandon Williams, Qinfan Wu, Xifan Yan, Tim Zakian,
  and Runtian Zhou.
\newblock The {L}ibra {B}lockchain.
\newblock \url{https://developers.libra.org/docs/the-libra-blockchain-paper},
  2019.

\bibitem[Backes et~al.(2014)Backes, Hritcu, and Maffei]{catalin-rs}
Michael Backes, Catalin Hritcu, and Matteo Maffei.
\newblock Union, intersection and refinement types and reasoning about type
  disjointness for secure protocol implementations.
\newblock \emph{Journal of Computer Security}, 22\penalty0 (2):\penalty0
  301--353, 2014.
\newblock \doi{10.3233/JCS-130493}.
\newblock URL \url{http://dx.doi.org/10.3233/JCS-130493}.

\bibitem[Barnett et~al.(2005)Barnett, Chang, DeLine, Jacobs, and
  Leino]{DBLP:conf/fmco/BarnettCDJL05}
Michael Barnett, Bor{-}Yuh~Evan Chang, Robert DeLine, Bart Jacobs, and
  K.~Rustan~M. Leino.
\newblock Boogie: {A} modular reusable verifier for object-oriented programs.
\newblock In Frank~S. de~Boer, Marcello~M. Bonsangue, Susanne Graf, and
  Willem~P. de~Roever, editors, \emph{Formal Methods for Components and
  Objects, 4th International Symposium, {FMCO} 2005, Amsterdam, The
  Netherlands, November 1-4, 2005, Revised Lectures}, volume 4111 of
  \emph{Lecture Notes in Computer Science}, pages 364--387. Springer, 2005.
\newblock \doi{10.1007/11804192\_17}.
\newblock URL \url{https://doi.org/10.1007/11804192\_17}.

\bibitem[Bengtson et~al.(2011)Bengtson, Bhargavan, Fournet, Gordon, and
  Maffeis]{refty-sec-impl}
Jesper Bengtson, Karthikeyan Bhargavan, C{\'e}dric Fournet, Andrew~D. Gordon,
  and Sergio Maffeis.
\newblock Refinement types for secure implementations.
\newblock \emph{ACM Trans. Program. Lang. Syst.}, 33\penalty0 (2):\penalty0
  8:1--8:45, February 2011.
\newblock ISSN 0164-0925.
\newblock \doi{10.1145/1890028.1890031}.
\newblock URL \url{http://doi.acm.org/10.1145/1890028.1890031}.

\bibitem[Bernardy et~al.(2017)Bernardy, Boespflug, Newton, Peyton~Jones, and
  Spiwack]{linearhs}
Jean-Philippe Bernardy, Mathieu Boespflug, Ryan~R. Newton, Simon Peyton~Jones,
  and Arnaud Spiwack.
\newblock Linear haskell: Practical linearity in a higher-order polymorphic
  language.
\newblock \emph{Proc. ACM Program. Lang.}, 2\penalty0 (POPL), December 2017.
\newblock \doi{10.1145/3158093}.
\newblock URL \url{https://doi.org/10.1145/3158093}.

\bibitem[Blackshear et~al.(2019)Blackshear, Cheng, Dill, Gao, Maurer, Nowacki,
  Pott, Qadeer, Rain, Russi, Sezer, Zakian, and Zhou]{move_white}
Sam Blackshear, Evan Cheng, David~L. Dill, Victor Gao, Ben Maurer, Todd
  Nowacki, Alistair Pott, Shaz Qadeer, Rain, Dario Russi, Stephane Sezer, Tim
  Zakian, and Runtian Zhou.
\newblock Move: A language with programmable resources.
\newblock \url{https://developers.libra.org/docs/move-paper}, 2019.

\bibitem[Blackshear et~al.(2020)Blackshear, Dill, Qadeer, Barrett, Mitchell,
  Padon, and Zohar]{blackshear2020resources}
Sam Blackshear, David~L. Dill, Shaz Qadeer, Clark~W. Barrett, John~C. Mitchell,
  Oded Padon, and Yoni Zohar.
\newblock Resources: A safe language abstraction for money, 2020.

\bibitem[Blackshear et~al.(2022)Blackshear, Mitchell, Nowacki, and
  Qadeer]{blackshear2022borrow}
Sam Blackshear, John Mitchell, Todd Nowacki, and Shaz Qadeer.
\newblock The move borrow checker, 2022.

\bibitem[Buterin(2016)]{re_dao}
Vitalik Buterin.
\newblock Critical update re {DAO}, 2016.
\newblock URL
  \url{https://ethereum.github.io/blog/2016/06/17/critical-update-re-dao-vulnerability}.

\bibitem[Consensys(2021)]{sc_best_practices}
Consensys.
\newblock Smart contract best practices, 2021.
\newblock URL
  \url{https://consensys.github.io/smart-contract-best-practices/known_attacks}.

\bibitem[Coughlin and Chang(2014)]{DBLP:conf/popl/CoughlinC14}
Devin Coughlin and Bor{-}Yuh~Evan Chang.
\newblock Fissile type analysis: modular checking of almost everywhere
  invariants.
\newblock In Suresh Jagannathan and Peter Sewell, editors, \emph{The 41st
  Annual {ACM} {SIGPLAN-SIGACT} Symposium on Principles of Programming
  Languages, {POPL} '14, San Diego, CA, USA, January 20-21, 2014}, pages
  73--86. {ACM}, 2014.
\newblock \doi{10.1145/2535838.2535855}.
\newblock URL \url{https://doi.org/10.1145/2535838.2535855}.

\bibitem[Cousot and Cousot(1977)]{10.1145/512950.512973}
Patrick Cousot and Radhia Cousot.
\newblock Abstract interpretation: A unified lattice model for static analysis
  of programs by construction or approximation of fixpoints.
\newblock In \emph{Proceedings of the 4th ACM SIGACT-SIGPLAN Symposium on
  Principles of Programming Languages}, POPL '77, page 238–252, New York, NY,
  USA, 1977. Association for Computing Machinery.
\newblock ISBN 9781450373500.
\newblock \doi{10.1145/512950.512973}.
\newblock URL \url{https://doi.org/10.1145/512950.512973}.

\bibitem[de~Moura and Bj{\o}rner(2008)]{DBLP:conf/tacas/MouraB08}
Leonardo~Mendon{\c{c}}a de~Moura and Nikolaj Bj{\o}rner.
\newblock {Z3:} an efficient {SMT} solver.
\newblock In C.~R. Ramakrishnan and Jakob Rehof, editors, \emph{Tools and
  Algorithms for the Construction and Analysis of Systems, 14th International
  Conference, {TACAS} 2008, Held as Part of the Joint European Conferences on
  Theory and Practice of Software, {ETAPS} 2008, Budapest, Hungary, March
  29-April 6, 2008. Proceedings}, volume 4963 of \emph{Lecture Notes in
  Computer Science}, pages 337--340. Springer, 2008.
\newblock \doi{10.1007/978-3-540-78800-3\_24}.
\newblock URL \url{https://doi.org/10.1007/978-3-540-78800-3\_24}.

\bibitem[Dill et~al.(2022)Dill, Grieskamp, Park, Qadeer, Xu, and
  Zhong]{prover-new}
David Dill, Wolfgang Grieskamp, Junkil Park, Shaz Qadeer, Meng Xu, and Emma
  Zhong.
\newblock Fast and reliable formal verification of smart contracts with the
  move prover.
\newblock In Dana Fisman and Grigore Rosu, editors, \emph{Tools and Algorithms
  for the Construction and Analysis of Systems}, pages 183--200, Cham, 2022.
  Springer International Publishing.
\newblock ISBN 978-3-030-99524-9.

\bibitem[Disselkoen et~al.(2019)Disselkoen, Renner, Watt, Garfinkel, Levy, and
  Stefan]{DBLP:conf/isca/DisselkoenRWGLS19}
Craig Disselkoen, John Renner, Conrad Watt, Tal Garfinkel, Amit Levy, and Deian
  Stefan.
\newblock Position paper: Progressive memory safety for webassembly.
\newblock In \emph{Proceedings of the 8th International Workshop on Hardware
  and Architectural Support for Security and Privacy, HASP@ISCA 2019, June 23,
  2019}, pages 4:1--4:8. {ACM}, 2019.
\newblock \doi{10.1145/3337167.3337171}.
\newblock URL \url{https://doi.org/10.1145/3337167.3337171}.

\bibitem[Foundation(2018)]{solidity}
Ethereum Foundation.
\newblock Solidity documentation, 2018.
\newblock URL \url{http://solidity.readthedocs.io}.

\bibitem[Fournet et~al.(2007)Fournet, Gordon, and Maffeis]{tydisa}
C{\'e}dric Fournet, Andrew~D. Gordon, and Sergio Maffeis.
\newblock A type discipline for authorization policies.
\newblock \emph{ACM Trans. Program. Lang. Syst.}, 29\penalty0 (5), August 2007.
\newblock ISSN 0164-0925.
\newblock \doi{10.1145/1275497.1275500}.
\newblock URL \url{http://doi.acm.org/10.1145/1275497.1275500}.

\bibitem[Girard(1987)]{linear_logic}
Jean{-}Yves Girard.
\newblock Linear logic.
\newblock \emph{Theor. Comput. Sci.}, 1987.

\bibitem[Google(2019)]{goog}
Google.
\newblock Sandboxed api, 2019.
\newblock https://github.com/google/sandboxed-api.

\bibitem[Gordon and Jeffrey(2003)]{autysec}
Andrew~D. Gordon and Alan Jeffrey.
\newblock Authenticity by typing for security protocols.
\newblock \emph{J. Comput. Secur.}, 11\penalty0 (4):\penalty0 451--519, July
  2003.
\newblock ISSN 0926-227X.
\newblock URL \url{http://dl.acm.org/citation.cfm?id=959088.959090}.

\bibitem[Grossman et~al.(2018)Grossman, Abraham, Golan{-}Gueta, Michalevsky,
  Rinetzky, Sagiv, and Zohar]{DBLP:journals/pacmpl/GrossmanAGMRSZ18}
Shelly Grossman, Ittai Abraham, Guy Golan{-}Gueta, Yan Michalevsky, Noam
  Rinetzky, Mooly Sagiv, and Yoni Zohar.
\newblock Online detection of effectively callback free objects with
  applications to smart contracts.
\newblock \emph{Proc. {ACM} Program. Lang.}, 2\penalty0 ({POPL}):\penalty0
  48:1--48:28, 2018.
\newblock \doi{10.1145/3158136}.
\newblock URL \url{https://doi.org/10.1145/3158136}.

\bibitem[Grumberg and Long(1994)]{rs-orig}
Orna Grumberg and David~E. Long.
\newblock Model checking and modular verification.
\newblock \emph{ACM Trans. Program. Lang. Syst.}, 16\penalty0 (3):\penalty0
  843–871, May 1994.
\newblock ISSN 0164-0925.
\newblock \doi{10.1145/177492.177725}.
\newblock URL \url{https://doi.org/10.1145/177492.177725}.

\bibitem[Haas et~al.(2017)Haas, Rossberg, Schuff, Titzer, Holman, Gohman,
  Wagner, Zakai, and Bastien]{DBLP:conf/pldi/HaasRSTHGWZB17}
Andreas Haas, Andreas Rossberg, Derek~L. Schuff, Ben~L. Titzer, Michael Holman,
  Dan Gohman, Luke Wagner, Alon Zakai, and J.~F. Bastien.
\newblock Bringing the web up to speed with webassembly.
\newblock In Albert Cohen and Martin~T. Vechev, editors, \emph{Proceedings of
  the 38th {ACM} {SIGPLAN} Conference on Programming Language Design and
  Implementation, {PLDI} 2017, Barcelona, Spain, June 18-23, 2017}, pages
  185--200. {ACM}, 2017.
\newblock \doi{10.1145/3062341.3062363}.
\newblock URL \url{https://doi.org/10.1145/3062341.3062363}.

\bibitem[Jones and Muchnick(1979)]{Jones-al:POPL79}
Neil~D. Jones and Steven~S. Muchnick.
\newblock Flow analysis and optimization of {LISP}-like structures.
\newblock In \emph{POPL}, 1979.

\bibitem[Jung et~al.(2018)Jung, Krebbers, Jourdan, Bizjak, Birkedal, and
  Dreyer]{iris}
Ralf Jung, Robbert Krebbers, Jacques-Henri Jourdan, Ales Bizjak, Lars Birkedal,
  and Derek Dreyer.
\newblock Iris from the ground up: A modular foundation for higher-order
  concurrent separation logic.
\newblock \emph{Journal of Functional Programming}, 28:\penalty0 e20, 2018.
\newblock \doi{10.1017/S0956796818000151}.

\bibitem[Larus and Hunt(2010)]{singularity}
James~R. Larus and Galen~C. Hunt.
\newblock The singularity system.
\newblock \emph{Commun. {ACM}}, 53\penalty0 (8):\penalty0 72--79, 2010.
\newblock \doi{10.1145/1787234.1787253}.
\newblock URL \url{https://doi.org/10.1145/1787234.1787253}.

\bibitem[Lindholm and Yellin(1997)]{jvm}
Tim Lindholm and Frank Yellin.
\newblock \emph{The {J}ava Virtual Machine Specification}.
\newblock Addison-Wesley, 1997.

\bibitem[Maffeis et~al.(2008)Maffeis, Abadi, Fournet, and Gordon]{cca}
Sergio Maffeis, Mart{\'i}n Abadi, C{\'e}dric Fournet, and Andrew~D. Gordon.
\newblock \emph{Code-Carrying Authorization}, pages 563--579.
\newblock Springer Berlin Heidelberg, Berlin, Heidelberg, 2008.
\newblock ISBN 978-3-540-88313-5.
\newblock \doi{10.1007/978-3-540-88313-5_36}.
\newblock URL \url{http://dx.doi.org/10.1007/978-3-540-88313-5_36}.

\bibitem[Matsakis and Klock(2014)]{rust}
Nicholas~D. Matsakis and Felix~S. Klock, II.
\newblock The rust language.
\newblock \emph{Ada Lett.}, 34\penalty0 (3):\penalty0 103--104, October 2014.
\newblock ISSN 1094-3641.
\newblock \doi{10.1145/2692956.2663188}.
\newblock URL \url{http://doi.acm.org/10.1145/2692956.2663188}.

\bibitem[Meijer et~al.(2000)Meijer, Wa, and Gough]{clr}
Erik Meijer, Redmond Wa, and John Gough.
\newblock Technical overview of the common language runtime, 2000.

\bibitem[Mettler et~al.(2010)Mettler, Wagner, and
  Close]{DBLP:conf/ndss/MettlerWC10}
Adrian Mettler, David~A. Wagner, and Tyler Close.
\newblock Joe-e: {A} security-oriented subset of java.
\newblock In \emph{Proceedings of the Network and Distributed System Security
  Symposium, {NDSS} 2010, San Diego, California, USA, 28th February - 3rd March
  2010}. The Internet Society, 2010.
\newblock URL
  \url{https://www.ndss-symposium.org/ndss2010/joe-e-security-oriented-subset-java}.

\bibitem[Meyer et~al.(2021)Meyer, Arkadova, Kogtenkov, and
  Naumchev]{DBLP:journals/corr/abs-2109-06557}
Bertrand Meyer, Alisa Arkadova, Alexander Kogtenkov, and Alexandr Naumchev.
\newblock The concept of class invariant in object-oriented programming.
\newblock \emph{CoRR}, abs/2109.06557, 2021.
\newblock URL \url{https://arxiv.org/abs/2109.06557}.

\bibitem[Miller et~al.(2003)Miller, Yee, and Shapiro]{Miller03capabilitymyths}
Mark Miller, Ka-Ping Yee, and Jonathan Shapiro.
\newblock Capability myths demolished.
\newblock Technical report, 2003.

\bibitem[Morrisett et~al.(2003)Morrisett, Crary, Glew, and
  Walker]{typed_assembly}
J.~Gregory Morrisett, Karl Crary, Neal Glew, and David Walker.
\newblock Stack-based typed assembly language.
\newblock \emph{J. Funct. Program.}, 13\penalty0 (5):\penalty0 957--959, 2003.
\newblock \doi{10.1017/S0956796802004446}.
\newblock URL \url{https://doi.org/10.1017/S0956796802004446}.

\bibitem[Mozilla(2019)]{moz}
Mozilla.
\newblock Script security, 2019.
\newblock Technical Report.
  https://developer.mozilla.org/en-US/docs/Mozilla/Gecko/Script\_security.

\bibitem[Patrignani et~al.(2011)Patrignani, Clarke, and Sangiorgi]{ot4jc}
Marco Patrignani, Dave Clarke, and Davide Sangiorgi.
\newblock Ownership {T}ypes for the {J}oin {C}alculus.
\newblock In \emph{FMOODS/FORTE 2011}, volume 6722 of \emph{LNCS}, pages
  289--303, 2011.

\bibitem[Sammler et~al.(2020)Sammler, Garg, Dreyer, and Litak]{dg-rs}
Michael Sammler, Deepak Garg, Derek Dreyer, and Tadeusz Litak.
\newblock The high-level benefits of low-level sandboxing.
\newblock \emph{{PACMPL}}, 4\penalty0 ({POPL}):\penalty0 32:1--32:32, 2020.
\newblock \doi{10.1145/3371100}.
\newblock URL \url{https://doi.org/10.1145/3371100}.

\bibitem[Sangiorgi and Walker(2001)]{bookpi}
Davide Sangiorgi and David Walker.
\newblock \emph{PI-Calculus: A Theory of Mobile Processes}.
\newblock Cambridge University Press, USA, 2001.
\newblock ISBN 0521781779.

\bibitem[Sergey et~al.(2019)Sergey, Nagaraj, Johannsen, Kumar, Trunov, and
  Hao]{DBLP:journals/pacmpl/SergeyNJ0TH19}
Ilya Sergey, Vaivaswatha Nagaraj, Jacob Johannsen, Amrit Kumar, Anton Trunov,
  and Ken Chan~Guan Hao.
\newblock Safer smart contract programming with scilla.
\newblock \emph{Proc. {ACM} Program. Lang.}, 3\penalty0 ({OOPSLA}):\penalty0
  185:1--185:30, 2019.
\newblock \doi{10.1145/3360611}.
\newblock URL \url{https://doi.org/10.1145/3360611}.

\bibitem[Swasey et~al.(2017)Swasey, Garg, and Dreyer.]{davidcaps}
David Swasey, Deepak Garg, and Derek Dreyer.
\newblock Robust and compositional verification of object capability patterns.
\newblock In \emph{Proceedings of the 2017 {ACM} {SIGPLAN} International
  Conference on Object-Oriented Programming, Systems, Languages, and
  Applications, {OOPSLA} 2017. October 22 - 27, 2017}, 2017.

\bibitem[Wadler(1990)]{Wadler90lineartypes}
Philip Wadler.
\newblock Linear types can change the world!
\newblock In \emph{PROGRAMMING CONCEPTS AND METHODS}, 1990.

\bibitem[Watson et~al.(2015)Watson, Woodruff, Neumann, Moore, Anderson,
  Chisnall, Dave, Davis, Gudka, Laurie, Murdoch, Norton, Roe, Son, and
  Vadera]{DBLP:conf/sp/WatsonWNMACDDGL15}
Robert N.~M. Watson, Jonathan Woodruff, Peter~G. Neumann, Simon~W. Moore,
  Jonathan Anderson, David Chisnall, Nirav~H. Dave, Brooks Davis, Khilan Gudka,
  Ben Laurie, Steven~J. Murdoch, Robert~M. Norton, Michael Roe, Stacey~D. Son,
  and Munraj Vadera.
\newblock {CHERI:} {A} hybrid capability-system architecture for scalable
  software compartmentalization.
\newblock In \emph{2015 {IEEE} Symposium on Security and Privacy, {SP} 2015,
  San Jose, CA, USA, May 17-21, 2015}, pages 20--37. {IEEE} Computer Society,
  2015.
\newblock \doi{10.1109/SP.2015.9}.
\newblock URL \url{https://doi.org/10.1109/SP.2015.9}.

\bibitem[Wood(2014)]{evm}
Gavin Wood.
\newblock Ethereum: A secure decentralised generalised transaction ledger.
\newblock 2014.
\newblock URL \url{https://ethereum.github.io/yellowpaper/paper.pdf}.

\bibitem[Zhong et~al.(2020)Zhong, Cheang, Qadeer, Grieskamp, Blackshear, Park,
  Zohar, Barrett, and Dill]{DBLP:conf/cav/ZhongCQGBPZBD20}
Jingyi~Emma Zhong, Kevin Cheang, Shaz Qadeer, Wolfgang Grieskamp, Sam
  Blackshear, Junkil Park, Yoni Zohar, Clark~W. Barrett, and David~L. Dill.
\newblock The move prover.
\newblock In Shuvendu~K. Lahiri and Chao Wang, editors, \emph{Computer Aided
  Verification - 32nd International Conference, {CAV} 2020, Los Angeles, CA,
  USA, July 21-24, 2020, Proceedings, Part {I}}, volume 12224 of \emph{Lecture
  Notes in Computer Science}, pages 137--150. Springer, 2020.
\newblock \doi{10.1007/978-3-030-53288-8\_7}.
\newblock URL \url{https://doi.org/10.1007/978-3-030-53288-8\_7}.

\end{thebibliography}

\end{document}